\documentclass[11pt]{article}

\usepackage{booktabs} 
\usepackage[font=footnotesize,labelfont=bf]{caption}
\usepackage{subcaption}
\usepackage{geometry}
\usepackage{graphicx}
\usepackage{amsmath}
\usepackage[svgnames]{xcolor}
\usepackage{amssymb}
\usepackage[all]{xy}
\usepackage{csquotes}
\usepackage{tikz}
\usetikzlibrary{patterns}
\usetikzlibrary{matrix}
\usetikzlibrary{backgrounds} 
\usepackage{float}

\usepackage{amsthm}

\newtheorem{theorem}{Theorem}[section]
\newtheorem{corollary}[theorem]{Corollary}
\newtheorem{lemma}[theorem]{Lemma}
\newtheorem{proposition}[theorem]{Proposition}
\newtheorem{claim}[theorem]{Claim}
\newtheorem{definition}[theorem]{Definition}

\newcommand{\eps}{\varepsilon} 
\newcommand{\cupdot}{\mathbin{\mathaccent\cdot\cup}}

\newcommand{\shortversion}[1]{}

\newcommand{\cA}{{\cal A}}
\newcommand{\cB}{{\cal B}}

\newcommand{\cM}{{\cal M}}

\newcommand{\cS}{{\cal S}}

\newcommand{\E}{\mbox{\bf E}}
\newcommand{\RR}{\mathbb R}

\DeclareMathOperator*{\argmax}{argmax}

\usepackage[colorinlistoftodos,textsize=tiny,textwidth=2cm,color=red!25!white,obeyFinal]{todonotes}

\begin{document}

\title{On the Hardness of Dominant Strategy Mechanism Design}

\author{Shahar Dobzinski \thanks{Weizmann Institute of Science. Email: shahar.dobzinski@weizmann.ac.il.
Work supported by BSF grant 2016192 and ISF grant 2185/19.	
}
\and Shiri Ron\thanks{Weizmann Institute of Science. Email:  shiriron@weizmann.ac.il.} 
 \and Jan Vondrak \thanks{Stanford University. Email: jvondrak@stanford.edu. Work supported by BSF-NSF grant (BSF number: 2021655, NSF number: 2127781).}}

%


\maketitle
\begin{abstract}
We study the communication complexity of dominant strategy implementations of combinatorial auctions. We start with two domains that are generally considered ``easy'': multi-unit auctions with decreasing marginal values and combinatorial auctions with gross substitutes valuations. For both domains we have fast algorithms that find the welfare-maximizing allocation with communication complexity that is poly-logarithmic in the input size. This immediately implies that welfare maximization can be achieved in ex-post equilibrium with no significant communication cost, by using VCG payments. In contrast, we show that in both domains the communication complexity of any dominant strategy implementation that achieves the optimal welfare is polynomial in the input size.

We then move on to studying the approximation ratios achievable by dominant strategy mechanisms. For multi-unit auctions with decreasing marginal values, we provide a dominant-strategy communication FPTAS. For combinatorial auctions with general valuations, we show that there is no dominant strategy mechanism that achieves an approximation ratio better than $m^{1-\eps}$ that uses $poly(m,n)$ bits of communication, where $m$ is the number of items and $n$ is the number of bidders. In contrast, a \emph{randomized} dominant strategy mechanism that achieves an $O(\sqrt m)$ approximation with $poly(m,n)$ communication is known. This proves the first gap between computationally efficient deterministic dominant strategy mechanisms and randomized ones.

En route, we answer an open question on the communication cost of implementing dominant strategy mechanisms for more than two players, and also solve some open problems in the area of simultaneous combinatorial auctions.
\end{abstract}


\section{Introduction}

In his seminal 1961 paper \cite{Vic61}, Vickrey considers single item auctions: there is one item and $n$ bidders, the value of each bidder $i$ for the item is $v_i$. Vickrey defines the second-price auction: the highest bidder wins the item and pays the second highest bid. It is shown that in a second-price auction, bidding truthfully is a dominant strategy for each bidder. However, observe that our definition of a second-price auction was a bit careless. Bidding truthfully is indeed a dominant strategy when the second price auction is held by asking the bidders to simultaneously submit their bids, or when implemented iteratively, by conducting a (continuous) ascending auction. However, this is not always the case. Consider a ``serial'' implementation of a second price auction in which the bids of players $1,\ldots, i-1$ are publicly revealed before player $i$ makes a bid. Truth-telling is no longer a dominant strategy for, e.g., player $1$: if the strategy of all other players is ``bid $0$ unless player $1$ bids $10$, in which case bid $9$'', then player $1$ is better off bidding $11$ when his true value is $v_1=10$.

\subsubsection*{The Setting}

In this paper we analyze the hardness of dominant strategy implementations in combinatorial auctions. Recall that in a combinatorial auction there is a set $M$ of heterogeneous items $(|M|=m)$ and a set $N$ of bidders ($|N|=n$). The private information of each bidder $i$ is his value for every subset of the items: $v_i:2^M\rightarrow \mathbb R$. The standard assumptions are that the valuations are non-decreasing (for each $S\subseteq T$, $v(T)\geq v(S)$) and normalized ($v_i(\emptyset)=0$), though we will sometimes impose additional restrictions on the valuations. Our goal is to find an allocation of the items $(S_1,\ldots, S_n)$ that maximizes the social welfare: $\Sigma_iv_i(S_i)$.

We use communication protocols to model mechanisms. Specifically, we work in the blackboard model, so all messages sent are observable by all players. The input of each player $i$ is his valuation $v_i$. As usual, the communication protocol is represented by a tree that dictates which players (simultaneously) speak at each node, and the identity of the next node given the messages. The leaves of the protocol specify the outcome: the allocation and payments. We assume that each player is interested in maximizing his profit: the value of his assignment minus his payment. 
We assume that all mechanisms are \emph{normalized}, i.e. the price of the empty bundle is always zero. 


A strategy $\mathcal{S}_i$ of player $i$ dictates (given the valuation $v_i$) which messages player $i$ sends at each node. We say that $\mathcal{S}_i$ is \emph{dominant} for player $i$ if for every set of possible strategies $\mathcal{S}_{-i}'$ of the other players, every valuation profile $(v_1,\ldots,v_n)$ and every strategy $\mathcal{S}'_i$ of player $i$ it holds that: 
$$
v_i(f_i(\mathcal{S}_i(v_i),\mathcal{S}_{-i}'(v_{-i}))) - p_i(\mathcal{S}_i(v_i),\mathcal{S}_{-i}'(v_{-i})) \geq v_i(f_i(\mathcal{S}_i'(v_i),\mathcal{S}_{-i}'(v_{-i}))) - p_i(\mathcal{S}'_i(v_i),\mathcal{S}_{-i}'(v_{-i}))
$$
where  $f_i(\mathcal{S}_i(v_i),\mathcal{S}_{-i}'(v_{-i}))$  and  $p_i(\mathcal{S}_i(v_i),\mathcal{S}_{-i}'(v_{-i}))$ specify the allocation and payment of player $i$, respectively, given that player $i$ follows the  actions specified by  $\mathcal{S}_i(v_i)$ and the other players follow  the actions specified in the vector  $\mathcal{S}_{-i}'(v_{-i})=(\mathcal{S}_1'(v_1),\ldots, \mathcal{S}_{i-1}'(v_{i-1}), \mathcal S_{i+1}'(v_{i+1}), \ldots, \mathcal{S}_n'(v_n))$. Consider a mechanism in which each player $i$ has a dominant strategy $\mathcal{S}_i$. Let $V_i$  be the set of possible valuations of a player and let $\mathcal A$ is the set of all possible allocations. Let $f:V_1\times\cdots\times V_n\rightarrow \mathcal A$ be the social choice function defined by $f(v_1,\ldots, v_n)_i =f_i(\mathcal{S}_i(v_i),\mathcal{S}_{-i}(v_{-i}))$. In this case we say that the mechanism implements $f$ in dominant strategies. 

The importance of the specifics of the implementation and how they affect the solution concept are well known. The notion of ex-post equilibrium, defined by Cremer and McLean \cite{CM85}, attempts in a sense to get around this by ignoring the specifics of the implementation. A function $f:V_1\times\cdots\times V_n\rightarrow \mathcal A$ is implementable in \emph{ex-post equilibrium} if there are functions $p_1,\ldots, p_n:V_1\times\cdots\times V_n\rightarrow \mathbb R$ such that for every player $i$, valuations $v_i$ and $v'_i$ of player $i$, and valuations $v_{-i}$ of the other players:
$$
v_i(f_i(v_i,v_{-i})) - p_i(v_i,v_{-i}) \geq v_i(f_i(v'_i,v_{-i})) - p_i(v'_i,v_{-i})
$$
Roughly speaking, in an ex-post equilibrium none of the players regrets playing according to his true value, if the other players are playing according to their true values as well. This rules out ``unreasonable'' strategies like in the serial second price auction described above. An alternative description would be that an ex-post incentive compatible implementation of a function $f$ is a communication protocol that computes a function $f$  and the associated payments, where $f$ can be implemented in dominant strategies. However, in this protocol the players might not have dominant strategies. 
Clearly, every dominant strategy implementation is also an ex-post implementation. The other direction is not true, as the serial implementation of a second price auction demonstrates.\footnote{Admittedly, in some algorithmic mechanism design papers that analyze the communication complexity of incentive compatible mechanisms the distinction between the two notions is less explicit than it should be.}
Thus the communication cost of ex-post implementations is potentially much smaller than the cost of dominant strategy implementations.

The goal of this paper is to determine whether the communication cost of dominant strategy implementations  is significantly larger than the cost of ex-post implementations. Intuitively, one might suspect that dominant strategy mechanisms require significantly more communication than ex-post mechanisms. However, prior research can offer only mixed evidence to support this. First, the revelation principle implies, in particular, that every ex-post implementable function $f$ is also dominant strategy implementable (the implementation is simple: each player simultaneously reveals his valuation, and the outcome is determined accordingly). However, as was already observed by Conitzer and Sandholm \cite{CS04}, this naive implementation method might easily result in an exponential blow-up in the communication complexity. Yet, this method works well in domains in which the private information of the players can be succinctly described, e.g., single-parameter domains.

Our interest is in the more complicated multi-parameter domains. Almost all known \emph{deterministic} incentive compatible mechanisms\footnote{The only exception is \cite{BGN03} which is the only known example of a mechanism that is not maximal in range and achieves the state of the art results in a well-studied domain.} \cite{HKMT04, DNS05, DN07b,PSS08} are maximal-in-range mechanisms. Moreover, in each of them each bidder sends in the first round his value of all the (polynomially many) bundles he might win. Hence these mechanisms are dominant strategy.

The evidence that implementation in ex-post equilibrium does not buy much computational power comparing to dominant strategy implementation is more than anecdotal. In \cite{D16b} it is shown -- perhaps counter-intuitively -- that \emph{every} two player ex-post mechanism for combinatorial auctions in a rich enough domain (in particular, one that includes all XOS valuations) can be implemented in a dominant strategy equilibrium with only a polynomial blow-up in the communication complexity\footnote{In \cite{D16b} an analogous result is proved also for domains that include all submodular valuations, under certain constraints.}.

In contrast, there is evidence supporting the idea that ex-post mechanism design is significantly less costly, communication-wise. Very recently, \cite{RSTWZ21} presented a carefully crafted setting in which there is a mechanism that implements a welfare maximizer in an ex-post equilibrium with $c$ bits, but every dominant strategy implementation requires $exp(c)$ bits.

\subsubsection*{Our Results}

We begin our explorations by considering the result of \cite{D16b} discussed above, that shows that every function $f$ for two players in a ``rich enough'' auction domain that can be implemented in an ex-post equilibrium can also be implemented in dominant strategies with only a polynomial blow-up in the communication. The paper \cite{D16b} leaves open the question of whether this result holds also for mechanisms with more than two players. In Section \ref{separation-sec}, we answer this question in the negative by showing that the equivalence in implementations is unique for two player mechanisms: there is a three-player social choice function for general valuations that has an ex-post implementation that uses only $c$ bits, but $exp(c)$ bits are required for any dominant strategy implementation.
Next, in Section \ref{sec-exact-welfare} we consider two auction domains that are largely considered ''easy'' in the algorithmic mechanism design literature: multi-unit auctions with decreasing marginal values and combinatorial auctions with gross substitutes valuations (see, e.g., the surveys \cite{N15, BN07}). 
In multi-unit auctions with decreasing marginal values, the welfare maximizing solution can be found with $poly(n,\log m)$ communication, and in combinatorial auctions with gross substitutes valuations, the welfare maximizing solution can be found with $poly(n,m)$ communication \cite{NS06}. We thus get that in both settings the function that outputs the welfare maximizing allocation is implementable with low communication (since VCG payments can be computed with only a polynomial blow up in the communication). However, these results hold only in an ex-post equilibrium.
In a sharp contrast, we show that an exponential blow up is required for dominant strategy implementations (again, in the blackboard model):

\vspace{0.1in}\noindent \textbf{Theorem:} 
\begin{enumerate}
\item The communication complexity of every normalized mechanism that finds a welfare-maximizing allocation for two players in dominant strategies in multi unit auctions when the valuations exhibit decreasing marginal values is $\Omega(m\cdot \log m)$. In contrast, there is a mechanism with communication complexity $poly(\log m)$ that finds the welfare-maximizing allocation in an ex-post equilibrium. 
\item The communication complexity of every normalized mechanism that finds a welfare-maximizing allocation even for two players in dominant strategies in combinatorial auctions with gross substitutes valuations is $exp(m)$. In contrast, there is a mechanism with communication complexity $poly(m)$ that finds the welfare-maximizing allocation in an ex-post equilibrium.
\end{enumerate} 

\vspace{0.1in}\noindent We note that these results echo the recent result of \cite{RSTWZ21} which was the first to show that the communication cost of dominant strategy implementations of welfare maximizers might be exponential comparing to the communication cost of ex-post implementations but in an artificial domain. In contrast, our results prove an exponential blow-up of welfare maximizers in well-studied auction domains.

Perhaps in contrast to the common perception, the theorem demonstrates that these domains are not ``easy'' from the point of view of dominant strategy mechanism design. This immediately raises the question of whether we can have good approximations to the social welfare by low-communication dominant strategy mechanisms. For multi-unit auctions, we answer this question in the affirmative:

\vspace{0.1in}\noindent \textbf{Theorem:} Let $\eps>0$. There is a dominant strategy $(1+\eps)$-approximation mechanism for multi-unit auctions with valuations that exhibit decreasing marginal values that makes $poly(n,\log m,\frac 1 \eps)$ value queries.

\vspace{0.1in}\noindent Whether one can get good approximation ratios for combinatorial auctions with gross substitutes valuations remains an open question. The maximal-in-range mechanism of \cite{DNS05} achieves an approximation ratio of $O(\sqrt m)$ in dominant strategies for the much larger class of subadditive valuations. However, we do not even know whether dominant strategy \emph{maximal-in-range} mechanisms with polynomial communication can achieve a better approximation ratio (known impossibilities for maximal-in-range mechanisms \cite{DN07a,DSS15} hold for ex-post mechanisms but not for gross-substitutes valuations). 

We then move on to analyzing the approximation ratios achievable by dominant strategy mechanisms in the standard domain of combinatorial auctions with general (monotone) valuations. From a pure optimization point of view, there is an $O(\sqrt m)$ approximation algorithm that is not incentive compatible and this is the best achievable with polynomial communication \cite{N02,LOS:J}. Whether this is achievable with a deterministic ex-post incentive compatible mechanism remains a major open question, but we are able to answer this question in the negative for dominant strategy mechanisms (Section \ref{sec-general-outline}):

\vspace{0.1in}\noindent \textbf{Theorem:} Fix $\eps>0$. The communication complexity of a mechanism that provides an $m^{1-\eps}$ approximation for combinatorial auctions with general valuations in dominant strategies is $exp(m)$. 

\vspace{0.1in}\noindent The best currently known mechanism (dominant strategy or ex-post incentive compatible) is the simultaneous maximal-in-range algorithm of \cite{HKMT04} that guarantees an approximation ratio of $O(\frac m {\sqrt {\log m}})$. To put the theorem in context, so far, following a long line of research, the only separation between the approximation ratios achievable by ex-post mechanisms and non incentive compatible algorithms for combinatorial auctions that use polynomial communication was achieved in \cite{AKSW20}. This separation applies to two-player combinatorial auctions with XOS valuations, and relies on the taxation framework \cite{D16b}. Recall that \cite{D16b} shows the equivalence of ex-post and dominant strategy implementations for two player settings, thus the result of \cite{AKSW20} is also the first to separate dominant strategy mechanisms for combinatorial auctions and their non-truthful counterparts. 

However, a proof for our theorem requires more players, since for two players a second-price auction on the bundle of all items provides an approximation ratio of $2$.\footnote{The taxation framework \cite{D16b} offers also a different path to proving bounds for more than $2$ players by providing lower bounds on the taxation complexity, but this path was not applied successfully so far.} Thus, new tools are required to prove a bound that is worse than $2$.
The proof consists of two mains steps. First, we prove in Section \ref{sec-simultanous-general} that:

\vspace{0.1in}\noindent \textbf{Theorem:} Fix $\eps>0$. The communication complexity of a simultaneous algorithm that provides an $m^{1-\eps}$ approximation for combinatorial auctions with general valuations is $exp(m)$. 

\vspace{0.1in}\noindent 
Simultaneous combinatorial auctions were introduced by \cite{DNO14}: in these (not necessarily incentive compatible) algorithms, all players simultaneously send a message of length $poly(n,m)$ and the allocation is determined based only on these messages. Previous work (e.g., \cite{BO17, A20, ANRW15, BMW18}) considered simultaneous combinatorial auctions with restricted classes of valuations, e.g., subadditive valuations. 

In the second step, we leverage 
the hardness result to dominant strategy mechanisms by showing that the existence of a deterministic dominant strategy mechanism with approximation ratio $c$ implies a simultaneous algorithm with approximation ratio (close to) $c$.

We note that for general valuations, there exists a \emph{randomized} dominant strategy mechanism that achieves an approximation ratio of $O(\sqrt m)$ \cite{DNS06}. The mechanism is a probability distribution over dominant strategy mechanisms. Hence, we also obtain a separation of the approximation ratio possible by polynomial communication randomized dominant strategy mechanisms and deterministic dominant strategy mechanisms. An analogous separation for \emph{ex-post} mechanisms is not known.

In addition, we study the structure of dominant strategy mechanisms for general valuations (Section \ref{sec-characterization}). 
Roughly speaking, we prove that such mechanisms must be {\em semi-simultaneous} in the sense that for each player $i$, the mechanism ``commits'' on player $i$'s allocation and payment right after player $i$'s first message, unless the player sends a special message which ``postpones'' determining the allocation and payment to the next rounds. 
One example of a semi-simultaneous mechanism that is not simultaneous is an ascending auction on a bundle of some of the items.

\subsubsection*{Open Questions and Future Directions}

We conclude with some open questions. We showed that dominant strategy mechanisms cannot exactly maximize the welfare in polynomial communication in combinatorial auctions with gross substitutes valuations. As was already mentioned, it is an open question to determine the approximation ratio achievable for this class or for other classes of valuations that were extensively studied in the literature, such as subadditive, XOS, and submodular. 

We do provide some evidence that good dominant strategy mechanisms do not exist. Observe that all useful constructions of deterministic dominant strategy mechanisms that we know are based on simultaneous algorithms. In Section \ref{sec:simul} we prove that:

\vspace{0.1in}\noindent \textbf{Theorem:} Fix $\eps>0$. The communication complexity of a simultaneous algorithm that provides an $m^{\frac 1 {16}}$ approximation for combinatorial auctions with gross substitutes valuations is $exp(m)$. 

\vspace{0.1in}\noindent This answers an open question of \cite{DNO14}. Before our work, it was not even known if there is a simultaneous algorithm for combinatorial auctions with submodular valuations that achieves a constant approximation ratio.

Another exciting direction is proving impossibilities for randomized mechanisms. A recent line of work provides sub-logarithmic approximation ratios for various classes of valuations \cite{D16a, AS19, AKS21}. All these mechanisms are a probability distribution over dominant strategy mechanisms.\footnote{Only \cite{D16a} claim explicitly that the mechanism is dominant strategy and not just ex-post incentive compatible, but this is likely to be the case also for the other papers as they follow the basic structure that was introduced in \cite{D16a}.} Are randomized ex-post mechanisms more powerful than dominant strategy mechanisms?\footnote{In contrast, many of the truthful-in-expectation mechanisms in the literature are based on solving an LP and are not dominant strategies \cite{LS05, DRY11}, though some dominant strategy truthful-in-expectation mechanisms do provide an optimal approximation ratio \cite{DD09}. Analyzing the power of dominant strategy truthful-in-expectation mechanisms is also a fascinating avenue for future research.}

We end by noting that our mechanisms work in the blackboard model and all messages sent are observable by all players. A more relaxed model would allow private channels between the players and the center. This assumes that the players trust the center not to leak their messages and the private communication channel itself is not leaky. We do not know how to take advantage of this relaxed model, except for the case of combinatorial auctions with $k$ copies from each good, where the mechanism of \cite{BGN03} (the usual outlier) 
cannot be implemented in dominant strategies but can be implemented in the relaxed model. We leave studying this model to future research.

\section{Formalities and Basic Observations}\label{preliminaries}

In this section we discuss some basic properties of dominant strategy mechanisms. These properties hold for every possible domain, not only for combinatorial auctions. Thus, in this section $\mathcal A$ is the set of alternatives (which are not necessarily allocations) and the valuation of each player is $v_i:\mathcal A\to \mathbb R$.   

Here and subsequently, when we talk about a fixed mechanism $\mathcal{M}$ together with its dominant strategies $\mathcal{S}_1,\ldots,\mathcal{S}_n$ we will slightly abuse notation: We say that player $i$ with valuation $v_i$ sends a message $z$ at vertex $r$ instead of saying that the dominant strategy of player $i$ with valuation $v_i$ is to send message $z$ in at vertex $r$. We also say that valuations $v_1\ldots,v_n$ reach a leaf of a protocol, instead of saying that the strategy profile $(\mathcal{S}_1(v_1),\ldots,\mathcal{S}_n(v_n))$ leads to it.  

\subsection{Minimal Dominant Strategy Mechanisms}
In this section, we show that all dominant strategy mechanisms can be simplified without harming their dominant strategy equilibria and without any communication burden.  Since our main interest in this paper is in impossibility results, it implies that  we can analyze the power of \textquote{minimal} dominant strategy mechanisms without loss of generality. Formally:
\begin{definition}
	We say that a mechanism $\mathcal{M}$ is \emph{minimal} with respect to the strategies $(\mathcal{S}_1,\ldots,\mathcal{S}_n)$ and the valuations $V=V_1\times \cdots \times V_n$ if it satisfies the following properties:
	\begin{enumerate}
		\item There are no useless messages in the protocol, i.e. if some player $i$ can send some message in some particular vertex, we assume that it is a dominant strategy for some type $v_i$ to send this message. It immediately implies that for every leaf in the protocol there exist valuations $(v_1,\ldots,v_n)$ such that the strategies $(\mathcal{S}_1(v_1),\ldots,\mathcal{S}_n(v_n)$ reach this leaf. \label{obs1}
		\item 	Every node in the protocol tree (that is not a leaf) satisfies that there  is  at least one player $i$ that has two valuations $v_i,v_i'\in V_i$  such that the strategies  $\mathcal{S}_i(v_i)$ and $\mathcal{S}_i(v_i')$ dictate sending different messages in it.  \label{obs2}
	\end{enumerate} 
\end{definition}
\begin{lemma}\label{minimal-lemma}
	Let $\mathcal{M}$ be mechanism and strategies $(\mathcal{S}_1,\ldots,\mathcal{S}_n)$ that realize a  social choice function $f:V\to\mathcal{A}$ with payments $P_1,\dots,P_n:V\to \mathbb{R}^n$ in dominant strategies with communication complexity of $c$ bits.
	Then, there exists a \emph{minimal} mechanism $\mathcal{M}'$ and strategies $(\mathcal{S}_1',\ldots,\mathcal{S}_n')$ that realize $f$ with the payments schemes $P_1,\ldots,P_n$ in dominant strategies with at most $c$ bits. 
\end{lemma}
\begin{proof}
	Given a mechanism, we can assume that it has no useless messages because  otherwise we can simplify the protocol by not letting player $i$ send this message. Note that removing actions that are dominant strategy for none of the players does not make dominant strategies not dominant. 
	
	Similarly, if the second condition does not hold for some vertex,  then due to the fact that there are no useless messages, it has only one child. Then, we can delete this vertex and replace it with its child. We continue with this iterative trimming until we reach a vertex that has a player $i$ with a \textquote{meaningful} message. 
	If no such vertex is found until we reach a leaf, we replace the original vertex by this leaf.  
\end{proof}

\subsection{Induced Trees of Mechanisms}
We now introduce the notion of induced trees and prove a simple property of them. 
Consider some vertex $u$ in a minimal dominant strategy mechanism. Let $Z_{j,u}$ denote the set of possible messages that player $j$ can send at node $u$ (assume that $Z_{j,u}=\varnothing$ if player $i$ does not send any message at node $u$). Fix some player $i$ with $Z_{i,u}\neq \varnothing$ and some message profile for the other players $z^u_{-i}=(z_1, \ldots,z_{i-1},z_{i+1},\ldots, z_{n})$ where each $z_j\in  Z_{j,u}$. The \emph{induced tree of player $i$ at vertex $u$ given $z_{-i}^u$} is the tree that is rooted by $u$ and contains all subtrees that are connected to $u$ by an edge $(z_i,z^u_{-i})$ for every possible $z_i\in Z_{i,u}$. I.e., we fix the messages of all other players except player $i$ and think about each message $z_i$ as leading to the subtree that the set of messages $(z_i,z_{-i}^u)$ leads to. For an illustration,  see Figure \ref{my-third-tikz}.

\begin{figure}[H]
	\centering 
	\setlength{\belowcaptionskip}{-1pt}
	\caption{Illustration of the tree rooted at $u$ of a two-player protocol and one of its induced trees.
		The vertex $u$ satisfies that $Z_{1,u}=\{z_1,z_1'\}$ and $Z_{2,u}=\{z_2,z_2'\}$,
		i.e. each player has two possible messages.  The leaf $l$ that is labeled with $(A,p_1,p_2)$ satisfies that the mechanism outputs alternative $A$, player pays $p_1$ and player $2$ pays $p_2$. The same holds for the leaf $l'$ with respect to its outcome $(A',p_1',p_2')$. The induced tree at Figure \ref{subfig-induced} describes how the protocol looks like from the point of view of player $1$ when player $2$ sends the message $z_2$. 
	}
	
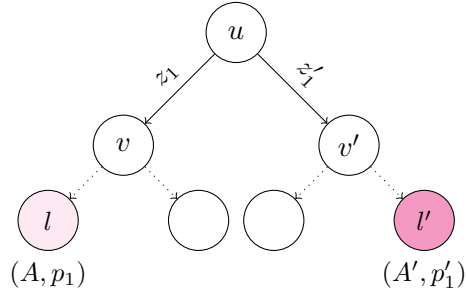
\label{my-third-tikz}
	\begin{subfigure}[b]{1.0 \linewidth}
		\centering
		\setlength{\belowcaptionskip}{-1pt}
		\caption{An illustration of the full tree protocol at vertex $u$.}
		\label{subfig-full-protocol}
		\begin{tikzpicture}
			\node[shape=circle,draw=black,minimum size=0.75cm] (r) at (1.5,1.5) {\small$u$};
			\node[shape=circle,draw=black,minimum size=0.75cm] (v) at (-3,0) {\small$v$};
			\node[shape=circle,fill=magenta!10,draw=black,minimum size=0.75cm,label=below:{\small $(A,p_1,p_2)$}] (l) at (-4,-1) {\small$l$};
			\node[shape=circle,draw=black,minimum size=0.75cm] (v') at (6,0) {}; 
			\node[shape=circle,draw=black,minimum size=0.75cm] (l') at (7,-1) {};
			
			\node[shape=circle,draw=black,minimum size=0.75cm] (n') at (5,-1) {}; 
			\node[shape=circle,draw=black,minimum size=0.75cm] (n) at (-2,-1) {}; 
			
			\node[shape=circle,draw=black,minimum size=0.75cm] (v-newl) at (0,0) {\small$v'$};
			\node[shape=circle,draw=black,minimum size=0.75cm] (v-newl-kid-left) at (-1,-1) {};
			\node[shape=circle,draw=black,,fill=magenta!50,minimum size=0.75cm,label=below:{\small $(A',p_1',p_2')$}] (v-newl-kid-right) at (1,-1) {\small$l'$};
			
			\node[shape=circle,draw=black,minimum size=0.75cm] (v-newr) at (3,0) {};
			\node[shape=circle,draw=black,minimum size=0.75cm] (v-newr-kid-left) at (4,-1) {};
			\node[shape=circle,draw=black,minimum size=0.75cm] (v-newr-kid-right) at (2,-1) {};
			
			%
			%

			\draw [->] (r) edge  node[sloped, above] {\footnotesize $z_1,z_2$} (v);
			\draw [->] (r) edge  node[sloped, above] {\footnotesize $z_1',z_2'$} (v');
			\draw [->] (r) edge  node[sloped, below] {\footnotesize $z_1',z_2$} (v-newl);
			\draw [->] (r) edge  node[sloped, below] {\footnotesize $z_1,z_2'$} (v-newr);
			\draw [->] (v) edge[dotted]  node[sloped, above] {}  (l);
			\draw [->] (v') edge[dotted]  node[sloped, above] {} (l');
			\draw [->] (v) edge[dotted]  node[sloped, above] {} (n);
			\draw [->] (v') edge[dotted]  node[sloped, above] {} (n');
			\draw [->] (v-newl) edge[dotted] node[sloped,above] {} (v-newl-kid-left);
			\draw [->] (v-newl) edge[dotted] node[sloped,above] {} (v-newl-kid-right);
			\draw [->] (v-newr) edge[dotted] node[sloped,above] {} (v-newr-kid-left);
			\draw [->] (v-newr) edge[dotted] node[sloped,above] {} (v-newr-kid-right);
			
			
		\end{tikzpicture}
		
	\end{subfigure}
	\vspace{2.0em}
	\begin{subfigure}[b]{ 1.0 \linewidth}
		\centering
		\setlength{\belowcaptionskip}{3pt}
		\setlength{\abovecaptionskip}{3pt}
		\caption{The induced tree of player $1$ at vertex $u$ given the message $z_2$ of player $2$. The tree has two subtrees: a left subtree that contains node $v$ and its descendants and a right subtree with node $v'$ and its descendants. 
		}	
		\label{subfig-induced}
		\begin{tikzpicture}
			\node[shape=circle,draw=black,minimum size=0.8cm] (u) at (1.5,1.5) {$u$};
			\node[shape=circle,draw=black,minimum size=0.8cm] (v) at (0,0) {$v$};
			\node[shape=circle,fill=magenta!10,draw=black,minimum size=0.8cm,label=below:{\small $(A,p_1)$}] (l) at (-1,-1) {\small$l$};
			\node[shape=circle,draw=black,minimum size=0.8cm] (v') at (3,0) {$v'$}; 
			\node[shape=circle,draw=black,,fill=magenta!50,minimum size=0.8cm,label=below:{\small $(A',p_1')$}] (l') at (4,-1) {\small$l'$};
			\node[shape=circle,draw=black,minimum size=0.8cm] (n') at (2,-1) {}; 
			\node[shape=circle,draw=black,minimum size=0.8cm] (n) at (1,-1) {};

			\draw [->] (u) edge  node[sloped, above] {\small$z_1$} (v);
			\draw [->] (u) edge  node[sloped, above] {\small $z_1'$} (v');
			\draw [->] (v) edge[dotted]  node[sloped, above] {}  (l);
			\draw [->] (v') edge[dotted]  node[sloped, above] {} (l');
			\draw [->] (v) edge[dotted]  node[sloped, above] {} (n);
			\draw [->] (v') edge[dotted]  node[sloped, above] {} (n');
		\end{tikzpicture}
		
	\end{subfigure}
	
\end{figure}

\begin{lemma}\label{lemma-known-prices}
	Fix some player $i$, vertex $u$, and messages of the other players $z^u_{-i}$ in a minimal dominant strategy mechanism. Consider the induced tree of player $i$ at vertex $u$ given $z^u_{-i}$. 
	If alternative $A\in \mathcal{A}$ appears in two different subtrees, then 
	all the leaves in this induced tree that are labeled with $A$ have the same payment for player $i$.
\end{lemma} 
\begin{proof}
	Let $\ell$ and $\ell'$ be two leaves labeled with
	$(A,p_A)$ and with $(A,p_A')$ that belong in different subtrees, $t$ and $t'$.
	By the minimality of the mechanism, every leaf in the protocol has 
	valuations such that  $(\mathcal{S}_1(v_1),\ldots,\mathcal{S}_n(v_n))$ reach this leaf. Thus,	 
	there exist valuations $v,v'\in V_i,v_{-i},v_{-i}'\in V_{-i}$ such that:
	$$
	(\mathcal{S}_i(v),\mathcal{S}_{-i}(v_{-i}))\to \ell,\quad
	(\mathcal{S}_i(v'),\mathcal{S}_{-i}(v_{-i}'))\to \ell'
	$$
	Observe the following strategy profile $\mathcal{S}_{-i}''$: For every valuation $v_{-i}''$, choose the actions specified by $\mathcal{S}_{-i}(v_{-i})$ until vertex $u$. Afterwards, at the subtree $t$, pick the actions that $\mathcal{S}_{-i}(v_{-i})$ specifies, and at the subtrees $t'$ pick the actions that $\mathcal{S}_{-i}(v_{-i}')$ specifies. Since $s_{-i}$ and $s_{-i}'$ do not diverge until vertex $u$, we have that
	$$
	(\mathcal{S}_i(v),\mathcal{S}_{-i}''(v_{-i}''))\to \ell, \quad (\mathcal{S}_i(v'),\mathcal{S}_{-i}''(v_{-i}''))\to \ell'
	$$
	where the profit of player $i$ with valuation $v$ has to be larger than her profit at $\ell'$, since $\mathcal{S}_i(v)$ is a dominant strategy for her. Thus, $v(A)-p_A\ge v(A)-p_A'$, so $p_A'\ge p_A$. By applying the same argument for the valuation $v'$, we get that $p_A\ge p_A' \implies p_A=p_A'$. 
	Thus, we have that every two leaves labeled with alternative $A$ in the induced tree of player $i$ given $z_{-i}^u$ have the same payment for player $i$, which completes the proof. 
\end{proof}

\section{Hardness of Exact Welfare Maximization}\label{sec-exact-welfare}

We now consider two domains that are generally considered ``easy'' in the sense that the welfare maximizing allocation can be found in time that is polylogarithmic in the representation size of the valuations. For both domains we show that -- in contrast to what is perhaps a common misconception -- incentive compatible mechanisms that maximize the welfare are incentive compatible only in ex-post equilibrium. For dominant strategy mechanisms, we show that the communication complexity is linear in the size of the representation of the valuations. 

Let us first recall how to obtain an ex-post incentive compatible algorithm for combinatorial auctions with two players. Denote the valuations by $v_1$ and $v_2$, and for every $1\leq x \leq m$ let $v'_1(x)=v_1(x)-v_1(x-1)$ and $v'_2(x)=v_2(x)-v_2(x-1)$ be the marginal values. The decreasing marginal values property guarantees that the welfare-maximizing allocation $(o_1,o_2)$ is a point where $v'_1$ and $v'_2$ \textquote{cross} each other, i.e. where $v'_1(o_1)\ge v'_2(o_2+1)$ and $v'_1(o_1+1)\le v_2'(o_2)$ (see also Lemma \ref{onlylemma}). $v_1'$ and $v_2'$ are monotone, so we have to find where two ordered arrays ``cross'' each other. Thus, a simple binary search will find the optimal allocation with $poly(\log m)$ value queries. VCG prices (player $1$ pays $v_2(m)-v(o_2)$, player $2$ pays $v_1(m)-v(o_1)$) guarantee incentive compatibility in an ex-post equilibrium. 

For combinatorial auctions with gross-substitutes bidders the optimal allocation can be found with communication $poly(m,n)$ for valuations that can be represented by $exp(m)$ bits \cite{NS06}.

Despite the fact that in ex-post equilibrium the optimal welfare can be achieved efficiently, if we require dominant strategy equilibrium, we get an exponential blowup in the communication complexity in both domains.

\begin{theorem}\label{newmainthm}
	Fix a normalized  mechanism which implements in dominant strategies a welfare-maximizer for a multi-unit auction where the valuations have decreasing marginal utilities, and the value of a bundle can be represented with $\mathcal{O}(\log(m))$ bits. Then, the communication complexity of the mechanism is 	$\Omega(m\log(m))$.
\end{theorem}


\begin{theorem}\label{gsmainthm}
	Fix a normalized mechanism which implements in dominant strategies a welfare-maximizer for a combinatorial auction with gross substitutes valuations, where the value of each bundle can be represented with $poly(m)$ bits. 
Then, the communication complexity of the mechanism is exponential in $m$.	
\end{theorem}

Proofs of these theorems can be found in Appendices \ref{sec-mua-impossibility} and \ref{sec-gs-impossibility}. Both proofs share a similar structure. We now give some intuition for the proof in the context of multi-unit auctions with decreasing marginal values.


Consider the following scenario. We restrict ourselves to some (large) set of valuations. Suppose that Bob is decisive: for (almost) every allocation $(s,m-s)$, there exist two valuations of Bob $v_b^1,v_b^2$, such that for each valuation $v_a$ of Alice that is in this set, the optimal allocation in the instances $(v_a,v_b^1)$ and $(v_a,v_b^2)$ is $(s,m-s)$. Furthermore, assume that the dominant strategy of Bob dictates a different message when his valuation is $v_b^1$ than when it is $v_b^2$.

Let $v_a^1,v_a^2$ be two valuations of Alice that are in the set. Since we are implementing a welfare maximizer, Bob must get $m-s$ items for every valuation  $v_a^1,v_a^2$ of Alice. For simplicity, we assume for now (but not in the proof) that we are using VCG payments, so Bob's payment might be different: it can be either $v_a^1(m)-v_a^1(s)$ or $v_a^2(m)-v_a^2(s)$. Thus, if Bob sends a different message for $v_b^1$ than that of $v_b^2$ and Alice sends the \emph{same} message for both $v_a^1,v_a^2$, Bob does not have a dominant strategy, since Alice can ``force'' him to choose one such message by guaranteeing that his payment will be higher otherwise.

To avoid this, Alice has to ``commit'' on her value for $s$ items. That is, if $v_a^1$ and $v_a^2$ have a different value for $s$ items, then the message that the dominant strategy of Alice dictates cannot be the same for both of them. In fact, we show that this implies, roughly speaking, that Alice's first message must be so informative that we can fully reconstruct Alice's valuation from her first message. Thus, her first message is very big, and the proof is complete. 
The main challenge of the proof is to construct a big enough set of valuations that satisfies all those properties. For an illustration, see Figure \ref{exact-welfare-tikz}.

\begin{figure} [H] 
	\centering
	\caption{Let $v_b^1,v_b^2$ be two valuations of Bob that dictate different messages $z_B^1,z_B^2$ at the root vertex $r$.  $v_a^1,v_a^2$ are two valuations of Alice such that she sends the message $z_A$ for both of them at the first round. Assume that the unique optimal allocation for the valuations $(v_a^1,v_b^1)$, $(v_a^1,v_b^2)$, $(v_a^2,v_b^2)$ $(v_a^2,v_b^2)$ is $(s,m-s)$.
		\newline 
		Below we have the induced tree of Bob at vertex $r$ given the message $z_A$. The leaves $l_{1,1},l_{1,2},l_{2,1},l_{2,2}$ respectively are the leaves that the protocol ultimately reaches given the valuations $(v_a^1,v_b^1)$, $(v_a^1,v_b^2)$, $(v_a^2,v_b^1)$, $(v_a^2,v_b^2)$ according to the dominant strategies of Alice and Bob. By assumption, they are all labeled with the welfare maximizing allocation and with VCG prices. 
		\newline Note that if $v_a^1(m)-v_a^1(s) >  v_a^2(m)-v_a^2(s)$,  sending $z_B^1$ is no longer a dominant strategy for Bob given the valuation $v_b^1$ (he might get better price for $m-s$ items by sending $z_B^2$). Thus, to avoid this situation, Alice has to commit for her value of $s$ items.}
	
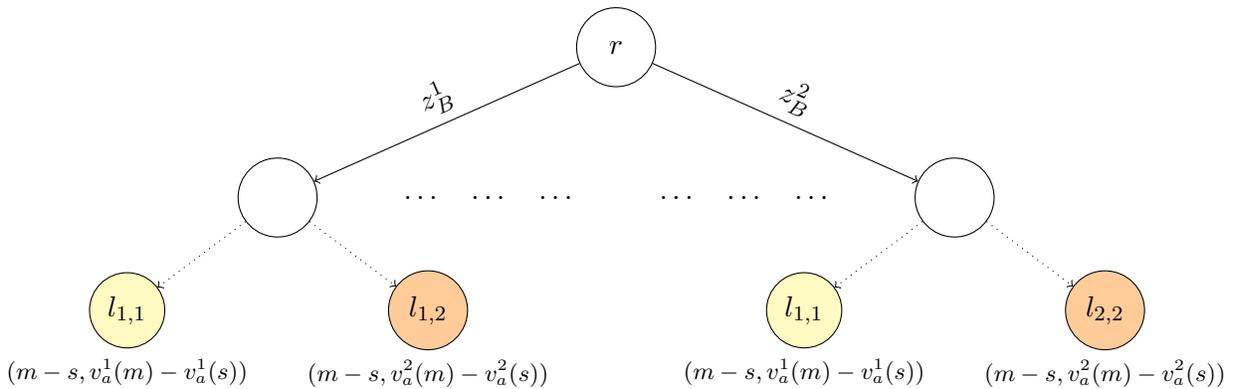
\begin{tikzpicture}
		\node[shape=circle,draw=black,minimum size=1.05cm] (r) at (1.5,1.5) {$r$};
		\node[shape=circle,draw=black,minimum size=1.05cm] (v) at (-3,-0.5) {};
		\node[shape=circle,fill=yellow!30,draw=black,minimum size=1cm,label=below:{\footnotesize $(m-s,v_a^1(m)-v_a^1(s))$}] (l) at (-5,-2) {$l_{1,1}$};
		\node[shape=circle,draw=black,minimum size=1.05cm] (v') at (6,-0.5) {}; 
		\node[shape=circle,fill=orange!40,draw=black,minimum size=1.05cm,label=below:{\footnotesize $(m-s,v_a^2(m)-v_a^2(s))$}] (l') at (8,-2) {$l_{2,2}$};
		\node[shape=circle,fill=yellow!30,draw=black,minimum size=1cm,label=below:{\footnotesize $(m-s,v_a^1(m)-v_a^1(s))$}] (n') at (4,-2) {$l_{1,1}$};
		
		\node[shape=circle,fill=orange!40,draw=black,minimum size=1.05cm,,label=below:{\footnotesize $(m-s,v_a^2(m)-v_a^2(s))$}] (n) at (-1,-2) {$l_{1,2}$};
		
		\node[] (dots-r) at (3.2,-0.5) {$\ldots \quad \ldots \quad \ldots$};
		\node[] (dots-l) at (-0.2,-0.5) {$\ldots \quad \ldots \quad \ldots $};
		
		%
		%

		\draw [->] (r) edge  node[sloped, above] {$z_B^1$} (v);
		\draw [->] (r) edge  node[sloped, above] {$z_B^2$} (v');
		\draw [->] (v) edge[dotted]  node[sloped, above] {}  (l);
		\draw [->] (v') edge[dotted]  node[sloped, above] {} (l');
		\draw [->] (v) edge[dotted]  node[sloped, above] {} (n);
		\draw [->] (v') edge[dotted]  node[sloped, above] {} (n');
	\end{tikzpicture}
	
	\label{exact-welfare-tikz}
\end{figure}

To complement this hardness result, in Subsection \ref{subsec-fptas} we show that for multi-unit auctions with decreasing marginal values, arbitrarily good approximations are possible in dominant strategies (a ``communication FPTAS'').

\begin{theorem}\label{thm-mua-fptas}
For every $\eps>0$, there is a dominant strategy algorithm for multi-unit auctions with decreasing marginal values that makes $poly(\frac 1 \eps,n)$ value queries and provides an allocation with social welfare at least $(1-\eps)\cdot OPT$, where $OPT$ is the value of the optimal social welfare.
\end{theorem}

In contrast, the only known upper bound on the approximation ratio of efficient dominant strategy mechanisms for combinatorial auctions with gross substitutes valuations is $\mathcal{O}(\sqrt m)$ \cite{DNS05}. Determining the approximation ratio possible for this class remains an open problem.

\section{Inapproximability of Mechanisms for General Valuations}\label{sec-general-outline}

In this section we prove that no deterministic dominant strategy mechanism with polynomial communication for general valuations achieves an approximation ratio better than $m^{1-\eps}$. In contrast, there is a \emph{randomized} dominant strategy mechanism that achieves an approximation ratio of $O(\sqrt m)$ \cite{DNS06}. Note that an approximation ratio of $O(\sqrt m)$ is the best  possible with polynomial communication even when ignoring incentives \cite{NS06}. We refer the reader to the full version for the exact statement. 


\begin{theorem} \label{theorem-general-impossibility}
	Fix some $\epsilon>0$. Let $\mathcal{M}$ be a deterministic, normalized, no negative transfers dominant strategy mechanism for combinatorial auctions with $n=\Omega(m^{2-\epsilon})$ bidders with general valuations, where the value of each bundle can be described with $poly(m)$ bits. If the approximation ratio of $\mathcal M$ is better than $m^{1-4\eps}$, then the communication complexity of $\mathcal{M}$ is at least
	 $poly(\frac {2^{m^{\frac {\eps^2} 2}}} n)$ bits.  
\end{theorem}

The proof sketch is as follows.

\subsubsection*{Step I: A Lower Bound on Simultaneous Algorithms (Section \ref{sec-simultanous-general})}

In general, we prove Theorem \ref{theorem-general-impossibility} by showing that dominant strategy mechanisms for combinatorial auctions with general valuations are as powerful as simultaneous (non-incentive compatible) algorithms. Recall that perhaps the ``easiest'' way to obtain a dominant strategy mechanism is by designing an ex-post mechanism and making it ``simultaneous''. Indeed, almost all deterministic dominant strategy mechanisms in the literature are simultaneous. Thus, the first step is done in Section \ref{sec-simultanous-general}: a proof that no simultaneous algorithm can achieve an approximation ratio better than $m^{ 1 -\eps}$ with polynomial communication. 

\subsubsection*{Step II: Efficient Dominant Strategy Mechanisms Imply Efficient Simultaneous Algorithms (Section \ref{sec-general2})}
Step I does not suffice because 
not all dominant strategy mechanisms are simultaneous. Consider the following example of a combinatorial auction with two players with additive valuations $v_A, v_B$, where all values are integers in $\{1,\ldots,{{\frac m 2}\choose 2}\}$. 
Split the items arbitrarily to two equal sets $A$ and $B$. Alice can win only items from $A$, and Bob can win only items from $B$. We associate each possible value of Alice   for some arbitrary item $b\in B$, $v_A(\{b\})$, with a distinct pair of items in $B$, and similarly we associate Bob's value  for some item $a\in A$, $v_B(\{a\})$, with a distinct pair of items in $A$. 
According to the social choice function,  Alice wins her more valuable item among the pair that $v_B(\{a\})$ points to and Bob wins his more valuable item among the pair that $v_A(\{b\})$ points to.

A protocol with $\mathcal{O}(\log m)$ bits where they simultaneously send $v_A(\{b\})$ and $v_B(\{a\})$ in the first round and then each reports the preferred item among the possible two items is clearly dominant-strategy incentive compatible. However, it is not hard to show that any simultaneous mechanism for this auction requires $\Omega(m\cdot \log m)$ bits.
Thus, this instance exhibits a separation between dominant strategy and simultaneous implementations. 

However, we will show that if a mechanism provides an approximation ratio better than $m^{1-4\epsilon}$ to the welfare for general valuations, it can be used to construct a simultaneous algorithm with approximation ratio $m^{1-\eps}$.

\section{Simultaneous Algorithms for Combinatorial Auctions}\label{sec-simultanous-general}
\label{sec:simul}

In this section we consider simultaneous combinatorial auctions. The hardness results that we obtain will be used later in Section \ref{sec-general2} to prove impossibility result for dominant strategy mechanisms for combinatorial auctions with general valuations.

The setup is as follows: as usual, there is a set of items $M, |M|=m$, and $n$ bidders with valuation functions $v_1,\ldots,v_n: 2^M \rightarrow \RR_+$. Each of them simultaneously sends a message $s_i$ to a central authority; the messages all together are bounded by bit-length $L$. The algorithm, given the messages, produces an allocation $\cA(s_1,\ldots,s_n) = (A_1,A_2,\ldots,A_n)$. The goal is the maximize the social welfare $\sum_{i=1}^{n} v_i(A_i)$. We impose no computational constraints on the bidders or the central authority. 

\begin{theorem}\label{thm-simultaneous-general}
For $m$ items and $n=\Omega(m^{2-\epsilon})$ bidders with general monotone (binary) functions as valuations, there is no randomized simultaneous mechanism with messages of size at most $\frac {2^{m^{\frac {\epsilon^2} 2}}} n$ which achieves an approximation ratio better than $m^{1-\epsilon}$, for any fixed $\epsilon>0$. 
\end{theorem}

\begin{theorem}\label{thm-simultaneous-matroids}
For $m$ items and  $n=\Omega(m^{\frac{3}{32}})$ bidders with matroid rank functions as valuations, there is no randomized simultaneous mechanism with messages of length $\frac {2^{m^{\frac 1 {32}}}} n$ which achieves an approximation ratio better than $m^{\frac 1 {16}}$.
\end{theorem}
The first theorem is directly used in the construction of Section \ref{sec-general2}. The second one (proof in Appendix \ref{sec-simultaneous-GS}) solves an open problem of \cite{DNO14} that asks whether there is a simultaneous algorithm that provides a constant approximation for submodular valuations. Therefore, Theorem \ref{thm-simultaneous-matroids} answers this question negatively, even for matroid rank functions (which are also gross substitutes valuations). We note that the result almost settles completely the approximation ratio achievable in this setting, as a simultaneous $\tilde O(m^{\frac 13})$-approximation algorithm for all subadditive valuations exists \cite{DNO14}.

\subsection{Proof of Theorem \ref{thm-simultaneous-general}: An Impossibility for General Valuations} \label{impossible-sim-general-subsec}

\paragraph{The Hard Distribution.}
We prove our impossibility for randomized mechanisms by applying Yao's principle. Thus, we now describe a distribution over instances and analyze the performance of deterministic mechanisms on it. 

Fix $\epsilon>0$. Let the number of bidders be $n = m^{2-\epsilon}-m$, divided into  $\ell = m^{1-\epsilon}-1$ groups $G_1, \ldots, G_\ell$ of $m$ bidders each. Let $(A_1,A_2,\ldots,A_\ell,B)$ be a random partitioning of the $m$ items, such that for all $j$, $|A_j| = |B| = m^{\epsilon}$ (note that $m^{\epsilon} (\ell + 1) = m$).
For each group $G_j$, the set of relevant items is $A_j \cup B$. Let $\cA_j$ be a family of $t = 2^{\Theta(\epsilon^2 m^\epsilon)}$ subsets of $A_j \cup B$ of size $m^\epsilon$, such that one of the sets is always $A_j$ and the other sets are chosen uniformly at random. 
By standard concentration bounds, with high probability, these sets overlap pseudo-randomly in the sense that the intersection of any two sets in $\cA_j$ has size $(\frac12 \pm \epsilon) m^{\epsilon}$. In the following, we will only use a weaker statement which is that for any two sets $A \in \cA_j, A' \in \cA_{j'}$ such that $A \neq A_j, A' \neq A_{j'}$, we have $A \cap A' \neq \emptyset$ w.h.p. For any two such sets $A,A'$, we have $A \subseteq B \cup A_j$ and $A' \subseteq B \cup A_{j'}$, and the probability that they are disjoint is at most $e^{-\Omega(m^{\epsilon)}}$, since for every $b \in B$, the probability that $b \in A \cap A'$ is $1/4$ and these events are negatively correlated. The number of such pairs of sets is $2^{\Theta(\epsilon^2 m^\epsilon)}$; i.e. by the union bound, all pairs of sets $A \in \cA_j \setminus \{A_j\}, A' \in \cA_{j'} \setminus \{A_{j'}\}$ intersect with probability $1 - e^{-\Omega(m^\epsilon)}$.

For each bidder $i$ in group $G_j$, the valuation is supported on the set of items $A_j \cup B$. For each bidder $i \in G_j$, we choose a random sub-family $\cB_i \subseteq \cA_j$ such that each set in $\cA_j$ appears in $\cB_i$ independently with probability $\frac{1}{m}$. More specifically, we do this in such a way that for each set $A \in \cA_j$, we choose independently a random bidder $i \in G_j$ for whom $A \in \cB_i$; for the other bidders $i' \neq i$, $A \notin \cB_{i'}$.

We define the valuation of bidder $i$ as:
$$
v_i(S)  =\begin{cases}
 1 &   S \supseteq B \mbox{ for some } B \in \cB_i, \\
0 & \mbox{    otherwise}.
\end{cases}
$$
I.e., a bidder $i$ is satisfied if she gets the items of some set in $\cB_i$. We call each subset in $\cB_i$ a set that bidder $i$ is \emph{interested in}. In particular, if $A_j \in \cB_i$, one way to satisfy a bidder in group $G_j$ is to allocate the set $A_j$. However, this set is valuable only for the bidder $i \in G_j$ such that $A_j \in \cB_i$. We call  this bidder {\em special} in group $G_j$. 

Note also that only a small number of non-special bidders can be satisfied overall, since these bidders want random sets which intersect with each other with high probability. This leads to the following lemma.

\begin{lemma}\label{lemma-welfare-comes-from-special}
	With probability $1-e^{-\Omega(m^\epsilon)}$, the welfare of an allocation is at most $1$ plus the number of special bidders who receive the respective set $A_j$.
\end{lemma} 
\begin{proof}
	Any player who is not special can get value $1$ only if she gets a set in $\cB_i$, which does not include the special set $A_i$. As we argued above, all the sets in $\cB_i \setminus \{A_i\}$, for different values of $i$, intersect pairwise with probability $1 - e^{-\Omega(m^\epsilon)}$. Hence, at most one bidder can be satisfied this way. 
	Any additional value comes from special bidders who receive the respective set $A_j$.
\end{proof}

We prove Theorem \ref{thm-simultaneous-general} by proving the following proposition, that will be used directly in the impossibility result for dominant strategy mechanisms (Section \ref{sec-general2}).
\begin{proposition}\label{prop-simul-hard-dist}
There is no simultaneous mechanism with messages of size at most $\frac {2^{m^{\frac {\epsilon^2} 2}}} n$ which achieves in expectation an approximation ratio better than  $m^{1-\epsilon}$ on instances sampled from the hard distribution described above. 
\end{proposition}

\begin{lemma}\label{optimal-high-in-expectation}
The optimal welfare for every instance is $OPT \geq m^{1-\epsilon} - 1$.
\end{lemma}

\begin{proof}
Each group $G_j$ contains exactly 1 bidder who wants the special set $A_j$.
Hence, a solution which allocates $A_j$ to the special bidder in group $G_j$, achieves value exactly $\ell = m^{1-\epsilon}-1$.
\end{proof}

We now analyze the expected welfare achieved by any mechanism on the random instance described above.
By Yao's principle, we assume that that the mechanism is deterministic.
A good mechanism should ensure that many of the sets $A_j$ go to some special bidder in group $G_j$. But how can it determine who the special bidders are? For that, it would intuitively need to know the value of $A_j$ for each bidder, but the bidders do not know which of their sets is special and there are too many sets to encode in a message. Our goal is to prove that this indeed implies an impossibility result in the simultaneous model.

We prove that the messages $(s_i: i \in G_j)$ sent by the bidders in group $G_j$ typically do not give us much information about who the special bidder is.
Suppose that the messages $(s_i: i \in G_j)$ altogether have bit-length bounded by $L$. These messages are chosen depending on the random valuations $(v_i: i \in G_j)$, so each choice of messages appears with a certain probability. We distinguish between ``frequent'' and ``rare'' message sets. 

\begin{definition}
We call a message set $(s_i: i \in G_j)$ {\em frequent} if it appears with probability at least $\frac 1 {4^L}$; otherwise it is {\em rare}.
\end{definition}

Observe that since the total number of messages is at most $2^L$, all rare messages together appear with probability less than $\frac{1}{2^L}$. Next, we prove that a frequent set of messages cannot give us much information about the distribution of high-value sets. Recall that without any conditioning, for a particular bidder $i \in G_j$, each set in $\cA_j$ is chosen to be in $\cB_i$ with probability $\frac 1 {|G_j|} = \frac 1 m$. The key lemma is the following.

\begin{lemma}
\label{lem:frequent}
Let $\bar{s} = (s_i: i \in G_j)$ be a frequent set of messages. Then for every bidder $i \in G_j$, there are fewer than $L \cdot  |G_j|$ sets $A \in \cA_j$ such that conditioned on bidders in $G_j$ sending $\bar{s}$, $\Pr[A \in \cB_i \mid \bar{s}] > \frac 7 { |G_j|}$.
\end{lemma}

\begin{proof}
Suppose that $\bar{s}$ is a frequent set of messages and there is a family of $L \cdot |G_j|$ sets $A \in \cA_j$ with $\Pr[A \in \cB_i \mid \bar{s}] > \frac 7 { |G_j|}$; denote it by $\cS \subset \cA_j$.

Consider the choices whether $A \in \cB_i$ for $A \in \cS$. Without any conditioning, each $A$ is chosen to be in $\cB_i$ independently with probability $\frac 1 {|G_j|}$. In expectation, the number of sets in $\cS \cap \cB_i$ is $\frac {|\cS|} { |G_j|} = L$. Hence, by the Chernoff bound,
$$ \Pr[|\cS \cap \cB_i| > (1+\delta) L] \leq \left( \frac{e^\delta}{(1+\delta)^{1+\delta}} \right)^L < \frac{1}{2^{\delta L}}$$
for $\delta \geq 5$.
Consider now the conditioning on $\bar{s}$. Since $\Pr[\bar{s}] \geq \frac 1 {4^L}$, this conditioning cannot increase the probability of any event by more than a factor of $4^L$. Therefore,
$$ \Pr[|\cS \cap \cB_i| > (1+\delta) L \mid \bar{s}] < \frac{4^L}{2^{\delta L}}.$$
For $\delta=5$, we get $\Pr[|\cS \cap \cB_i| > 6L \mid \bar{s}] < \frac 1 {8^L}$ and the tail probability decays exponentially beyond that. This implies that:
$$ \E[|\cS \cap \cB_i| \mid \bar{s}] < 7L$$

We assumed that $\Pr[A \in \cB_i \mid \bar{s}] > \frac {7 }{|G_j|}$ for every $A \in \cS$, and $|\cS| = L \cdot  |G_j|$, a contradiction.
\end{proof}

\medskip

We are now ready to conclude the proof of Proposition \ref{prop-simul-hard-dist} by showing that for any simultaneous mechanism with messages of total size at most $L=2^{m^{\epsilon^2 / 2}}$, executed on a random instance as described above, the expected welfare is $O(1)$, while the optimum is $OPT=\Omega(m^{1-\epsilon})$. 

We already showed above that $OPT=\Omega(m^{1-\epsilon})$.
Let us bound the expected welfare achieved by bidders in group $G_j$, assuming that their messages together are bounded by $L$ bits. The contribution of cases where $\bar{s} = (s_i: i \in G_j)$ is a rare message set is small, because they happen with total probability less than $\frac 1 {2^L}$; hence the expected contribution from these cases is negligible (less than $\frac{m}{2^L}$).

In the case of a frequent message set $\bar{s}$, consider the partitioning of the items $B \cup A_j$ among the bidders in group $G_j$. This partitioning is determined by $\bar{s}$. Lemma~\ref{lem:frequent} says that for each bidder $i \in G_j$, fewer than $L \cdot |G_j|$ sets $A \in \cA_j$ have the property that $\Pr[A \in \cB_i \mid s_i] > \frac 7 {|G_j|}$. Recall that $|G_j| = m$. Hence, among all the sets in $\cA_j$, at most $L \cdot  |G_j|^2 = L  m^2$ sets are ``biased'' in the sense that the value is $1$ for some bidder with conditional probability more than $\frac 7 m$. 

Considering group $G_j$ in isolation, the special set $A_j$ is uniformly random among all sets in $\cA_j$, and this is true even conditioned on the valuations in group $G_j$, and hence also conditioned on the message set $\bar{s}$. (Recall that given the set of items $B \cup A_j$ relevant to group $G_j$, there is no way to distinguish the subset $A_j$, which is equally likely to be any of the sets in $\cA_j$). Furthermore, unless the special set $A_j$ is one of the at most $L m^2$ biased sets discussed above, however the items in $A_j$ are allocated, each bidder is the special bidder for it with conditional probability at most $\frac 7 m$. If $A_j$ is split among multiple bidders, none of them receives all of $A_j$. If $A_j$ goes to a particular bidder, then this bidder is special with conditional probability at most $\frac 7 m$. Hence, conditioned on a message set $\bar{s}$, we satisfy a special bidder with conditional probability at most $\frac 7 m$.

Finally, in case the special set $A_j$ is one of the biased sets, we can assume that we derive value of $1$ from it;  however this happens with probability at most $\frac {L m^2} { |\cA_j|} = O(m^2\cdot  2^{-m^ {\frac {\epsilon^2} 2}})$. 
The contribution 
of these cases is negligible. 

We have $\ell = m^{1-\epsilon} - 1$ groups of bidders. There are also the items in $B$, which can contribute value at most $1$ in total, with high probability. Hence, the total expected welfare is at most $1 + \frac {7 \ell } m = O(1)$.

\section{Proof of Theorem \ref{thm-simultaneous-matroids}: Simultaneous Algorithms for Matroid Rank Functions}
\label{sec-simultaneous-GS}

Here we combine the ideas of Section~\ref{sec:simul} with a construction of matroids by Balcan and Harvey, which we recap here.

\begin{theorem}[\cite{BH11}]
For any $k \geq 8$ with $k = 2^{o(\tilde{m}^{\frac 1 3})}$, there exists a family of sets $\cA \subseteq 2^{[\tilde{m}]}$ and a family of matroids $\{ \cM_\cB: \cB \subseteq \cA \}$ with the following properties:
\begin{itemize}
\item $|\cA| = k$ and $|A| = \tilde{m}^{\frac 1 3}$ for every $A \in \cA$.
\item For every $\cB \subseteq \cA$ and every $A \in \cA$, we have:
$$ \mbox{rank}_{\cM_\cB}(A) = |A|,     \ \ \ \ \   \mbox{if } A \in \cB.$$
$$ \mbox{rank}_{\cM_\cB}(A) = 8 \log k, \ \ \ \ \  \mbox{if } A \in \cA \setminus \cB,$$
\end{itemize}
\end{theorem}

For an instance of combinatorial auctions with $m$ items, we will use this construction with $\tilde{m} = m^{\frac 3 4}$ and $k = 2^{m^{\frac 1 {16}}}$; hence $\mbox{rank}_{\cM_\cB}(A)$ is either $m^{\frac 1 4}$ or $8 \cdot m^{\frac 1 {16}}$, depending on the choice of $\cB$.\footnote{Note that compared to Balcan-Harvey, we switch the meaning of $\cB$ and $\cA \setminus \cB$; we find it more natural to use $\cB$ to denote bases of the matroid. However, the reader should keep in mind that there are also other bases in $\cM_\cB$.}

\paragraph{The Hard Distribution.}
We prove our impossibility for randomized mechanism by applying Yao's principle. Thus, we now describe a distribution over instances and analyze the performance of deterministic mechanisms on it. We define instances as follows. Let the number of bidders be $n = m^{\frac 1 8} (m^{\frac 3 4} - m^{\frac 1 2} + 1)$, divided into  $\ell = m^{\frac 3 4} - m^{\frac 1 2} + 1$ groups $G_1, \ldots, G_\ell$ of $m^{\frac 1 8}$ bidders each. Let $(A_1,A_2,\ldots,A_\ell,B)$ be a random partitioning of the $m$ items, such that $|A_j| = m^{\frac 1 4}$ and $|B| = m^{\frac 3 4} - m^{\frac 1 4}$. (Note that $m^{\frac 1 4} \cdot \ell + m^{\frac 3 4} - m^{\frac  1 4} = m$.)
For each bidder $i$ in group $G_j$, the valuation is supported on the set of items $A_j \cup B$; it is a matroid rank function of a Balcan-Harvey matroid on $\tilde{m} = m^{\frac 3 4}$ elements, with parameter $k = 2^{m^{\frac 1 {16}}}$, defined by set families $\cB_i \subseteq \cA_j$ and embedded in $A_j \cup B$ so that a random one of the sets in $\cA_j$ is mapped onto $A_j$, and the remaining elements are mapped randomly onto $B$. (Note that we use a $j$ subscript for $\cA_j$, because this family is shared among all the bidders in $G_j$.) The sub-family $\cB_i \subseteq \cA_j$ of high-value sets for bidder $i$ is chosen randomly in the following way:
For each set $A \in \cA_j$, we choose independently and uniformly at random one bidder $i$ in group $G_j$ such that $A \in \cB_i$. For all the other bidders $i' \in G_j$, we don't include $A$ in $\cB_{i'}$.
Note that in expectation we have $\E[|\cB_i|] = \frac {|\cA_j|} {|G_j|} = m^{-\frac 1 8}\cdot  2^{m^{\frac 1 {16}}}$, and $|\cB_i|$ is tightly concentrated. Exactly one bidder in group $G_j$ has a high value for the set mapped to $A_j$, and we call 
this bidder the {\em special bidder} in $G_j$.

\begin{lemma}
The optimal welfare for this instance is $OPT=m$.
\end{lemma}

\begin{proof}
In each group $G_j$, we allocate the special set $A_j$ to the special bidder, who receives value $|A_j| = m^{\frac 1 4}$. The items in $B$ can be allocated arbitrarily to some non-special bidders (since $|B| = m^{\frac 3 4} - m^{\frac 1 4}$ and the number of non-special bidders is $\Omega(m^{\frac 7 8})$), who get value $1$ each. Hence, each item contributes exactly $1$ and $OPT = m$. 
\end{proof}

We analyze the expected welfare achieved by any mechanism on the random instance described above. We make the following simple claim.

\begin{lemma}
If at most $m_j$ of the items in $A_j$ are allocated to the special bidder in group $G_j$, then
the welfare of the allocation is at most $O(m^{\frac {15}{16}}) + \sum_j m_j$.
\end{lemma}

\begin{proof}
The items in $B$ contribute at most $|B| = m^{\frac 3 4} - m^{\frac 1 4}$ altogether.
Any player who is not special can get value at most $O(m^{\frac 1 {16}})$ from the items in $A_j$, hence 
all these players together can get at most $m^{\frac 3 4} + O(n \cdot m^{\frac 1 {16}}) = O(m^{\frac {15}{16}})$.
Finally, the special players can get at most $m_j$ each from the items in $A_j$; hence $\sum_j m_j$.
\end{proof}

From here, the proof is similar to the proof of Theorem \ref{thm-simultaneous-general}. We complete the proof by showing that for any simultaneous mechanism with messages of total size at most $L = 2^{m^{\frac 1 {32}}}$, executed on a random instance as described above, the expected welfare is $O(m^{\frac {15}{16}})$, while the optimum is $OPT=m$. 

We already showed above that $OPT=m$.
Let us bound the expected welfare achieved by bidders in group $G_j$, assuming that their messages together are bounded by $L$ bits. The contribution of cases where $\bar{s} = (s_i: i \in G_j)$ is a rare message set is small, because they happen with total probability less than $\frac{1}{2^L}$; hence the expected contribution from these cases is negligible (less than $\frac{m}{2^L}$).

In the case of a frequent message set $\bar{s}$, consider the partitioning of the items $B \cup A_j$ among the bidders in group $G_j$. This partitioning is determined by $\bar{s}$. Lemma~\ref{lem:frequent} says that for each bidder $i \in G_j$, fewer than $L \cdot |G_j|$ sets $A \in \cA_j$ have the property that $\Pr[A \in \cB_i \mid s_i] > \frac 7 {|G_j|}$.  Here, we have $|G_j| = m^{\frac 1 8}$. Hence, among all the sets in $\cA_j$, at most $L |G_j|^2 = L m^{\frac 1 4}$ sets are ``biased'' in the sense that the value is high for some bidder with conditional probability more than $\frac{7}{m^{1/8}}$. 

The special set $A_j$ is uniformly random among all sets in $\cA_j$, and this is true even conditioned on the valuations in group $G_j$, and hence also conditioned on the message set $\bar{s}$. (Recall that given the set of items $B \cup A_j$ relevant to group $G_j$, there is no way to distinguish the subset $A_j$, which is equally likely to be any of the sets in $\cA_j$). Furthermore, unless the special set $A_j$ is one of the at most $L \cdot m^{1/4}$ biased sets discussed above, however the items in $A_j$ are split, each bidder is the special bidder for it with conditional probability at most $\frac{7}{m^{1/8}}$. Suppose that bidder $i$ receives $k_i$ items from $A_j$ in this allocation. Then the expected value that the bidders derive from $A_j$ is at most
$$ \sum_{i \in G_j} \frac{7}{m^{1/8}} \cdot k_i + \sum_{i \in G_j} \left(1 - \frac{7}{m^{1/8}} \right) O(m^{\frac 1 {16}}) 
< \frac{7\cdot |A_j|}{m^{1/8}} + O(m^{1/16} |G_j|) = O(m^{3/16})$$
because a bidder who is special gets value $1$ for each item received from $A_j$, $|A_j| = m^{\frac 1 4}$, and a bidder who is not special receives value at most $O(m^{\frac 1 {16}})$ from $A_j$.
Finally, in case the special set $A_j$ is one of the biased sets, we can assume that we derive full value $|A_j| = m^{\frac 1 4}$ from it; however this happens with probability at most $L \cdot \frac  {m^{\frac 1 4}} {|\cA_j|} = O(m^{\frac 1 4} \cdot 2^{-m^{\frac 1 {32}}})$. The contribution of these cases is negligible.

We have $\ell \leq m^{\frac 3 4}$ groups of bidders. There are also the items in $B$, $|B| \leq m^{\frac 3 4}$, which can contribute at most $|B|$ in total. Hence, the total expected welfare is at most $|B| + O(\ell\cdot  m^{\frac 3 {16}}) = O(m^{\frac {15}{16}})$.


\section{Multi-Unit Auctions With Decreasing Marginal Valuations}\label{sec-mua-impossibility}


Consider a social choice function $f$ that always outputs an allocation that maximizes the welfare. This social choice function can be implemented in dominant strategies by the VCG mechanism. 
 The next theorem shows that even if we restrict ourselves to a subset of the valuations such that each valuation can be represented by $\mathcal O(m\cdot \log m)$ bits, any dominant strategy normalized and no-negative-transfers implementation of $f$  requires $ \Omega(m\cdot \log m)$ bits, even when there are only two players. In contrast, recall that an ex-post implementation of a welfare maximizer  with VCG payments for this set of valuations requires only $poly(\log m)$ bits.

In Section \ref{sec-gs-impossibility} we show an exponential blow up also in the implementation of dominant-strategy welfare maximizers for combinatorial auctions with gross substitute valuations. The two hardness proofs share a very similar structure.

\subsection{The Hardness Result: Proof of Theorem \ref{newmainthm}}\label{subsec-hardness}

Consider a multi-unit auction of $m\ge 5$ items and two players (Alice and Bob). 
The valuations that we consider belong to three families: ``semi-decisive'' valuations $V^{D}$, non-decisive valuations $V^{ND}$ and another set of valuations $V^{P}$ that we will use to show that payments can be used as sketches of valuations.  

Every semi-decisive and non-decisive valuation will have a ``weight'' which is a scalar $\gamma \in \{1,\ldots,m^{5}\}$ that captures its magnitude. We now define the set $V^{D,\gamma}$ of semi-decisive valuations with the scalar $\gamma$.  
Every $v\in V^{D,\gamma}$ has two parameters: a special bundle  $x^\ast\in \{2,\ldots,m-2\}$ and a margin $d_m \in \{\frac{1}{2},1\}$ such that: 
\begin{equation*}\label{v2def}
v(x)=	\begin{cases}
0 \quad&x=0, \\
\gamma\cdot 3m^8 \quad&x=1, \\
\gamma\cdot (m^2-m+1)+v(x-1) \quad& x\in \{1,\ldots,x^\ast\}, \\
\gamma+v(x-1) \quad&x\in \{x^\ast+1,\ldots,m-1\}, \\
d_m +v(m-1) 
\quad&x=m.
\end{cases}
\end{equation*}
To define the set of non-decisive valuations, we define for every number of items $x\in \{2,\ldots,m\}$ its set of all its possible marginal utilities:
\begin{gather*}
	\forall x\in\{2,\ldots,m-1\}, \quad
	D_{x}= \{m^2-m\cdot x,m^2-m\cdot x+1,\ldots, m^2-m\cdot (x-1)\} \\
	D_m=\{\frac{1}{2},1\}
\end{gather*}
For every weight $\gamma \in \{1,\ldots, m^5\}$, every valuation in the set $V^{ND,\gamma}$ is parameterized by a vector $(d_2,\ldots,d_{m})\in D_2\times\cdots\times D_{m}$ such that:
\begin{gather*} \label{v1gammadescription}
 v(x)=	\begin{cases}
	0 \quad& x=0, \\
	\gamma\cdot 3m^8 \quad& x=1, \\
	\gamma\cdot d_x+v(x-1) \quad& x\in \{1,\ldots,m-1\}, \\
	d_m+v(m-1)    
	\quad&x=m.
\end{cases}
\end{gather*}
Throughout the proof, we use the notations $V^{ND}=\bigcup\limits_{\gamma=1}^{m^5} V^{ND,\gamma}$ and $V^{D}=\bigcup\limits_{\gamma=1}^{m^5} V^{D,\gamma}$.  

We are now going to define another set of valuations $V^{P}$ with the purpose of guaranteeing that  different valuations in $V^{D}\cup V^{ND}$ induce different payments. We use this fact later on to sketch valuations.  
Every $v\in V^{P}$ has a valuation $v'\in V^{ND}\cup V^{D}$, a sign $sn\in \{0,1\}$ and a special bundle $t^\ast\in \{1,\ldots,m\}$ such that:
\begin{gather*} 
	v(x)=\begin{cases}
		0 \quad& x=0, \\
		m^{15}+ v(x-1) \quad& 0<x<t^\ast, \\
		v'(m-x+1)-v'(m-x)+(-1)^{sn}\cdot \frac{1}{8m^2}+v(x-1) \quad& x=t^\ast, \\
		v(x-1) \quad& x>t^\ast.
	\end{cases}
\end{gather*}

It is easy to see that all the valuations in all three families 
are normalized, monotone and have decreasing marginal utilities. Also, the value of each bundle can be represented with $\mathcal{O}(\log m)$ bits.  
We begin with a simple observation regarding the properties of welfare maximizing allocations. 
\begin{lemma}\label{onlylemma}
	Let $v_A,v_B:[m]\to\mathbb{R}_{+}$ be multi-unit valuations with decreasing marginal values. Suppose that $s\in \{1,\ldots,m-1\}$ is a number of items such that:
	\begin{enumerate}
		\item $v_B(m-s)-v_B(m-s-1)>v_A(s+1)-v_A(s)$.
		\label{assum1}
		\item $v_A(s)-v_A(s-1)>v_B(m-s+1)-v_B(m-s)$.
		\label{assum2}
	\end{enumerate}
	\emph{Then,} $(s,m-s)$ is the \emph{unique} welfare maximizing allocation. 
If $v_B(m)-v_B(m-1)>v_A(1)-v_A(0)$, then the unique welfare maximizing allocation is $(0,m)$.
Equivalently, $v_A(m)-v_A(m-1)>v_B(1)-v_B(0)$ implies that  the only welfare maximizing allocation is $(m,0)$. 
\end{lemma}

\begin{proof}
	Let $(s,m-s)$ be an allocation such that inequalities \ref{assum1} and \ref{assum2} hold. 
	Consider an allocation $(s+t,m-s-t)$ where $t>0$. Note that Alice gets $t$ more items, so due to the property of diminishing marginal utilities, her value increases by at most $t\cdot [v_A(s+1)-v_A(s)]$. Bob gets $t$ items less, so his value decreases by at least $t\cdot [v(m-s)-v(m-s-1)]$. By inequality \ref{assum1}, we get that the welfare of $(s+t,m-s-t)$ is strictly smaller than the welfare of $(s,m-s)$. 
	
	For the other direction, consider an allocation $(s-t,m-s+t)$ where $t>0$. This time, Bob gets $t$ more items, so his value increases by at most $t\cdot [v(m-s+1)-v(m-s)]$. Alice gets $t$ items less, so her value decreases by at least $t\cdot [v(s)-v(s-1)]$. Similarly, it implies that the welfare of $(s-t,m-s+t)$ is strictly smaller than the welfare of $(s,m-s)$. 
	
	The proof for the second and third part of the lemma is identical.          
\end{proof}



It is easy to see that the following two propositions together imply Theorem \ref{newmainthm}.

\begin{proposition}\label{dsicsketchprop}
	Let $\mathcal{M}$ be a normalized mechanism with $c$ bits that implements in dominant strategies a welfare maximizer for a multi-unit auction
	where the valuations have decreasing marginal utilities and the value of a bundle can be represented with $\mathcal{O}(\log m)$ bits. \emph{Then},
	there exists $\gamma\in \{1,m^5\}$ such that every element of $V^{ND,\gamma}$ can be represented with at most $c+\mathcal{O}(\log (m))$ bits.	
\end{proposition}

\begin{proposition}\label{onewayhardmuaclaim}
	For every $\gamma \in \{1,\ldots,m^5\}$, the representation size of a valuation in $V^{ND,\gamma}$ is $\Omega(m\log (m))$.
\end{proposition}

\begin{proof}[Proof of Proposition \ref{onewayhardmuaclaim}]
	By definition, for every $\gamma \in \{1,\ldots,m^5\}$, $|V^{ND,\gamma}|=2\cdot (m+1)^{m-2}$. Thus, by the pigeonhole principle, the representation size of an element in $V^{ND,\gamma}$ is  $\Omega(m\log (m))$ bits. 
\end{proof}
\subsection{Proof of Proposition 
\ref{dsicsketchprop}}
Fix a dominant strategy normalized two-player mechanism $\mathcal{M},\mathcal{S}_A,\mathcal{S}_B$ that implements a welfare-maximizer $f$\footnote{There is more than one welfare-maximizer due to tie breaking.} with payment schemes $P_A,P_B$ for a multi-unit auction where the valuations have decreasing marginal utilities and the value of a bundle can be represented with $\mathcal{O}(\log m)$ bits. Observe that $\mathcal M$ is in particular dominant strategy when the domain of each player is $V^{D}\cup V^{ND}\cup V^{P}$.  Denote with 
$c$ the communication complexity of the mechanism $\mathcal M$. 

Observe that $\mathcal M$ is incentive compatible, so by the taxation principle every valuation $v_A$ of Alice is associated with a menu of prices to Bob, such that for every valuation $v_B$ of Bob the action profile $(\mathcal S_A(v_A),\mathcal S_B(v_B))$ reaches a leaf that is labeled with  a profit-maximizing bundle given this menu. The same can be said of Bob's valuation and the menu presented to Alice.   

The proof idea is as follows. We begin by showing that the payments in the menu associated with a valuation are closely related to its values (Subsection \ref{subsec-payments-sketch}). In Section \ref{bobsayssomething-subsec}, we show that there exists a set of valuations of Bob such that he sends the price of some bundle (e.g., the price of one item), or otherwise Alice's strategy $\mathcal{S}_A$ is not dominant. 
Consider now two valuations $v_B,v'_B$ from this set that differ only in the price of $1$ item. Assume towards a contradiction that Alice has two valuations $v_A,v'_A$ with the same message such that the optimal solution in every one of the four possible instance is $(s,m-s)$ but $P_B(m-s,v_A)\neq P_B(m-s,v_A')$ . In this case, the worry is that Alice can determine Bob's payment to be either $P_B(m-s,v_A)$ or $P_B(m-s,v_A')$ without changing Bob's allocation, based only on the price that $v_B,v_B'$ display for one item.  Thus, Bob will not have a dominant strategy in this case unless Alice commits on the price she displays for $m-s$ items (Subsection \ref{alicesayspayment-sec}). However, if this happens for too many bundles, we can reconstruct Alice's valuation from her first message (Subsection \ref{sketch-subsec}).

\subsubsection{Payments are Good Sketches}\label{subsec-payments-sketch}
We now prove that the payments in the menu that each player presents to the other player are tightly related to the valuation. 


\begin{lemma}\label{lemma-approx-uniqueness}
	Let $v_A\in V^{ND}\cup V^{D}$ and let $x\in \{1,\ldots,m-1\}$ be a number of items. Then: 
	\begin{equation*}\label{eq-alice-valpay}
	P_B(x,v_A)\in \big[v_A(m)-v_A(m-x)\pm\frac{1}{8m}
	\big]
	\end{equation*}
	where $P_B(x,v_A)$ is the price of $x$ items presented to Bob when Alice has the valuation $v_A$.
	Similarly,  every valuation of Bob $v_B\in V^{ND}\cup V^{D}$ and every $x\in 
	\{1,\ldots,m-1\}$ satisfy that:
	\begin{equation*}
	P_A(x,v_B)\in \big[v_B(m)-v_B(m-x)\pm\frac{1}{8m}\big]
	\end{equation*}
\end{lemma}	
We defer the proof of Lemma \ref{lemma-approx-uniqueness} to Subsection \ref{subsec-approx-uniqueness}.
\begin{corollary}\label{cor-approx-payments}
	Fix $v_A\in V^{ND}\cup V^{D}$ and a number of items  $x\in \{1,\ldots,m-1\}$. 
	Given every $P_B(m-x,v_A)$ and $v_A(m)$, the exact value of $v_A(x)$ can be deduced. 
\end{corollary}
\begin{proof}
	By Lemma \ref{lemma-approx-uniqueness}, we have that 
	$
	v_A(x)\in [v_A(m)-P_B(m-x,v_A)\pm \frac{1}{8m}]
	$. 
	Thus, given $P_B(m-x,v_A)$ and $v_A(m)$, we can construct an interval of size $\frac{1+1}{8m}\le \frac{1}{4}$ such that $v_A(x)$ belongs in it. Recall that $x\le m$ so by definition $v_A(x)$ is an integer and an interval of size at most $\frac{1}{4}$ has only one integer in it, so we can immediately identify it.  
\end{proof}





\subsubsection{Bob Reveals Information That Does Not Affect the Allocation}
 \label{bobsayssomething-subsec}
 From now on, we focus on the following subsets of valuation sets of Alice and Bob: 
\begin{equation*}\label{eq-valuations}
	V_A=V_B=\big\{\bigcup\limits_{\gamma=1}^{m^5} V^{ND,\gamma} \big\} \bigcup \big \{\bigcup\limits_{\gamma=1}^{m^5} V^{D,\gamma} 
	\big \}
\end{equation*}
Observe that the mechanism $\mathcal{M}$ together with the strategies $\mathcal{S}_A,\mathcal{S}_B$ is also a dominant strategy implementation of $f,P_A,P_B$ with respect to $V_A\times V_B$,
since they have decreasing marginal values and the value of a bundle can be described with $\mathcal{O}(\log m)$ bits. 
By Lemma \ref{minimal-lemma}, given the valuations $V_A\times V_B$  there exists a \emph{minimal} dominant strategy mechanism $\mathcal{M}'$ with strategies $(\mathcal S_A',\mathcal{S}_B')$ that realize the welfare-maximizer $f$ with payment schemes $P_A,P_B$ with $c'\le c$ bits. 
 
 We remind that throughout the proof we slightly abuse notation: we say that a player with valuation $v$ sends a message $z$ at vertex $r$ instead of saying that the dominant strategy of the player is to send message $z$ given the valuation 
$v$.
We also use the notations $V^\gamma$, $V^{\le \gamma}$ or $V^{\ge \gamma}$ to denote all the valuations in $V_A$ or $V_B$ with weight $\gamma$, or the valuations with a weight which is smaller or larger than $\gamma$.  Note that all these three sets do not include valuations for $V^{P}$. 

Observe that since $\mathcal M$ is minimal,  there exists a player, without loss of generality Alice, that sends different messages for different valuations at the root vertex of the protocol, which we denote with $r$. The reason for that is that $\mathcal M'$ is minimal and there exist $(v_A,v_B),(v_A',v_B')\in V_A\times V_B$ such that the optimal allocation for them differs.  
We will show that since she sends non-trivial message in the first round, she has a dominant strategy in   $\mathcal{M}'$ only
if Bob discloses very specific information that, in certain situations, does not affect the allocation. Formally:

\begin{claim} \label{claimbobsayspayment}
	One of the two conditions below necessarily holds:
	\begin{enumerate}
		\item For every $v_B^1,v_B^2\in V_B^{D,\gamma=m^5}$ such that $P_A(1,v_B^1)\neq P_A(1,v_B^2)$, Bob sends different messages at vertex $r$. \label{condicondi1}
		\item For every $v_B^1,v_B^2\in V_B^{D,\gamma=1}$ such that $P_A(m-1,v_B^1)\neq P_A(m-1,v_B^2)$, Bob sends different messages at vertex $r$. \label{condicondi2}
	\end{enumerate}
\end{claim}
For the proof of Claim \ref{claimbobsayspayment}, we
 prove the following lemma, which is the main working horse of this subsection:
\begin{lemma} \label{unite-big-small-lemma}
Let $v_A^{1},v_A^{2}$ be two valuations  of Alice, and let $v_B^1,v_B^2$ be two valuations of Bob such that: 
\begin{enumerate}
	\item The unique optimal solution for the instances $(v_A^1,v_B^1)$ and $(v_A^2,v_B^2)$ is $(x,m-x)$.
	\item $P_A(x,v_B^1) \neq P_A(x,v_B^2)$.
	\item Alice sends different messages at the root vertex $r$ for $v_A^1$ and $v_A^2$.  
\end{enumerate} 
\emph{Then,} Bob  sends different messages at the root vertex $r$ for the valuations $v_B^1$ and $v_B^2$.
\end{lemma}
\begin{proof}
	Denote with $z_A^1$ and $z_A^2$ the messages that Alice sends for $v_A^1,v_A^2$.
	Assume towards a contradiction that Bob sends the same message $z_B$ for the valuations $v_B^1,v_B^2$ at the root vertex $r$. Let $t_1,t_2$ be the subtrees that the message profiles $(z_A^1,z_B)$ and $(z_A^2,z_B)$ lead to. 
	Denote with $l_1,l_2$ the leaves that $(v_A^1,v_B^1)$ and $(v_A^2,v_B^2)$ reach (respectively). 
	For an illustration, see Figure \ref{new-figure}.
	\begin{figure} [H] 
		\centering
		
		\begin{tikzpicture}
			\node[shape=circle,draw=black,minimum size=0.8cm] (r) at (1.5,1.5) {$r$};
			\node[shape=circle,draw=black,minimum size=0.8cm] (v) at (0,0) {};
			\node[shape=circle,draw=black,minimum size=0.8cm,fill=pink!90] (l) at (-1,-1) {$l_1$};
			\node[] (bundle) at (-1,-1.7) {\footnotesize	$x,P_A(x,v_B^1)$};
			\node[shape=circle,draw=black,minimum size=0.8cm] (v') at (3,0) {}; 
			\node[shape=circle,draw=black,minimum size=0.8cm,fill=blue!30] (l') at (4,-1) {$l_2$};
			\node[] (bundle') at (4,-1.7) {\footnotesize	$x,P_A(x,v_B^2)$}; 
			\node[shape=circle,draw=black,minimum size=0.8cm] (n') at (2,-1) {}; 
			\node[shape=circle,draw=black,minimum size=0.8cm] (n) at (1,-1) {}; 
			
			\draw[thick,red!60,dashed] (-2,0.5) rectangle (1.45,-2) {};
			\node[black,font=\itshape,text=red!85] at (-0.3,-2.3) {subtree $t_1$};
			
			\draw[thick,blue!60,anchor=mid west,dashed] (1.55,0.5) rectangle (5,-2) {};
			\node[black,font=\itshape,text=blue!85] at (3.4,-2.3) {subtree $t_2$};

			\draw [->] (r) edge  node[sloped, above] {\footnotesize$z_A^1$} (v);
			\draw [->] (r) edge  node[sloped, above] {\footnotesize$z_A^2$} (v');
			\draw [->] (v) edge[dotted]  node[sloped, above] {}  (l);
			\draw [->] (v') edge[dotted]  node[sloped, above] {} (l');
			\draw [->] (v) edge[dotted]  node[sloped, above] {} (n);
			\draw [->] (v') edge[dotted]  node[sloped, above] {} (n');
			
		\end{tikzpicture}
		\caption{An illustration for the proof of Lemma \ref{unite-big-small-lemma}. It describes two subtrees $t_1,t_2$ in the tree that the message $z_B$ of Bob induces for Alice at the root vertex $r$. The leaves $l_1,l_2$ are the leaves that $(v_A^1,v_B^1)$ and $(v_A^2,v_B^2)$ reach, so by assumption they are labeled with the allocation $x$ items for Alice with a price of $P_A(x,v_B^1)$  and $P_A(x,v_B^2)$, respectively.
		}
		\label{new-figure}
		
	\end{figure}
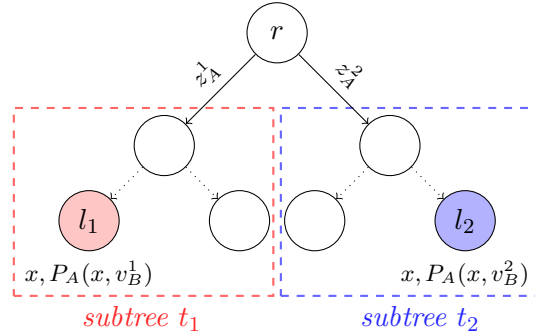
	
	Note that the leaf $l_1$ is labeled with the allocation $(x,m-x)$ and with the payment $P_A(x,v_B^1)$ for Alice, and similarly the leaf $l_2$ is labeled with the allocation $(x,m-x)$ and with the payment $P_A(x,v_B^2)$ for Alice. 
	Observe that $l_1,l_2$ appear in different subtrees $t_1,t_2$, so by Lemma \ref{minimal-lemma}, they are labeled with the same payment for Alice. However, $P_A(x,v_B^1) \neq P_A(x,v_B^2)$ by assumption, so we reach a contradiction.

\end{proof}	

The following two lemmas are immediate corollaries of Lemma \ref{unite-big-small-lemma}:
\begin{lemma}\label{bigcorollary}
	Assume that there exist two valuations
	$v_A^1,v_A^2\in V_A^{\le m^2}$ that Alice sends different messages for at the root vertex $r$. 
	Let $v_B^1,v_B^2\in V_B^{D,\gamma=m^5}$ be two semi-decisive valuations of Bob such that $P_A(1,v_B^1)\neq P_A(1,v_B^2)$.  \emph{Then}, Bob sends different messages at the root vertex $r$ for the valuations $v_B^1$ and $v_B^2$. 
\end{lemma}
\begin{proof}
	We begin by showing that for every $v_A\in V_A^{\le m^2}$ and for every $v_B\in V_B^{D,\gamma=m^5}$, the unique optimal allocation is $(1,m-1)$. By Lemma \ref{onlylemma}, it suffices to prove the inequalities $v_B(m-1)-v_B(m-2)>v_A(2)-v_A(1)$ and $v_A(1)-v_A(0)>v_B(m)-v_B(m-1)$, which hold by definition: 
	\begin{gather*}
	v_B(m-1)-v_B(m-2) \ge m^5 > m^2\cdot (m^2-m+1) \ge v_A(2)-v_A(1) \\
		\implies v_B(m-1)-v_B(m-2)> v_A(2)-v_A(1)  \\
	v_A(1)-v_A(0)\ge 1\cdot 3m^8 >  1 \ge v_B(m)-v_B(m-1) 
	\implies v_A(1)-v_A(0)> v_B(m)-v_B(m-1)   
	\end{gather*}
	Thus, the unique optimal allocation for the instances  $(v_A^1,v_B^1),(v_A^2,v_B^2)$ is $(1,m-1)$.
	Recall that by assumption Alice sends different messages for $v_A^1,v_A^2$ and that $P_A(1,v_B^1)\neq P_A(1,v_B^2)$, so by Lemma \ref{unite-big-small-lemma}, Bob sends different messages at the root vertex for $v_B^1$ and $v_B^2$, as needed.  
\end{proof}

\begin{lemma}\label{smallcorollary}
	Assume that there exist two valuations
	$v_A^1,v_A^2\in V_A^{\ge m^2}$ that Alice sends different messages for at the root vertex $r$. 
	Let $v_B^1,v_B^2\in V_B^{D,\gamma=1}$ be two semi-decisive valuations of Bob such that $P_A(m-1,v_B^1)\neq P_A(m-1,v_B^2)$.  \emph{Then}, Bob sends different messages at the root vertex $r$ for the valuations $v_B^1$ and $v_B^2$. 
\end{lemma}
\begin{proof}
	We begin by showing that for every $v_A\in V_A^{\ge m^2}$ and for every $v_B\in V_B^{D,\gamma=1}$, the unique optimal allocation is $(m-1,1)$. 
		By Lemma \ref{onlylemma}, it suffices to prove the inequalities $v_A(m-1)-v_A(m-2)>v_B(2)-v_B(1)$ and $v_B(1)-v_B(0)>v_A(m)-v_A(m-1)$ that hold by definition:
	\begin{gather*}
		v_A(m-1)-v_A(m-2)\ge m^{2}  > 1\cdot (m^2-m+1) \ge v_B(2)-v_B(1) \\
		\implies
		v_A(m-1)-v_A(m-2)> v_B(2)-v_B(1)          \\
		v_B(1)-v_B(0)\ge 1\cdot 3m^8 >  1 \ge  v_A(m)-v_A(m-1) 
		\implies
		v_B(1)-v_B(0)> v_A(m)-v_A(m-1)          
	\end{gather*}
Thus, the optimal allocation for the instances  $(v_A^1,v_B^1),(v_A^2,v_B^2)$ is $(m-1,1)$.
Recall that by assumption Alice sends different messages for $v_A^1,v_A^2$ and that $P_A(m-1,v_B^1)\neq P_A(m-1,v_B^2)$, so by Lemma \ref{unite-big-small-lemma}, Bob sends different messages at the root vertex for $v_B^1,v_B^2$, as needed.
\end{proof}	
We can now prove Claim \ref{claimbobsayspayment}: 
\begin{proof}[Proof of Claim \ref{claimbobsayspayment}]
		Recall that we have assumed (without loss of generality) that there exist two valuations of Alice that she sends different messages for at the root vertex $r$. It implies that the mechanism $\mathcal{M}'$ satisfies at least one of the following conditions: either Alice  sends different messages for two valuations in $V_A^{\le m^2}$ or she sends different messages for two valuations in $V_A^{\ge m^2}$ (otherwise, she sends the same message for all valuations in $V_A$, since  $V_A^{\le m^2}$ and $V_A^{\ge m^2}$ are intersecting and $V_A=V_A^{\le m^2}\cup V_A^{\ge m^2}$). 
		
		If she sends different messages for two valuations in $V_A^{\le m^2}$ at the root vertex $r$, by Lemma \ref{bigcorollary}, we get that for every $v_B^1,v_B^2\in V_B^{D,\gamma=m^5}$ such that $P_A(1,v_B^1)\neq P_A(1,v_B^2)$, Bob sends different messages at vertex $r$. Similarly, if she sends different messages for two valuations in $V_A^{\ge m^2}$, then by applying Lemma \ref{smallcorollary} we have that for every $v_B^1,v_B^2\in V_B^{D,\gamma=1}$ with $P_A(m-1,v_B^1)\neq P_A(m-1,v_B^2)$, Bob sends different messages at vertex $r$.    
%
%
%
%
%
\end{proof}
\subsubsection{Alice Commits to Bob's Payment} \label{alicesayspayment-sec}
We now use the information revealed by Bob about the semi-decisive valuations in $V_B^{\gamma=1}$ or in $V_B^{\gamma=m^5}$ to show that there exists \textquote{large} set of valuations such that Alice has to commit to Bob's payment for every possible allocation in the first round of the mechanism.  
 In Subsection \ref{sketch-subsec}, we will show how to use the payment to reconstruct these valuations. 
 
 Observe that we now use the fact that $\mathcal{M}'$ is dominant strategies for Bob. 
For the statement of the claim, we define $v_{m-s}^{\gamma}\in V^{D,\gamma}$ as the semi-decisive valuation parameterized with weight $\gamma$, the special bundle $x^\ast=m-s$ and the margin $d_m=\frac{1}{2}$.
\begin{claim}\label{muamainclaim}
	The following holds for either $\gamma=1$ or for $\gamma=m^5$. Let $v_A\in V_A^{ND,\gamma}$ be a valuation, and let $z_A$ be the message that Alice sends for it at the root of the protocol. 
	Fix a number of items $s\in \{2,\ldots,m-2\}$ and let $z_B$ be the message that Bob sends at the root if his valuation is the semi-decisive valuation $v_{m-s}^{\gamma}$ defined above. 
	Denote with $t$ the subtree that the message profile $(z_A,z_B)$ leads to. 
	\emph{Then:}  
	\begin{enumerate}
		\item There exists a leaf at subtree $t$ labeled with the allocation $(s,m-s)$. 
		\item Every leaf at subtree $t$  that is labeled with the allocation $(s,m-s)$ satisfies that it is labeled with the payment $P_B(m-s,v_A)$ for Bob. 
	\end{enumerate}
\end{claim}

\begin{proof}
	We show that condition \ref{condicondi1} of Claim \ref{claimbobsayspayment} implies that  Claim \ref{muamainclaim} holds for $\gamma=m^5$. The proof that condition \ref{condicondi2} of Claim \ref{claimbobsayspayment} implies that Claim \ref{muamainclaim} holds for $\gamma=1$ is analogous. Claim \ref{muamainclaim} follows since by Claim \ref{claimbobsayspayment} at least one of the conditions specified in the statement of Claim \ref{claimbobsayspayment} holds. 
	
	Assume that condition \ref{condicondi1} holds. Let $v_A\in V_A^{ND,\gamma=m^5}$ be a valuation, and let $s \in \{2,\ldots,m-2\}$ be a number of items. 
	Define $v_B,v_B'\in V_B^{D,\gamma=m^5}$ as follows.  
	\begin{gather*}
	v_B=	v_{m-s}^{m^5},\quad
	v_B'(x)=	\begin{cases}
	0 \quad&x=0, \\
	m^5\cdot m^8 \quad&x=1, \\
	m^5\cdot (m^2-m+1)+v'(x-1) \quad& x\in \{2,\ldots,m-s\}, \\
	m^5+v'(x-1) \quad&x\in \{m-s+1,\ldots,m-1\}, \\
	v'(m-1)+ 1
	\quad&x=m.
	\end{cases} 
	\end{gather*}
In words, $v_B$ and $v_B'$ are the two decisive valuations with weight $\gamma=m^5$ and special bundle $x^\ast=m-s$. Note that the only difference between $v_B,v_B'$ is the marginal value of the $m'$th item.   

	We begin by explaining why the unique welfare maximizing allocation for the instance $(v_A,v_B)$ is $(s,m-s)$.
%
	By Lemma \ref{onlylemma}, it suffices to prove that: 
	\begin{gather*}
		v_B(m-s)-v_B(m-s-1)= m^5\cdot (m^2-m+1)>m^5\cdot (m^2-m) \ge  v_A(2)-v_A(1)\ge v_A(s+1)-v_A(s)  \\
		v_A(s)-v_A(s-1) \ge v_A(m-1)-v_A(m-2) \ge  m^5 \cdot m > m^5 \ge v_B(m-s+1)-v_B(m-s)
	\end{gather*}
Thus, the leaf $l$ that $(v_A,v_B=v_{m-s}^{m^5})$ reaches is labeled with the allocation $(s,m-s)$. By definition, this leaf belongs in the subtree $t$, so we have part $1$ of the claim.  
For the proof of the second part, recall that by Lemma \ref{lemma-approx-uniqueness} we have that: 
\begin{equation*}\label{eq-payment-range}
P_A(1,v_B)\le v_B(m)-v_B(m-1)+\frac{1}{8m},\quad 
P_A(1,v_B')\ge v_B'(m)-v_B'(m-1)-\frac{1}{8m}
\end{equation*}
Therefore: 
\begin{multline*}
P_A(1,v_B)\le v_B(m)-v_B(m-1)+\frac{1}{8m} < v_B'(m)-v_B'(m-1) - \frac{1}{8m}\le P_A(1,v_B')  \\ \implies P_A(1,v_B) < P_A(1,v_B')
\end{multline*}
where the strict inequality holds because $v_B'(m)-v_B(m)=\frac{1}{2}$ and $v_B'(m-1)=v_B(m-1)$.
Therefore, by condition  \ref{condicondi1} of Claim \ref{claimbobsayspayment} we have that 
Bob sends different message $z_B'$ for $v_B'$ than the message $z_B$ he sends for $v_B$ at vertex $r$. 

	Denote with $t'$ the subtree that the messages $(z_A,z_B')$ lead to, and denote the leaf in $t'$ that $(v_A,v_B')$ reaches with $l'$. 
	Since $v_B$ and $v_B'$ are equal for all the coordinates in $\{1,\ldots,m-1\}$, we have that the unique welfare-maximizing allocation for $(v_A,v_B')$ is also $(s,m-s)$, so $l'$ is labeled with it. For an illustration, see Figure \ref{mua-figure}.  
	
%
%
%
	
			Since the mechanism $\mathcal M'$ realizes the welfare-maximizer $f$ with the payment schemes $P_A,P_B$, we have that the leaf $l$ that is labeled with the allocation 
	$(s,m-s)$ is labeled with the payment $P_B(m-s,v_A)$ for Bob.  By Lemma \ref{lemma-known-prices} all the leaves in $t$ and in $t'$ that are labeled with the allocation $(s,m-s)$ have the same price for Bob.
	By combining these two facts, we get that all the leaves in the subtree $t$ labeled with the allocation $(s,m-s)$ are labeled with  the payment $P_B(m-s,v_A)$ for Bob, which completes the proof. 
\end{proof}
	
		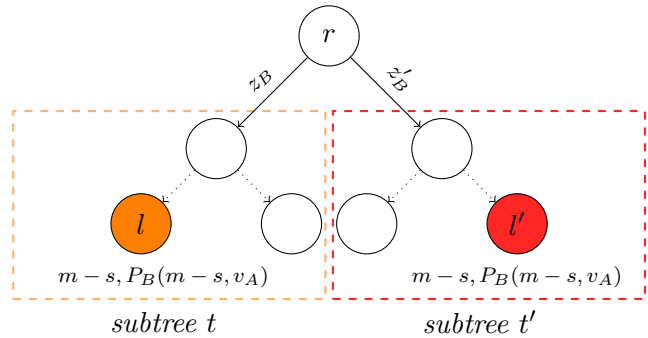
\begin{figure}[H]
		\centering
		
		\begin{tikzpicture}
			\node[shape=circle,draw=black,minimum size=0.8cm] (r) at (1.5,1.5) {$r$};
			\node[shape=circle,draw=black,minimum size=0.8cm] (v) at (0,0) {};
			\node[shape=circle,draw=black,minimum size=0.8cm,fill=orange] (l) at (-1,-1) {$l$};
			\node[] (bundle) at (-0.7,-1.7) {\scriptsize	$m-s, P_B(m-s,v_A)$};
			\node[shape=circle,draw=black,minimum size=0.8cm] (v') at (3,0) {}; 
			\node[shape=circle,draw=black,minimum size=0.8cm,fill=red!85] (l') at (4,-1) {$l'$};
			
			\node[] (bundle') at (4,-1.7) {\scriptsize $m-s, P_B(m-s,v_A)$}; 
			\node[shape=circle,draw=black,minimum size=0.8cm] (n') at (2,-1) {}; 
			\node[shape=circle,draw=black,minimum size=0.8cm] (n) at (1,-1) {}; 
			
			\draw[thick,orange!60,dashed] (-2.7,0.5) rectangle (1.45,-2) {};
			\node[black,font=\itshape] at (-0.7,-2.3) {subtree $t$};
			
			\draw[thick,red!85,anchor=mid west,dashed] (1.55,0.5) rectangle (5.7,-2) {};
			\node[black,font=\itshape] at (3.5,-2.3) {subtree $t'$};

			\draw [->] (r) edge  node[sloped, above] {\footnotesize$z_B$} (v);
			\draw [->] (r) edge  node[sloped, above] {\footnotesize$z_B'$} (v');
			\draw [->] (v) edge[dotted]  node[sloped, above] {}  (l);
			\draw [->] (v') edge[dotted]  node[sloped, above] {} (l');
			\draw [->] (v) edge[dotted]  node[sloped, above] {} (n);
			\draw [->] (v') edge[dotted]  node[sloped, above] {} (n');
			
		\end{tikzpicture}
		\caption{An illustration for the proof of Claim \ref{muamainclaim}. It describes the subtrees $t,t'$ in the tree that the message $z_A$ of Alice induces for Bob at the root vertex $r$. The leaves $l,l'$ are the leaves that $(v_A,v_B)$ and $(v_A,v_B')$ reach, so as we prove they are labeled with the allocation $(s,m-s)$. 
		}
		\label{mua-figure}
		
	\end{figure}

%

\subsubsection{Reconstructing Alice's Valuation} \label{sketch-subsec}
We can now complete the proof of Proposition \ref{dsicsketchprop}.	
Let $\gamma\in\{1,m^5\}$ be the scalar that Claim \ref{muamainclaim} holds for.
We will show how to represent every valuation in $V^{ND,\gamma}$ with at most $c'+\mathcal{O}(\log (m))\le c+ \mathcal{O}(\log (m))$ bits (we remind that $c,c'$ stand for the communication complexity of the mechanisms $\mathcal{M},\mathcal{M}'$).

The representation of a valuation $v$ is composed of the values $v(1),v(m-1),v(m)$ and the message $z_A$ Alice sends at the root vertex $r$ given the valuation $v$. For every number of items $s\in[m]$, we show how to compute $v(s)$ without any additional communication. 

$v(1),v(m-1)$ and $v(m)$ are specified in the sketch. Let $s \in \{2,\ldots,m-2\}$. Let $z_B$ be the message that Bob sends at the root vertex $r$ when his valuation is the decisive valuation $v_B=v_{m-s}^\gamma$. Let $\ell$ be an arbitrary leaf in the subtree that $(z_A,z_B)$ leads to that is labeled with the allocation $(s,m-s)$. By Claim \ref{muamainclaim}, such a leaf exists and it is labeled with the payment $P_B(m-s,v_A)$ for Bob. 
%
Recall that $v(m)$ is included in the representation, so by Corollary \ref{cor-approx-payments} we can extract $v(s)$.
\subsubsection{Proof of Lemma \ref{lemma-approx-uniqueness}} \label{subsec-approx-uniqueness}
We prove Lemma \ref{lemma-approx-uniqueness} for the valuations of Alice and the payment scheme of Bob. The proof for Bob's valuations and Alice's payment scheme is identical. 

Fix a valuation $v_A$ and number of items $x\in \{1,\ldots,m-1\}$.  
We begin by  showing that for every $1\le y\le x$, the payment of Bob satisfies that:
\begin{equation*}
P_B(y,v_A)-P_B(y-1,v_A) \in \big[v_A(m-y+1)-v_A(m-y)\pm \frac{1}{8m^2}\big] 
\end{equation*}
Fix $1\le y\le x$. Observe the valuation $v_{0}\in V^{P}$ that is parameterized with the valuation $v_A$, with the special bundle $t^\ast=y$ and with the sign $sn=0$:
\begin{gather*} 
	v_0(x)=\begin{cases}
		0 \quad& x=0, \\
		m^{15}+ v_0(x-1) \quad& 0<x<y, \\
		v_A(m-y+1)-v_A(m-y)+ \frac{1}{8m^2}+v_0(x-1) \quad& x=y, \\
		v_0(x-1) \quad& x>y.
	\end{cases}
\end{gather*}
Note that if Alice's valuation is $v_A$ and Bob's valuation is $v_0$, the following inequalities hold:
\begin{gather*}
	v_0(y)-v_0(y-1)=v_A(m-y+1)-v_A(m-y)+\frac{1}{8m^2}> v_A(m-y+1)-v_A(m-y) \\
	v_0(y+1)-v_0(y)< 0<v_A(m-y)-v_A(m-y-1)
\end{gather*}
Therefore, by Lemma \ref{onlylemma}, the unique welfare maximizing allocation is that Alice wins $m-y$ items and Bob wins $y$ items. We remind that $\mathcal{M}$ is ex-post incentive compatible (since it is dominant strategy incentive compatible), and that it realizes a welfare-maximizer with the payment schemes $P_A,P_B$, so:  
\begin{multline}\label{eq-payment-upper}
v_0(y)-P_B(y,v_A)\ge v_0(y-1)-P_B(y-1,v_A) \\ \implies
v_A(m-y+1)-v_A(m-y)+\frac{1}{8m^2}=v_0(y)-v_0(y-1)\ge
P_B(y,v_A) - P_B(y-1,v_A)  
\end{multline}
To prove a lower bound on $P_B(y,v_A)-P_B(y-1,v_A)$, we construct the valuation $v_1\in V^{P}$ which is parameterized with the valuation $v_A$, the special bundle $t^\ast=y$ and $sn=1$:
\begin{gather*} 
	v_1(x)=\begin{cases}
		0 \quad& x=0, \\
		m^{15}+ v_1(x-1) \quad& 0<x<y, \\
		v_A(m-y+1)-v_A(m-y)- \frac{1}{8m^2}+v_1(x-1) \quad& x=y, \\
		v_1(x-1) \quad& x>y.
	\end{cases}
\end{gather*}
Observe that the following inequalities hold for every $2\le y\le x$: 
\begin{gather}
	v_1(y)-v_1(y-1)=v_A(m-y+1)-v_A(m-y)-\frac{1}{8m^2}<v_A(m-y+1)-v_A(m-y) \label{eqeq} \\
	v_1(y-1)-v_1(y-2)=m^{15}>v_A(m-y+2)-v_A(m-y+1) \nonumber
\end{gather}
Note that for $y=1$, only inequality (\ref{eqeq}) holds. By Lemma \ref{onlylemma}, we get that welfare-maximizing allocation for every $1\le y\le m-1$ is $(y-1,m-y+1)$. Due to the same considerations as before, we have that: 
\begin{multline}\label{eq-payment-lower}
v_1(y-1)-P_B(y-1,v_A) \ge
v_1(y)-P_B(y,v_A) \\ \implies P_B(y,v_A)-P_B(y-1,v_A) \ge v_1(y)-v_1(y-1) = v_A(m-y+1)-v_A(m-y)-\frac{1}{8m^2}
\end{multline}
Combining (\ref{eq-payment-upper}) and (\ref{eq-payment-lower}) gives:
\begin{equation}\label{eq-partial-sums}
v_A(m-y+1)-v_A(m-y)-\frac{1}{8m^2}
\le P_B(y,v_A)-P_B(y-1,v_A) 
\le 
v_A(m-y+1)-v_A(m-y)+\frac{1}{8m^2}
\end{equation} 
We can now complete the proof. We remind that $\mathcal{M}$ is normalized so $P_B(0,v_A)=0$. 
Therefore, the following  telescopic sum equals $P_B(x,v_A)$:
\begin{align*}\label{eq-telescopic}
	P_B(x,v_A)&=P_B(x,v_A)-P_B(x-1,v_A)+P_B(x-1,v_A)-\ldots\\ &-P_B(1,v_A)  +P_B(1,v_A) - P_B(0,v_A) &\text{($\mathcal M$ is normalized, so $P_B(0,v_A)=0$)} \nonumber \\
	&= \sum_{y=1}^{x}  P_B(y,v_A)-P_B(y-1,v_A)  \nonumber 
\end{align*} 
Observe that by (\ref{eq-partial-sums}) we have that: 
\begin{align*}
P_B(x,v_A) &= \sum_{y=1}^{x} P_B(y,v_A)-P_B(y-1,v_A)\\ &\ge \sum_{y=1}^{x} \big[
v_A(m-y+1)-v_A(m-y)-\frac{1}{8m^2} \big] \\
&\ge v_A(m)-v_A(m-x) -\frac{1}{8m} 
\end{align*}
As needed. 
A similar analysis 
gives that $v_A(m)-v_A(m-x) +\frac{1}{8m}\ge P_B(x,v_A)$, which completes the proof. 
\subsection{An FPTAS for Multi-Unit Auctions with Decreasing Marginal Values - Proof of Theorem \ref{thm-mua-fptas}}\label{subsec-fptas}

In Section \ref{subsec-hardness} we showed that no mechanism finds the welfare maximizing allocation in dominant strategies and $poly(\log m)$ communication. In this section we show that this result is tight.

The mechanism is an adaptation of the maximal in range $2$-approximation algorithm for general multi unit auctions of \cite{DN07b}. A maximal in range algorithm (see \cite{DN07a},\cite{DN07b}) is an algorithm that finds the welfare maximizing solution in some pre-defined set of allocations. VCG payments are used to guarantee incentive compatibility.

Our maximal-in-range algorithm will split the items into $t=\frac {m} q$ bundles of size $q=\lfloor \frac {\eps\cdot m} {n^2} \rfloor$, and (possibly) one additional bundle of size $l=m-t\cdot q$. The maximal-in-range algorithm will optimally distribute these items among the bidders. We implement the algorithm by asking each bidder $i$ with valuation $v_i$ to send, simultaneously with the others, his values for all possible combinations of the bundles: $\{v_i(z\cdot q)\}_{z\leq t}$ and $\{v_i(z\cdot q+l)\}_{z\leq t}$.

It is clear that the number of value queries that the algorithm makes is $poly(n, \frac 1 \eps)$. In fact, the running time of the algorithm is also polynomial, the proof is essentially identical to that of \cite{DN07b}. The dominant strategy of each bidder is to send the true values, since this is a simultaneous maximal-in-range algorithm. It remains to prove the claimed approximation ratio. 

\begin{lemma}
The social welfare of the allocation that the algorithm outputs is at least $(1-\eps)\cdot OPT$.
\end{lemma}
\begin{proof}
We will show that there is an allocation in the range with social welfare at least $(1-\eps)\cdot OPT$. Since the algorithm is maximal-in-range, it must output a solution with at least that welfare.

Fix some optimal allocation of the items $(o_1,\ldots, o_n)$. Without loss of generality assume that all items are allocated: $\sum_io_i=m$. Thus, there must be some bidder, without loss of generality bidder $1$, such that $o_1 \geq m/n$.

For each $i>1$, obtain $o'_i$ by rounding up $o_i$ to the nearest multiple of $q$. Let $o_1'=m-\sum_{i>1}o'_i$. Note that this allocation is indeed in the range (each bidder $i>1$ gets a multiple of $q$, bidder $1$ gets the remaining bundles of size $q$ and the single bundle of size $l$).


We now analyze the social welfare of the allocation $(o'_1,\ldots, o'_n)$. By the monotonicity of the valuations, for each bidder $i'>1$ it holds that $v_i(o'_i)\geq v_i(o_i)$. As for bidder $1$, it holds that:
$o_1-o_1' =m-\sum_{i>1}o_i -m+\sum_{i>1}o'_i \le n\cdot q  = n \cdot \lfloor \frac {\eps\cdot m} {n^2} \rfloor \le \frac {\eps\cdot m} {n} $. Recall that $o_1\ge \frac{m}{n}$ and that $v_1$ exhibits decreasing marginal utilities, so by taking away at most $\epsilon$ fraction of the items of player $1$, his utility decreases by at most $\epsilon\cdot v_1(o_1)$. Thus, $v_1(o_1')\ge (1-\eps) \cdot v_1(o_1)$ and we have that 
$\sum_iv_i(o'_i)\geq (1-\eps)\cdot \sum_iv_i(o_i)$, as needed.
\end{proof}
\section{Combinatorial Auctions with Gross Substitutes Valuations}\label{sec-gs-impossibility}

Combinatorial auctions with gross substitutes valuations are another example for an important domain in which the welfare maximizing allocation can be efficiently found. In particular, Nisan and Segal \cite{NS06} show how to compute the optimal solution with $poly(m,n)$ communication if the value of  each bundle takes $poly(m)$ bits to represent. This implies that in this setting  the welfare maximizing allocation can be found by an ex-post incentive compatible mechanism that uses VCG payments with about the same communication complexity. We start with a definition of the class of gross substitutes valuations:

\begin{definition}
	A valuation $v:2^{M}\to \mathbb{R}$ satisfies the gross substitutes property if for every price vector $\vec p\in \mathbb{R}^{m}$ and
	for every $S \in \argmax_{S\subseteq M}\{v(T)-\sum_{j\in T}p_j\}$, if $\vec{p'}\ge \vec p$, then exists a bundle $S' \in \argmax_{T\subseteq M}\{v(T)-\sum_{j\in S}p'_j\}$ such that $S\cap \{j\hspace{0.25em}|\hspace{0.25em}p_j=p_j'\}\subseteq S'$.    
\end{definition}




The proof is very similar in structure and in spirit to  the proof of hardness of dominant strategy implementations for multi-unit auctions with decreasing marginal values (Subsection \ref{subsec-hardness}). We prove hardness for gross substitutes valuations since it is the largest set of valuations for which the exact optimum can be computed easily in ex-post equilibrium. However, the proof itself does not rely too much on the intricate definition and properties of gross substitutes. 

\subsection{The Hardness Result: Proof of Theorem \ref{gsmainthm}}
We will show that there exists a specific set of gross substitutes valuations where the value of a bundle can be represented with $poly(m)$ bits such that every dominant strategy implementation for them requires $\exp(m)$ bits.

Consider  a combinatorial auction of $m$  heterogeneous items (denoted with $M$) and two players (Alice and Bob), where $m\ge 3$.  
The valuations that we consider in the proof belong to three families: ``semi-decisive'' valuations $V^{ND}$, non-decisive valuations $V^{D}$ and another set of valuations $V^{P}$ that we will use to show that payments can serve as good sketches.  
Every semi-decisive and non-decisive valuation will have a ``weight'' which is a scalar $\gamma \in \{1,\ldots,m^5\}$ that captures its magnitude. In addition, we fix two items $a,b\in M$ as Alice's and Bob's special items.  

We now define the set of semi-decisive  valuations of Alice with  weight $\gamma$, denoted with $V_A^{D,\gamma}$. Every  $v\in V_A^{D,\gamma}$ is an additive valuation that has a subset of items $S \subseteq M \setminus \{a,b\}$ and a noise $\eta\in \{0,\frac{1}{2}\}$ such that:
$$
v(x)=\begin{cases}
	m^8 \quad & x=a, \\
	\eta \quad & x=b, \\
	\gamma\cdot (m+2) \quad & x\in S,\\
	\frac{\gamma}{2} \quad & \text{otherwise.} 
\end{cases}
$$
Similarly, every semi-decisive valuation of Bob $v\in V_B^{D,\gamma}$ is an additive valuation that has  a subset of items $S\subseteq M\setminus\{a,b\}$ and noise $\eta\in\{0,\frac{1}{2}\}$ such that:
$$
v(x)=\begin{cases}
	m^8 \quad & x=b, \\
	\eta \quad & x=a, \\
	\gamma\cdot (m+2) \quad & x\in S,\\
	\frac{\gamma}{2}  \quad & \text{otherwise.} 
\end{cases}
$$
We now define the set of non-decisive valuations of Alice and Bob. For this, we need to introduce the following notation. We use the notation $GS(M\setminus\{a,b\})$ for the set of gross substitutes valuations on the domain $2^{M\setminus\{a,b\}}$, where the value of each subset of $M\setminus\{a,b\}$ is in $\{0,\ldots,m\}$. 
We now define the set of non-decisive valuations of Alice with weight $\gamma$, i.e. $V_A^{ND,\gamma}$. Every $v\in V_A^{ND,\gamma}$ has a valuation $\tilde{v}\in GS(M\setminus\{a,b\})$ and a noise $\eta \in \{0,\frac{1}{2}\}$ such that:
 $$
v(S)=\begin{cases}
 	\gamma\cdot  \tilde{v}(S\setminus\{a,b\} + \gamma\cdot |S|+  m^8 + \eta \quad & \{a,b\} \subseteq S, \\	
 	\gamma\cdot \tilde{v}(S\setminus\{a\})+ \gamma\cdot |S|+ m^8 \quad & a\in S, \\
 	\gamma\cdot \tilde{v}(S\setminus\{b\})+ \gamma\cdot |S|+\eta  \quad & b\in S, \\
 	\gamma\cdot \tilde{v}(S)+ \gamma\cdot |S| \quad &\text{otherwise.}
 \end{cases} 	
 $$ 
 
 Similarly, every non-decisive valuation of Bob with weight $\gamma$, i.e. $v\in V_B^{ND,\gamma}$, has  a valuation $\tilde{v}\in GS(M\setminus\{a,b\})$ and a noise $\eta\in \{0,\frac{1}{2}\}$ such that:
$$
v(S)=\begin{cases}
	\gamma\cdot \tilde{v}(S\setminus\{a,b\})+ \gamma\cdot |S|+ m^8 + \eta \quad & \{a,b\} \subseteq S, \\	
	\gamma\cdot \tilde{v}(S\setminus\{b\})+ \gamma\cdot |S|+m^8 \quad & b\in S, \\
	\gamma\cdot \tilde{v}(S\setminus\{a\})+ \gamma\cdot |S|+\eta \quad & a\in S, \\
	\gamma\cdot \tilde{v}(S) \quad &\text{otherwise.}
\end{cases} 	
$$ 
Throughout the proof, we use the notations $V^{ND}=\bigcup\limits_{\gamma=1}^{m^5} V^{ND,\gamma}$ and $V^{D}=\bigcup\limits_{\gamma=1}^{m^5} V^{D,\gamma}$.  
We now define another set of valuations $V^{P}$ with the purpose of guaranteeing that different valuations in $V^{D}\cup V^{ND}$ induce different payments. 
Every $v\in V^{P}$ is an additive valuation that is parameterized with a valuation $v'\in V^{ND}\cup V^{D}$, a special bundle $S^\ast\subseteq M$, a special item $x^\ast \notin S^\ast$ and a sign $sn\in \{0,1\}$ such that:
$$
v(x)=\begin{cases}
	m^{15} \quad & x\in S^\ast, \\
	v'(M-S^\ast+\{x^\ast\})-v'(M-S^\ast) \pm (-1)^{sn}\cdot \frac{1}{8m^2} \quad & x=x^\ast, \\
	0  \quad & \text{otherwise.} 
\end{cases}
$$
Note that all the valuations in $V^{P},V^{D}$ are additive, so they are gross substitutes. We now explain why the valuations in $V^{ND}$ are gross substitutes as well. First,  gross substitute valuations are known to be closed under multiplication by a non-negative scalar and addition. Also,  by Lemma \ref{gslemma} below, they are also closed under the extension operation, which adds another item to the set of items, where the value of the item is additive. 
\begin{lemma}\label{gslemma}
	Given a valuation $v:2^{[S]}\to \mathbb{R}$, an additional item $x\notin S$ and a scalar $c\in \mathbb{R}^{+}$, we define an extended valuation $\tilde{v}:2^{S\cup \{x\}}\to \mathbb{R}$ as follows:
	$$
	\tilde{v}(S)=\begin{cases}
		v(S)+ c \quad & x\in S, \\	
		v(S) \quad & x\notin S,
	\end{cases} 	
	$$
	If $v$ is gross substitutes, then $\tilde{v}$ is also gross substitutes. 
\end{lemma}
\begin{proof}
	Note that there exists a bundle the contains the item $x$ is in the demand set of $\tilde{v}$ given the price vector $p$ if and only if $p_x\le c$. Therefore, it is easy to see that adding an additive value of $c$ to bundles that contain $x$ preserves the gross substitutes property.   
\end{proof}


The proof consists of the following  propositions that together imply Theorem \ref{gsmainthm}. 
\begin{proposition}\label{dsicsketchclaim-gs}
		Let $\mathcal{M}$ be a normalized mechanism with $c$ bits that implements in dominant strategies a welfare maximizer for a combinatorial auction with gross substitutes valuations, where the value of each bundle can be represented with $poly(m)$
 bits. \emph{Then}, every element of $GS(M\setminus\{a,b\})$ can be represented with $c+\mathcal{O}(\log m)$ bits.

\end{proposition}
\begin{proposition}\label{sketchhardclaim}
The description size of valuation in $GS(M\setminus\{a,b\})$ is $\exp(m)$ bits.
\end{proposition}
\begin{proof}
	By \cite{knuth1974asymptotic}, the number of matroid rank functions (which are a strict subset of  gross substitutes valuations) over a set of $m-2$ items is doubly exponential in $m$. Thus, by the pigeonhole principle, the description size of an item in $GS(M\setminus\{a,b\})$ is at least $\exp(m)$ bits.  
\end{proof}
\subsection{Proof of Proposition \ref{dsicsketchclaim-gs}}
The main ideas of the proof are the same as the  proof of Proposition \ref{dsicsketchprop}, with minor adaptations to the modified construction. We repeat the proof for completeness. 

Fix a normalized two-player mechanism $\mathcal{M}$ with strategies $\mathcal{S}_A,\mathcal{S}_B$ that realize in dominant strategies a welfare-maximizer $f$ with payment schemes $P_A,P_B$ for a combinatorial auction with gross substitutes valuations, where the value of a bundle can be represented with $poly(m)$ bits. 
Observe that $\mathcal M$ is in particular dominant strategy when the domain of each player is $V^{D}\cup V^{ND}\cup V^{P}$.  Denote with 
$c$ the communication complexity of the mechanism $\mathcal M$. 

Observe that $\mathcal M$ is incentive compatible, so by the taxation principle every valuation $v_A$ of Alice is associated with a menu of prices to Bob, such that for every valuation $v_B$ of Bob the action profile $(\mathcal S_A(v_A),\mathcal S_B(v_B))$ reaches a leaf that is labeled with  a profit-maximizing bundle given this menu. The same can be said of Bob's valuation and the menu that it presents to Alice.

The structure of the proof is as follows. We begin by showing that the prices in the menu that each valuation presents are closely related to its values (Subsection \ref{subsec-payments-sketch-gs}).
In Subsection \ref{gs-bobsayssomething-subsec}, we show that 
there exists a set of valuations of Bob such that his dominant strategy dictates that he sends  the price of some bundle (e.g., the price of the bundle of all items $M$) in the first round. 
Consider now two valuations $v_B,v'_B$ from this set that differ only in their price of bundles that contain the item $\{a\}$. 

Now, consider the case where Alice has two valuations $v_A,v'_A$ with the same message in the first round such that the optimal solution in every one of the four possible combinations of $v_A,v_A'$ and $v_B,v_B'$ is the allocation $(S+\{a\},M-S-\{a\})$ but $P_B(M-S-\{a\},v_A)\neq P_B(M-S-\{a\},v_A')$, where $S\subseteq M\setminus \{a,b\}$. In this case, the worry is that Alice can determine Bob's payment to be either $P_B(M-S-\{a\},v_A)$ or $ P_B(M-S-\{a\},v_A')$ without changing Bob's allocation, based only on the price of Bob for bundles that contain $\{a\}$.  Thus, Bob will not have a dominant strategy in this case unless Alice commits to her price for the bundle $,M-S-\{a\}$ (Section \ref{alicesayspayment-sec}). However, if this happens for too many bundles, we can reconstruct Alice's valuation from her first message, as we show in Subsection \ref{sketch-subsec}.


\subsubsection{Payments are Good Sketches}\label{subsec-payments-sketch-gs}
We now show the connection between payments and valuations. 
\begin{lemma}\label{lemma-approx-uniqueness-gs}
	Let $v_A\in V_A^{ND}\cup V_A^{D}$ and let $S\subseteq M$ be a non-empty bundle. Then: 
	\begin{equation*}\label{eq-alice-valpay-gs}
		P_B(S,v_A)\in \big[v_A(M)-v_A(M-S)\pm\frac{1}{8m}
		\big]
	\end{equation*}
	where $P_B(S,v_A)$ is the price of $S$ presented to Bob when Alice has the valuation $v_A$.
	Similarly,  every valuation of Bob $v_B\in V_B^{ND}\cup V_B^{D}$ and every non-empty bundle $S$  satisfy that:
	\begin{equation*}
		P_A(S,v_B)\in \big[v_B(M)-v_B(M-S)\pm\frac{1}{8m}\big]
	\end{equation*}
\end{lemma}	
We defer the proof of Lemma \ref{lemma-approx-uniqueness-gs} to Subsection \ref{subsec-sketch-proof-gs}.
\begin{corollary}\label{cor-approx-payments-gs}
	Fix $v_A\in V^{ND}\cup V^{D}$ and a bundle $S\neq M,\varnothing$.  
	Given every $P_B(M-S,v_A)$ and $v_A(M)$, the exact value of $v_A(S)$ can be deduced. 
\end{corollary}
\begin{proof}
	By Lemma \ref{lemma-approx-uniqueness}, we have that 
	$
	v_A(S)\in [v_A(M)-P_B(M-S,v_A)\pm \frac{1}{8m}]
	$. 
	Thus, given $P_B(M-S,v_A)$ and $v_A(M)$, we can construct an interval of size $\frac{1+1}{8m}\le \frac{1}{4}$ such that $v_A(S)$ belongs in it. Observe that by construction $v_A(S)$ is either an integer or an integer with addition of $\frac{1}{2}$, and an interval of size
at most $\frac{1}{4}$ has only one such number in it, so we can immediately identify it.  
\end{proof}

\subsubsection{Bob Reveals Information That Does Not Affect the Allocation} \label{gs-bobsayssomething-subsec}

 From now on, we focus on the following subsets of valuation sets of Alice and Bob: 
\begin{equation*}\label{eq-valuations}
	V_A=V_B=\big\{\bigcup\limits_{\gamma=1}^{m^5} V^{ND,\gamma} \big\} \bigcup \big \{\bigcup\limits_{\gamma=1}^{m^5} V^{D,\gamma} 
	\big \}
\end{equation*}
Observe that the mechanism $\mathcal{M}$ together with the strategies $\mathcal{S}_A,\mathcal{S}_B$ is also a dominant strategy implementation of $f,P_A,P_B$ with respect to $V_A\times V_B$,
since they have decreasing marginal values and the value of a bundle can be described with $\mathcal{O}(\log m)$ bits.  
By Lemma \ref{minimal-lemma}, given the valuations $V_A\times V_B$  there exists a \emph{minimal} dominant strategy mechanism $\mathcal{M}'$ with strategies $(\mathcal S_A',\mathcal{S}_B')$ that realize the welfare-maximizer $f$ with payment schemes $P_A,P_B$ with $c'\le c$ bits. 

We also use the notations $V^\gamma$, $V^{\le \gamma}$ or $V^{\ge \gamma}$ to denote all the valuations in $V_A$ or $V_B$ with weight $\gamma$, or the valuations with a weight which is smaller or larger than $\gamma$. Note that all these three sets do not include valuations from $V^{P}$.

Observe that since $\mathcal M$ is minimal, there exists a player, without loss of generality Alice, that sends different messages for different valuations at the root vertex of the protocol, which we denote with $r$. The reason for that is that $\mathcal M'$ is minimal and there exist $(v_A,v_B),(v_A',v_B')\in V_A\times V_B$ such that the optimal allocation for them differs.  
We will show that since she sends non-trivial message in the first round, she has a dominant strategy in   $\mathcal{M}'$ only
if Bob discloses very specific information that, in certain situations, does not affect the allocation. Formally:
\begin{claim} \label{gs-claimbobsayspayment}
	One of the two conditions below necessarily holds:
	\begin{enumerate}
		\item For every $v_B^1,v_B^2\in V_B^{D,\gamma=m^5}$ such that $P_A(\{a\},v_B^1)\neq P_A(\{a\},v_B^2)$, Bob sends different messages  at vertex $r$.\label{gs-condicondi1}
		\item For every $v_B^1,v_B^2\in V_B^{D,\gamma=1}$ such that $P_A(M\setminus\{b\},v_B^1)\neq P_A(M\setminus\{b\},v_B^2)$,
		 Bob sends different messages at vertex $r$. \label{gs-condicondi2}
	\end{enumerate}
\end{claim}
For the proof of Claim \ref{gs-claimbobsayspayment}, we state the following lemma, which is the main working horse of this subsection:
\begin{lemma} \label{gs-unite-big-small-lemma}
	Let $v_A^{1},v_A^{2}$ be two valuations  of Alice, and let $v_B^1,v_B^2$ be two valuations of Bob such that: 
	\begin{enumerate}
		\item The unique optimal solution for the instances $(v_A^1,v_B^1)$ and $(v_A^2,v_B^2)$ is $(X,M\setminus X)$.
		\item $P_A(X,v_B^1) \neq P_A(X,v_B^2)$.
		\item Alice sends different messages at the root vertex $r$ for $v_A^1$ and $v_A^2$.  
	\end{enumerate} 
	\emph{Then,} Bob  sends different messages at the root vertex $r$ for the valuations $v_B^1$ and $v_B^2$.
\end{lemma}
	The proof is identical to the proof of Lemma \ref{unite-big-small-lemma} and is omitted.
%
%
The following two lemmas are immediate corollaries of Lemma \ref{gs-unite-big-small-lemma}:
\begin{lemma}\label{gs-bigcorollary}
	Assume that there exist two valuations $v_A^1,v_A^2\in V_A^{\ge m^2}$ that Alice sends different messages for at the root vertex $r$. 
Let $v_B^1,v_B^2\in V_B^{D,\gamma=1}$ be two semi-decisive valuations of Bob such that $P_A(M\setminus \{b\},v_B^1)\neq P_A(M\setminus \{b\},v_B^2)$.  \emph{Then}, Bob sends different messages at the root vertex $r$ for the valuations $v_B^1$ and $v_B^2$. 
\end{lemma}
\begin{proof}
	We begin by showing that every $v_A\in V_A^{\ge m^2}$ and every $v_B\in V_B^{\gamma=1}$ satisfy that the unique optimal allocation is $(M\setminus \{b\},\{b\})$.  By definition, Alice always wins item $a$ and Bob always wins item $b$. Regarding the rest of the items, we have that the marginal utility of Alice for an item given any bundle that she already has is at least $m^2$, whereas Bob's marginal utility for every item is at most $m+2$. Therefore, Alice wins the rest of the items. 

Therefore,  the optimal allocation for the instances $(v_A^1,v_B^1)$ and $(v_A^2,v_B^2)$ is $(M\setminus\{b\},\{b\})$. By applying Lemma \ref{gs-unite-big-small-lemma}, we get that Bob sends different messages for $v_B^1$ and $v_B^2$  at vertex $r$.   
\end{proof}

\begin{lemma}\label{gs-smallcorollary}
	Assume that there exist two valuations
$v_A^1,v_A^2\in V_A^{\le m^2}$ that Alice sends different messages for at the root vertex $r$. 
Let $v_B^1,v_B^2\in V_B^{D,\gamma=m^5}$ be two semi-decisive valuations of Bob such that $P_A(\{a\},v_B^1)\neq P_A(\{a\},v_B^2)$.  \emph{Then}, Bob sends different messages at the root vertex $r$ for the valuations $v_B^1$ and $v_B^2$. 
\end{lemma}
\begin{proof}

		We begin by showing that every $v_A\in V_A^{\le m^2}$ and every $v_B\in V_B^{D,
		\gamma=m^5}$ satisfy that the unique optimal allocation is $(\{a\},M\setminus \{a\})$. By construction, Alice always wins item $a$ and Bob always wins item $b$. Regarding the rest of the items, we have that the marginal utility of Alice for an item given any bundle that she already has is at most $m^2\cdot (m+2)$ (because $\gamma\le m^2$), whereas Bob's marginal utility for every item is at least $\frac{m^5}{2}$. Therefore, Bob wins the rest of  the items.

	Therefore, the optimal allocation for the instances $(v_A^1,v_B^1)$ and $(v_A^2,v_B^2)$ is $(\{a\},M\setminus\{a\})$. By applying Lemma \ref{gs-unite-big-small-lemma}, we get that 
	Bob sends different messages at vertex $r$ for $v_B^1$ and for $v_B^2$, as needed. 
\end{proof}	
	The proof of Claim \ref{gs-claimbobsayspayment} stems from Lemma \ref{gs-bigcorollary} and from Lemma \ref{gs-smallcorollary}. It is identical to the proof of Claim \ref{claimbobsayspayment} and is omitted. 
%
\subsubsection{Alice Commits to Bob's Payment} \label{gs-alicesayspayment-sec}
We now use the information revealed by Bob about the semi-decisive valuations in $V_B^{D,\gamma=1}$ or in $V_B^{D,\gamma=m^5}$ to show that there exists \textquote{large} set of valuations such that Alice has to commit to Bob's payment for every possible allocation in the first round of the mechanism.  
In Subsection \ref{gs-sketch-subsec}, we will show how to use the payment to reconstruct these valuations. 
Observe that we now use the fact that $\mathcal{M}'$ is dominant strategies for Bob. 

Now, for every bundle $S \subseteq M \setminus \{a,b\}$, we define $\bar{S}=M\setminus (S\cup \{a,b\})$.
For the statement of the claim, we define
for every weight $\gamma$ and for every bundle $S \subseteq M \setminus \{a,b\}$, a semi-decisive valuation of Bob that is parameterized with weight $\gamma$, with the bundle $\bar{S}$ and the noise $\eta=0$, i.e.:
\begin{equation}\label{eq-decisive-val-sketch}
v_B^{\bar{S},\gamma,\eta=0}(x)=\begin{cases}
	m^8 \quad & x=b, \\
	0 \quad & x=a, \\
	\gamma \cdot (m+2) \quad & x\in \overline{S},\\
	\frac{\gamma }{2}  \quad & \text{otherwise.}	
\end{cases}
\end{equation}

We define another valuation that is identical to $v_B^{\bar{S},\gamma,\eta=0}$,
except that the noise $\eta$ is now equal to $\frac{1}{2}$:
$$
v_B^{\bar{S},\gamma,\eta=\frac{1}{2}}(x)=\begin{cases}
m^8 \quad & x=b, \\
\frac{1}{2} \quad & x=a, \\
\gamma \cdot (m+2) \quad & x\in \overline{S},\\
\frac{\gamma }{2}  \quad & \text{otherwise.}	
\end{cases}
$$

\begin{claim}\label{mainclaim-gs}
	The following holds for either $\gamma=1$ or for $\gamma=m^5$. Let $v_A\in V_A^{ND,\gamma}$ be a valuation, and let $z_A$ be the message that Alice sends for it at the root of the protocol. 
Fix a bundle $S\subseteq M \setminus \{a,b\}$ and let $z_B^0$ be the message that Bob sends at the root if his valuation is the semi-decisive valuation $v_B^{\bar{S},\gamma,\eta=0}$ defined above. 
Denote with $t^0$ the subtree that the message profile $(z_A,z_B)$ leads to. 
\emph{Then:}  
\begin{enumerate}
	\item There exists a leaf at subtree $t^0$ labeled with the allocation $(S\cup \{a\},M\setminus\{S\cup \{a\}\})$. 
	\item Every leaf at subtree $t^0$  that is labeled with the allocation $(S\cup \{a\},M\setminus\{S\cup \{a\}\})$ satisfies that it is labeled with the payment $P_B(M\setminus\{S\cup \{a\}\},v_A)$ for Bob. 
\end{enumerate}
\end{claim}

\begin{proof}
	We show that condition \ref{gs-condicondi1} of Claim \ref{gs-claimbobsayspayment} implies that  Claim \ref{mainclaim-gs} holds for $\gamma=m^5$. The proof that condition \ref{gs-condicondi2} of Claim \ref{gs-claimbobsayspayment} implies that Claim \ref{mainclaim-gs} holds for $\gamma=1$ is analogous. Claim \ref{mainclaim-gs} follows since by Claim \ref{gs-claimbobsayspayment} at least one of those conditions holds. 
	Assume that condition \ref{condicondi1} holds. Let $v_A\in V_A^{ND,\gamma=m^5}$ be a valuation, and let  $S \subseteq M\setminus\{a,b\}$ be a bundle. 
	
We explain why the unique welfare maximizing allocation for both instances $(v_A,v_B^{\bar{S},\gamma=m^{5},\eta=0})$ and $(v_A,v_B^{\bar{S},\gamma,\eta=\frac{1}{2}})$ is $(S\cup \{a\}, M\setminus \{S\cup \{a\}\})=(S\cup \{a\},\overline{S}\cup \{b\})$. Note that Alice  wins item $a$, because the marginal value of $a$ for Alice given any bundle is always $m^8$ whereas the marginal value of $a$ for Bob is at most $\frac{1}{2}$. Due to the same reason, Bob  wins $b$. For every item in $S$, Bob's value is  $\frac{m^5}{2}$, whereas the marginal value of Alice for it is at least $m^5$, so Alice wins all the items in $S$. For the items in $\overline{S}$, the  value of Bob  is $m^5\cdot (m+2)$ , whereas the marginal value of Alice for every such item is at most $m^5\cdot (m+1)$. Therefore, Bob wins all the items in $\bar{S}$. 
	Thus, the leaf $l^0$ that $(v_A,v_B^{\bar{S},\gamma=m^{5},\eta=0})$ reaches is labeled with the allocation $(S\cup \{a\},M\setminus\{S\cup \{a\}\})$.  By definition, this leaf belongs in the subtree $t^0$, so we have part $1$ of the claim.  
	
	For the proof of the second part, we denote $v_B^{\bar{S},\gamma,\eta=0}$ with $v_B^{0}$ and $v_B^{\bar{S},\gamma,\eta=\frac{1}{2}}$ with $v_B^{1}$. 
Recall that by Lemma \ref{lemma-approx-uniqueness-gs} we have that: 
	\begin{equation*}\label{eq-payment-range-gs}
	P_A(\{a\},v_B^0)\le v_B^0(M)-v_B^0(M\setminus \{a\})+\frac{1}{8m},\quad 
	P_A(\{a\},v_B^{1})\ge v_B^{1}(M)-v_B^{1}(M\setminus \{a\})-\frac{1}{8m}
	\end{equation*}
	Therefore: 
	\begin{multline*}
	P_A(\{a\},v_B^0)\le v_B^0(M)-v_B^0(M\setminus \{a\})+\frac{1}{8m} < v_B^{1}(M)-v_B^{1}(M\setminus \{a\})-\frac{1}{8m} \le P_A(\{a\},v_B^{1})  \\ \implies P_A(\{a\},v_B) < P_A(\{a\},v_B')
	\end{multline*}
	where the strict inequality holds because $v_B^1(M)-v_B^0(M)=\frac{1}{2}$ and $v_B^1(M \setminus \{a\})=v_B^0(M \setminus \{a\})$. 
	Therefore, by condition  \ref{gs-condicondi1} of Claim \ref{gs-claimbobsayspayment} we have that 
	Bob sends a different message $z_B^1$ for the valuation $v_B^1$ than the message $z_B^0$ he sends for the valuation $v_B^0$ at vertex $r$. 
	Denote with $t^1$ the subtree that the messages $(z_A,z_B^1)$ lead to, and denote the leaf in $t^1$ that $(v_A,v_B^1)$ reaches with $l^1$. By the above, $l_1$ is labeled with the allocation  $(S\cup \{a\}, M\setminus \{S\cup \{a\}\})$. 
	
	We remind that  the leaf $l^0$ is labeled with the allocation $(S\cup \{a\}, M\setminus \{S\cup \{a\}\})$ and with the payment $P_B(M\setminus \{S\cup \{a\},v_A)$ for Bob (the latter holds because $\mathcal M'$ realizes the welfare maximizer  $f$ with the payment schemes $P_A,P_B$).  By Lemma \ref{lemma-known-prices}, all the leaves in subtrees $t^0$ and in $t^1$ that are labeled with the allocation $(S\cup \{a\}, M\setminus \{S\cup \{a\}\})$ have the same price for Bob.
	By combining these two facts, we get that all the leaves in the subtree $t^0$ labeled with the allocation $(S\cup \{a\}, M\setminus \{S\cup \{a\}\})$ are labeled with  the payment $P_B(M\setminus \{S\cup \{a\}\},v_A)$ for Bob, which completes the proof. 
\end{proof}

\subsubsection{Reconstructing Alice's Valuation} \label{gs-sketch-subsec}
We can now complete the proof of Proposition \ref{dsicsketchclaim-gs}.	 Given the minimal mechanism $\mathcal{M}'$ for the valuations $V_A\times V_B$, we explain how to construct an exact representation for a valuation $v\in GS(M\setminus\{a,b\})$ with at most 
 $c'+\mathcal{O}(\log (m))\le c+ \mathcal{O}(\log (m))$ bits (we remind that $c,c'$ stand for the communication complexity of the mechanisms $\mathcal{M},\mathcal{M}'$ respectively).    

	For  $v \in GS(M\setminus\{a,b\})$, consider the following non-decisive valuation $v_A \in V_A^{ND,\gamma} \subseteq V_A$ that is parameterized with the noise $\eta=0$:  
	$$
	\forall S\subseteq M\cup\{a,b\}, \quad v_A(S)=\begin{cases}
		\gamma\cdot v(S\setminus \{a,b\})+ \gamma\cdot |S|+ m^8 \quad & \{a,b\} \subseteq S, \\	
		\gamma\cdot v(S\setminus \{a\})+\gamma\cdot |S|+m^8 \quad & a\in S, \\
		\gamma\cdot v(S\setminus \{b\})+\gamma\cdot |S|  \quad & b\in S, \\
		\gamma\cdot v(S) +\gamma\cdot |S|  \quad &\text{otherwise.}
	\end{cases} 	
	$$ 
	where $\gamma$ is the scalar in $\{1,m^5\}$ that Claim \ref{mainclaim-gs} holds for. 
	
	The description of a valuation $v$ is $v_A(M)$ and the message $z_A$ that Alice sends for it at the root vertex $r$.
	Note that the overall size of the description is at most $c'+\mathcal{O}(\log m)\le c+\mathcal{O}(\log m)$ bits. 
	Fix $S\subseteq M\setminus \{a,b\}$. We remind that we want to reconstruct $v(S)$ from $v_A(M)$ and from the message $z_A$. Let $z_B$ be the message that Bob sends at vertex $r$ of the mechanism $\mathcal M'$ when his valuation is the decisive valuation $v_B^{\bar{S},\gamma,\eta=0}$ that is defined in (\ref{eq-decisive-val-sketch}), where $\gamma$ is again the scalar that Claim \ref{mainclaim-gs} holds for.  Let $l$ be an arbitrary leaf that is labeled with the allocation $(S\cup \{a\}, M\setminus \{S\cup \{a\}\})$ in the subtree that $(z_A,z_B)$ leads to at the root vertex $r$.  
	By Claim \ref{mainclaim-gs}, such a leaf necessarily exists and it is labeled with the payment  $P_B(M\setminus \{S\cup \{a\}\},v_A)$ for Bob.    
	
	By Corollary \ref{cor-approx-payments-gs}, 
	we can reconstruct $v_A(S\cup \{a\})$ from $P_B(M\setminus \{S\cup \{a\}\},v_A)$ and $v_A(M)$ with no additional communication. 
	Now, to extract $v(S)$, we remind that by definition:
	$$v_A(S\cup\{a\})=\gamma \cdot\big[v(S)+|S|\big] +m^8
	\implies v(S)= \frac{v_A(S\cup\{a\})-m^8}{\gamma} - |S|$$
	which completes the proof.       
	
%
%
%

\subsubsection{Proof of Lemma \ref{lemma-approx-uniqueness-gs}}\label{subsec-sketch-proof-gs}  
We prove Lemma \ref{lemma-approx-uniqueness-gs} for the valuations of Alice and the payment scheme of Bob. The proof for Bob's valuations and Alice's payment scheme is identical. 

Fix a valuation $v_A\in V_A^{ND}\cup V_A^D$ and a bundle of items  $S \subseteq M$ of size $s$. We begin by defining an arbitrary order on the elements of $S$, i.e. $S=\{m_1,m_2,\ldots,m_s\}$. Then, we use it to define a strictly increasing sequence of subsets, $\{S_j\}_{j\in \{0,\ldots,s\}}$, where $S_j=\cupdot_{\ell=1}^{j}\{m_j\}$. In other words, $S_0=\emptyset$, $S_1$ contains the first item, $S_2$ contains the first and the second items and so on and so forth.

We begin by  showing that for every $1\le j\le s$, the payment of Bob satisfies that:
\begin{equation*}
	P_B(S_j,v_A)-P_B(S_{j-1},v_A) \in \big[v_A(M\setminus S_{j-1})-v_A(M\setminus S_j)\pm \frac{1}{8m^2}\big] 
\end{equation*}
Fix $1\le j\le s$. Observe the additive valuation $v_{0}\in V^{P}$ that is parameterized with the valuation $v'=v_A$, with the special bundle $S^\ast=S_{j-1}$, with the special item $m_j$ and with the sign $sn=0$:
\begin{gather*}
v_0(x)=\begin{cases}
	m^{15} \quad & x\in S_{j-1}, \\
	v_A(M \setminus S_{j-1})-v_A(M \setminus S_{j}) + \frac{1}{8m^2} \quad & x=m_j, \\
	0  \quad & \text{otherwise.} 
\end{cases}
\end{gather*}
Now, observe that the welfare maximizing allocation if Alice's valuation is $v_A$ and Bob's valuation is $v_0$  is $(M\setminus S_{j},S_{j})$. First of all, Bob necessarily wins all the items in $S_{j-1}$ because the marginal value of every item in $S_{j-1}$ for Bob given $v_0$ is larger than any value of Alice. 
Second, Alice necessarily wins the items in $\{m_{j+1},\ldots,m_s\}$ because the marginal value of $v_A$ is strictly positive for them, whereas Bob's value for them is $0$. Now, observe that the marginal value of Alice from $m_j$ given $\{m_{j+1},\ldots,m_s\}=M\setminus S_{j}$ is equal to $v_A(M\setminus S_{j-1})-v_A(M \setminus S_{j})$, whereas Bob's value increases by $v_A(M \setminus S_{j-1})-v_A(M \setminus S_{j}) + \frac{1}{8m^2}$ if he wins $m_j$. Thus, Bob necessarily wins $m_j$. 
 
We remind that $\mathcal{M}$ is ex-post incentive compatible (since it is dominant strategy incentive compatible), and that it realizes a welfare-maximizer with the payment schemes $P_A,P_B$, so:  
\begin{multline}\label{eq-payment-upper-gs}
	v_0(S_j)-P_B(S_j,v_A)\ge v_0(S_{j-1})-P_B(S_{j-1},v_A)  \implies \\
	v_A(M \setminus S_{j-1})-v_A(M \setminus S_{j})+\frac{1}{8m^2}=  v_0(m_j)=v_0(S_j)-v_0(S_{j-1})\ge P_B(S_j,v_A) - P_B(S_{j-1},v_A)  \implies \\
		v_A(M \setminus S_{j-1})-v_A(M \setminus S_{j})+\frac{1}{8m^2} \ge P_B(S_j,v_A) - P_B(S_{j-1},v_A)
\end{multline}
To prove a lower  bound on $P_B(S_j,v_A) - P_B(S_{j-1},v_A)$, we construct the valuation $v_1\in V^{P}$ which is parameterized with the valuation $v'=v_A$, 
with the special bundle $S^\ast=S_{j-1}$, with the special item $m_j$ and with the sign $sn=1$:
\begin{gather*} 
	v_1(x)=\begin{cases}
		m^{15} \quad & x\in S_{j-1}, \\
		v_A(M \setminus S_{j-1})-v_A(M \setminus S_{j}) - \frac{1}{8m^2} \quad & x=m_j, \\
		0  \quad & \text{otherwise.} 
	\end{cases}
\end{gather*}
Now, the welfare-maximizing allocation given $(v_A,v_1)$ is $(M\setminus S_{j-1},S_{j-1})$. Due to the same reasons as above, Bob wins all the items in $S_{j-1}$ and Alice wins the items $\{m_{j+1},\ldots,m_s\}$. Now, observe that the marginal value of Alice from $m_j$ given $\{m_{j+1},\ldots,m_s\}=M\setminus S_{j}$ is equal to $v_A(M\setminus S_{j-1})-v_A(M \setminus S_{j})$, whereas Bob's value increases by $v_A(M \setminus S_{j-1})-v_A(M \setminus S_{j})- \frac{1}{8m^2}$ if he wins $m_j$. Thus, Alice wins $m_j$ this time.
 Due to the same considerations as before, we have that: 
\begin{multline}\label{eq-payment-lower-gs}
	v_1(S_{j-1})-P_B(S_{j-1},v_A) \ge
	v_1(S_{j})-P_B(S_{j},v_A) \\ \implies P_B(S_{j},v_A)-P_B(S_{j-1},v_A) \ge 	v_1(S_{j})-	v_1(S_{j-1})=v_1(m_j)= v_A(M \setminus S_{j-1})-v_A(M \setminus S_{j}) - \frac{1}{8m^2} 
\end{multline}
Combining (\ref{eq-payment-upper-gs}) and (\ref{eq-payment-lower-gs}) gives:
\begin{equation}\label{eq-partial-sums-gs}
	v_A(M \setminus S_{j-1})-v_A(M \setminus S_{j}) - \frac{1}{8m^2} 
	\le P_B(S_{j},v_A)-P_B(S_{j-1},v_A) 
	\le 
	v_A(M \setminus S_{j-1})-v_A(M \setminus S_{j}) + \frac{1}{8m^2} 
\end{equation} 
We can now complete the proof. We remind that $\mathcal{M}$ is normalized and that $S_0=\emptyset$, so $P_B(S_0,v_A)=0$. Therefore, the following  telescopic sum equals $P_B(S,v_A)$:
\begin{align*}\label{eq-telescopic}
	P_B(S,v_A)&=P_B(S,v_A)-P_B(S_{s-1},v_A)+P_B(S_{s-1},v_A)-\ldots-P_B(S_1,v_A)\\ &  +P_B(S_1,v_A) - P_B(S_0,v_A)=  \nonumber \\
	&= \sum_{j=1}^{s}  P_B(S_j,v_A)-P_B(S_{j-1},v_A)  \nonumber 
\end{align*} 
Observe that by (\ref{eq-partial-sums-gs}) we have that: 
\begin{align*}
	P_B(S,v_A) &= \sum_{j=1}^{s}  P_B(S_j,v_A)-P_B(S_{j-1},v_A)\\ &\ge \sum_{j=1}^{s} \big[
	v_A(M \setminus S_{j-1})-v_A(M \setminus S_{j}) - \frac{1}{8m^2}  \big] \\
	&\ge v_A(M)-v_A(M\setminus S) -\frac{1}{8m} 
\end{align*}
As needed. 
A similar analysis 
gives that $v_A(M)-v_A(M\setminus S) +\frac{1}{8m}\ge P_B(S,v_A)$, which completes the proof.

\section{A Hardness Result for Dominant Strategy Mechanisms: Proof of Theorem \ref{theorem-general-impossibility}}
\label{sec-general2}
We will show that there exists a class of general valuations (where the value of each bundle can be described with $poly(m)$ bits) such that every  dominant strategy mechanism $\mathcal M$ for it
that $m^{1-4\epsilon}$ approximates the social welfare can be used to construct a simultaneous algorithm for the \textquote{hard distribution} (see Subsection \ref{impossible-sim-general-subsec}) with (roughly) the same communication and (roughly) the same approximation ratio as the mechanism. 

The structure of the proof is as follows.  
We begin by describing the class of valuations whose  dominant strategy mechanism will be used to construct a simultaneous algorithm (Subsection \ref{subsec-class-of-valuations}).
Afterwards, we describe the steps of the algorithm and state some observations that are necessarily for that (Subsection \ref{subsec-alg-description}). Then, for the analysis of the algorithm, we show that with high enough probability, properties that guarantee a sufficient approximation ratio hold. Very roughly speaking, in Subsection \ref{subsec-analysis} we prove that given those properties the algorithm has good approximation ratio, and we prove that a large enough fraction of the instances satisfy those properties in Subsection \ref{subsec-properties}. 
   
To use the mechanism $\mathcal M$ as a building block for a simultaneous algorithm, we assume that it is deterministic,  normalized and has  no negative transfers. We remind that a mechanism is normalized if the 
price of the empty bundle is always  $0$ and that no-negative-transfers means that the payment of each player is always non-negative.  Throughout the proof,  we will sometimes abuse notation when analyzing $\mathcal M$ by saying that a valuation \textquote{sends} a message, instead of saying that the dominant strategy given a valuation dictates it.

\subsection{Description of Class of Valuations}\label{subsec-class-of-valuations}
For the proof, we assume that $\mathcal M$ is dominant strategy for the following class of valuations.
Consider the set of all possible valuations of player $i$  in the hard distribution of Subsection \ref{impossible-sim-general-subsec}, with the following adjustments.  Each valuation in the support of the hard distribution has a \emph{weight} $\alpha\in B=\{1,2,2^2,\ldots,2^{2^m}\}$ and a \emph{noise} $\eta$, that is negligible compared to $\alpha$. Formally, if $u$ is a possible valuation of player $i$ in the hard distribution then we consider valuations $v=\alpha(1+\eta)\cdot u$ for any weight $\alpha
\in B$ and noise $\eta \in \{0,\frac 1 {4^m}, \frac {2} {4^m}, \frac {3} {4^m}, \ldots,  \frac {3^m} {4^m}\}$. 
In addition, we assume that for every valuation $v=\alpha(1+\eta)\cdot u$ that belongs in the class,
 its base valuation $u$ is not the all-zero valuation, even though it is in the support of the hard distribution. In other words, every valuation in the domain that we consider has at least one valuable set.\footnote{For the proof, we will use Proposition 
 	\ref{prop-simul-hard-dist}. Note that the hardness of approximation for simultaneous algorithms persists despite the omission of all-zero valuations, because by the definition of the hard distribution, the probability to sample an instance $(u_1,\ldots,u_n)$ that has a valuation  $u_i$ which is the all-zero valuation is negligible.} 
 We call a valuation a \emph{base valuation} if it is in the support of the hard distribution and has at least one valuable set.



We say that a base valuation $u_i$ of player $i$ is a \emph{random base valuation} if it is sampled as follows. 
$u_i$ has some base set  $T_i$, $|T_i|=2 m^\eps$, and it   
chooses $t=2^{\Theta(\epsilon^2\cdot m^\epsilon)}$ valuable sets uniformly at random (where one of them is fixed), each of size $m^\eps$ independently and uniformly at random from $T_i$, and gives a value of $1$ for each of them with probability $\frac{1}{m}$ (with the guarantee that at least one bundle has value of $1$).  Throughout the proof, we require that $\mathcal M$ provides approximation of $m^{1-4\epsilon}$ also for instances where the valuations are not necessarily correlated in the same way as in the hard distribution.

Recall that an instance the hard distribution of Subsection \ref{impossible-sim-general-subsec} satisfies that the players are partitioned into groups that have the same base set. Half of the base set of the players in each group $j$ is disjoint from the base sets of all other groups, and half of it is a shared set $B$ which we call the ``center''. Given a player in group $j$,  the disjoint half of his base set is called \emph{the special set}. 
Each group $j$ has a random family $\mathcal A_j$ of sets of size $m^\eps$, and each bidder $i$ in the group is interested in a set from  $\mathcal{A}_j$  (i.e., gives it a value of $1$) with probability $\frac 1 m$. See Subsection~\ref{impossible-sim-general-subsec} for the formal construction. We remind that each player that is interested in his special set is called a \emph{special bidder}. 

\subsection{A Simultaneous Algorithm for the Hard Distribution} \label{subsec-alg-description}
Our goal is to show that given the mechanism $\mathcal M$,  there exists a simultaneous deterministic algorithm for the hard distribution with approximation ratio better than $m^{1-\epsilon}$ and with communication complexity that is polynomial in the communication complexity of $\mathcal M$, which we denote with $cc(\mathcal M)$. Observe that by Lemma \ref{optimal-high-in-expectation}, the optimal welfare for an instance from the hard distribution is at least $m^{1-\epsilon} - 1$. 
Thus, if we manage to show a deterministic simultaneous algorithm that has communication complexity $poly(cc(\mathcal M),m,n)$  and provides in expectation $\Omega(m^{\epsilon})$ of the welfare for instances from the hard distribution, then by Proposition \ref{prop-simul-hard-dist} we get that the communication complexity of the mechanism $\mathcal M$ is at least $poly(\frac {2^{m^{\frac {\eps^2} 2}}} n)$. 
Throughout the proof, we heavily use the fact that  the mechanism $\mathcal M$ is deterministic, that is, provides an approximation ratio of $m^{1-4\eps}$ for every instance. 

  
Observe that it suffices to provide a randomized simultaneous algorithm instead of a deterministic one. The reason for it is as follows. Assume that there exists a randomized algorithm that provides $\Omega(m^\epsilon)$ of the welfare in expectation over the sampling from the hard distribution and over its random coins. Then,  by an averaging argument,  there exists at least one sequence of coins such that its expected welfare is $\Omega(m^\epsilon)$ (where the probability is taken only over the sampling from the hard distribution). Therefore, a randomized algorithm implies the existence of a deterministic algorithm.  

Thus, we want to build a randomized algorithm with $poly(cc(\mathcal M),m,n)$ bits that outputs an allocation whose welfare  is   $\Omega(m^{\epsilon})$ in expectation for $(u_1,\ldots,u_n)$. Note that the expectation is taken over the sampling of $(u_1,\ldots,u_n)$ from the hard distribution and over the random coins of the algorithm. 
Before we describe the algorithm, we need to state the following lemmas:
\begin{lemma}\label{lemma-first-round2}
	There exists some set $A$ of $\mathcal{O}(\log m)$  consecutive 
	powers of $2$ that satisfies the following:
	let $x$ be the first vertex in the mechanism $\mathcal M$ where two valuations with weights in $A$ send different messages. Then,  we have that for every player $i$,  for every base set $T_i$, and
	for every noise $\eta_i$, 	
	the following holds with probability at least $1-\frac 1 {exp(m^{\eps})}$: let $u_i$ be a random base valuation of player $i$ with base set $T_i$. Then, there exist two consecutive powers in $A$, $\alpha,2\alpha$, such that the message that player $i$ sends in vertex $x$ when his valuation is $\alpha(1+\eta_i)\cdot u_i$ is different from the message he sends when his valuation is  $2\alpha(1+\eta_i)\cdot u_i$. 
\end{lemma}
We use Lemma \ref{lemma-first-round2} in the proof of Lemma \ref{lemma-good-alpha}. We defer its proof to Subsection \ref{subsec-first-round-proof2}. 
Now, we say that $\alpha$ is \emph{critical} for player $i$ with base valuation $u_i$ and noise $\eta_i$ if he sends different messages at vertex $x$ for $\alpha(1+\eta_i)\cdot u_i$ and for $2\alpha(1+\eta_i)\cdot u_i$. 
Given an instance from the hard distribution $u=(u_1,\ldots,u_n)$ and noises $\eta=(\eta_1,\ldots,\eta_n)$, we define for every $\alpha$ the subset $P_\alpha(u,\eta)\subseteq N$ as the number of special bidders such that $\alpha$ is critical for player $i$ given the base valuation $u_i$ and the noise $\eta_i$. 
\begin{lemma}\label{lemma-good-alpha}
	There exists $\alpha^\ast \in A$ such that: $$\Pr_{(u_1,\ldots,u_n)\sim \mathcal U,\eta} \Big[\big|P_{\alpha^\ast}(u,\eta)\big|\ge \frac{m^{1-\epsilon}-1}{|A|}\Big]\ge  \frac{1}{2|A|}$$ 
\end{lemma}
\begin{proof}
	Let  $(u_1,\ldots,u_n)$ and $(\eta_1,\ldots,\eta_n)$ be base valuations and noises such that all players have a critical value of $\alpha$.  We remind that by the definition of the hard distribution, there are $m^{1-\epsilon}-1$ special bidders given the valuations $(u_1,\ldots,u_n)$. Since by assumption all bidders have a critical $\alpha$, by applying an averaging argument we get that there exists a weight $\alpha \in A$ such that $P_\alpha(u,\eta)\ge \frac{m^{1-\epsilon}-1}{|A|}$. We say that this weight is \emph{good} for $(u_1,\ldots,u_n)$ and for $(\eta_1,\ldots,\eta_n)$.  
	
	By applying the averaging argument once again, we get that there necessarily exists $\alpha^\ast\in A$
	 that is good with probability at least $\frac{1}{|A|}$ 
	 over the sampling of $(u_1,\ldots,u_n)$ from the hard distribution and the uniform sampling of the noises, conditioned on the assumption that $(u_1,\ldots,u_n)$ and $(\eta_1,\ldots,\eta_n)$ satisfy that all players have a critical value of $\alpha$. Note that this assumption holds with probability of at least $1-\frac{n}{exp(m^\epsilon)}$ by Lemma \ref{lemma-first-round2}. The claim immediately follows.   
%
%
%
\end{proof}

\subsubsection{Description of the Simultaneous Algorithm}
Our randomized algorithm for the hard distribution is as follows.
Given $(u_1,\ldots,u_n)$, every  player $i$ constructs a valuation $v_i$ that is based on $u_i$.
He samples a noise  $\eta_i$ from $\{\frac 1 {4^m}, \frac {2} {4^m}, \frac {3} {4^m}, \ldots,  \frac {3^m} {4^m}\}$. 
If he sends different messages for the valuations $\alpha^\ast(1+\eta_i)\cdot u_i$ and $2\alpha^\ast(1+\eta_i)\cdot u_i$ in vertex $x$,
 we say that $\alpha^\ast$ is \emph{critical} for the base valuation $u_i$ and for $\eta_i$. In this case,
we sample a weight $\alpha\in \{\alpha^\ast,2\alpha^\ast\}$, each with probability $\frac{1}{2}$.  Then, we set $v_i=\alpha(1+\eta_i)\cdot u_i$. 

 Otherwise, $\alpha^\ast$ is not critical for  player $i$ given the base valuation $u_i$ and the noise $\eta_i$.  In this case, 
 player $i$ samples uniformly at random one of his valuable sets, which we denote with $T_i$. Then, $v_i$  is defined as follows:
\begin{equation*}
v_i(S)=\begin{cases}
2\alpha^\ast(1+\eta_i), &\quad T_i \subseteq S \\
0, & \quad \text{otherwise.}
\end{cases}
\end{equation*}
We denote the distribution of the valuations $(v_1,\ldots,v_n)$ with $\mathcal I_{\alpha^\ast}$, and denote an instance in the support of this distribution with $I_{\alpha^\ast}$. 
Now, the message that player $i$ sends in the simultaneous algorithm consists of several blocks:
\begin{enumerate}
\item The message $z_i$ he sends in vertex $x$ of $\mathcal{M}$ when his valuation is $v_i$. \label{block-message-in-mech}
\item One bit that specifies whether the weight  of the valuation $v_i$ is $\alpha^\ast$ or $2\alpha^\ast$.
\item A bit that specifies whether $\alpha^\ast$ is critical for the base valuation $u_i$ and the noise $\eta_i$ or not.   
\item The base set of $u_i$, i.e. the union of all the items that belong in a bundle  $T$ such that $u_i(T)>0$ (with probability $1-\frac 1 {exp(m^\epsilon)}$ each player can identify his entire base set by taking the union of all the sets he is interested in, since an item appears in a set with probability $\frac 1 2$ and there are exponentially many sets). \label{block-baseset}
\item $poly(m)$ bits that specify the value of the random noise $\eta_i$. 
\end{enumerate}
  Now, based on the blocks of messages, we output an allocation $(S_1,\ldots,S_n)$ as follows.
  Note that all players sent their base sets as a block in the simultaneous protocols, so with high probability the messages of all players reveal  the center $B$ and the special sets $A_1,\ldots,A_l$. 
  If they do not, then we output the allocation with the empty bundle for all players. 
  
 For every player $i$, if the weight of  $v_i$ is $2\alpha^\ast$, then we set $S_i\gets \emptyset$. Otherwise, we 
denote with $G_j$ the group of players that player $i$ belongs to and observe the tree that is induced by the messages $(z_1,\ldots,z_{i-1},z_{i+1},\ldots,z_n)$ at the initial vertex $x$.  
We allocate to player $i$ either $A_j$, the special set of group $j$, or the empty bundle. 
We allocate $A_j$ to player $i$ only if there exists a valuation $v_i''$ that satisfies all the following conditions at the same time: 
\begin{enumerate}
	\item $v_i''$ has weight $\alpha^\ast$ and noise $\eta_i$.\footnote{Note that this is the reason for sending the noise as a block in the simultaneous algorithm.}
	\item The dominant strategy of  player $i$, $\mathcal S_i^{dom}$, dictates sending the message $z_i$ at vertex $x$ given the valuation $v_i''$.   
	\item The dominant strategy $\mathcal S_i^{dom}(v_i'')$ guarantees a valuable set for player $i$ if the protocol reaches vertex $x$ and the players in $N\setminus \{i\}$ send $z_{-i}$  (To be clear, $S_i^{dom}(v_i'')$ guarantees a valuable set if for every strategy profile $\mathcal S_{-i}$ of the other players in the mechanism $\mathcal M$ that is consistent with vertex $x$ and with the message $z_{-i}$, we have that given $S_i^{dom}(v_i'')$ and $\mathcal{S}_{-i}$, the mechanism $\mathcal M$ reaches a leaf where player $i$ wins a set $T$ such that $v_i''(T)>0$).  
\end{enumerate}
Otherwise, we allocate to player $i$ the empty bundle. 
\subsection{Analysis of the Simultaneous Algorithm} \label{subsec-analysis}
It is easy to see that the communication complexity of the randomized algorithm is $poly(cc(\mathcal{M}),m,n)$.  Therefore, to conclude  Theorem \ref{theorem-general-impossibility}, it remains to show that:
\begin{proposition}\label{prop-alg-is-good}
	$\E_{r,(u_1,\ldots,u_n)}\sum_i u_i(S_i)= \Omega(m^{\eps})$, 
	where the probability is taken over the sampling of $\vec{u}$ from the hard distribution and the sequence of random coins of the algorithm (which we denote with $r$). 
\end{proposition}
To analyze the welfare guarantees of the algorithm, we define for every instance $I_{\alpha^\ast}$ in the support of $\mathcal I_{\alpha^\ast}$ another instance that is \textquote{close} to it. Formally:
\begin{definition}
	Let $I_{\alpha^\ast}=(v_1,\ldots,v_n)$ be an instance in the support of $\mathcal I_{\alpha^\ast}$ that is defined as above. Then, $(v_1',\ldots,v_n')$ is a \emph{diluted} version of  
	$(v_1,\ldots,v_n)$ if for every player $i \in G_j$ such that:
	\begin{enumerate}
		\item  The weight of $v_i$ is $2\alpha^\ast$ \emph{and}
		\item He is interested in his special set (i.e. $v_i(A_j)>0$) \emph{and}
		\item  $\alpha^\ast$ is critical for $u_i$ and $\eta_i$, where $u_i$ and $\eta_i$ are the base valuation and the noise that $v_i$ is based on.
	\end{enumerate} 
we replace $v_i=2\alpha^\ast\cdot(1+\eta_i)\cdot u_i$ with $v_{i}'$ that sends the same message in vertex $x$ as $v_i$ and also has weight $2\alpha^\ast$, but has zero value for the special set.
If one of the conditions above does not hold, then we set  $v_i'\gets v_i$. 	
\end{definition}

We now define several events and prove that they occur
\emph{simultaneously} with high enough probability (Claim \ref{claim-simultaneously}), so 
analyzing the welfare guarantees of the algorithm only for the case where they all hold suffices.
For that, we will analyze the probability of each event separately (Subsection \ref{subsec-properties}). 
We remind that we denote the random base valuations that are sampled from the hard distribution with $(u_1,\ldots,u_n)$, and their weighted and noisy versions with $(v_1,\ldots,v_n)$. 
 The events are:
\begin{enumerate}
	\item $(v_1,\ldots,v_n)$ has a diluted version $(v_1',\ldots,v_n')$ (Claim \ref{claim-diluted2}). \label{event-diluted} \label{first-event2} 
		\item The optimal welfare for the diluted version $(v_1',\ldots,v_n')$ is at least $\alpha^\ast \cdot  m^{1-2\epsilon}$ (Claim \ref{claim-diluted-has-high-welfare}).   \label{event-diluted-high-welfare}
		\item 	The welfare of every allocation given the valuations $(v_1',\ldots,v_n')$ consists only of the welfare of  bidders that receive their special sets and are interested in them, plus one player (Due to the same arguments as in Lemma \ref{lemma-welfare-comes-from-special}, this holds with probability $1-e^{-\Omega(m^\epsilon)}$). \label{event-welfare-comes-from-special} 
		\item The diluted version $(v_1',\ldots,v_n')$ satisfies that every bidder that is interested in his special set has weight $\alpha^\ast$ (Claim \ref{claim-all-special-with-alpha}).  
\label{event-all-special-with-alpha}
	\item The allocation and the payments of the mechanism $\mathcal M$ for $(v_1',\ldots,v_n')$ satisfy that every player that wins a set that is valuable for him has a strictly positive profit (Claim \ref{claim-positive-profit2}). \label{event-positive-profit}
	\item The allocation $(S_1,\ldots,S_n)$ is feasible (Claim \ref{claim-feasible-allocation}).
	\item  Each player can identify his entire base set. As we explained above, this event occurs with probability at least $1-\frac{n}{\exp(m^\eps)}$. \label{almost-last-event}   
	\label{last-event2}
\end{enumerate}
 \begin{claim} \label{claim-simultaneously}
Events \ref{first-event2}-\ref{last-event2} occur simultaneously with probability at least $\frac{1}{4|A|}$. 
 \end{claim}
 \begin{proof}
By Claim \ref{claim-diluted2}, event \ref{first-event2} occurs with probability at least $1-\frac{3}{n}$ and by Claim \ref{claim-diluted-has-high-welfare} event \ref{event-diluted-high-welfare} occurs with probability at least $\frac{1}{3|A|}$. Events \ref{event-welfare-comes-from-special}-\ref{almost-last-event} each hold with probability of at least $1$ minus an exponentially small probability. By combining these three facts together, we get that events \ref{first-event2}-\ref{almost-last-event}  occur simultaneously with probability at least $\frac{1}{4|A|}$.   
 \end{proof}

\begin{proof}[Proof of Proposition \ref{prop-alg-is-good}]
Sample $(u_1,\ldots,u_n)$ from the hard distribution and sample $I_{\alpha^\ast}=(v_1,\ldots,v_n)$ 
as described in the algorithm above. 
	Note that by the union bound, all the desirable events (numbered \ref{first-event2}-\ref{last-event2})  occur simultaneously with probability of at least $\frac{1}{3|A|}$. 
	We will analyze the expected value of $\sum_i u_i(S_i)$, assuming that events \ref{first-event2}-\ref{last-event2}  all hold. 

	Given $(v_1,\ldots,v_n)$, let $(v_1',\ldots,v_n')$ be its diluted version. By condition \ref{event-diluted-high-welfare}, the optimal allocation $(O_1,\ldots,O_n)$ satisfies that
	 $\sum_{i=1}^{n}v_i'(O_i)\ge \alpha^\ast m^{1-2\eps}$.
	  Thus, the mechanism $\mathcal M$ necessarily outputs an allocation $(X_1,\ldots,X_n)$ such that $\sum_{i=1}^{n}v_i'(X_i)\ge \frac{\alpha^\ast m^{1-2\eps}}{m^{1-4\eps}}=\alpha^\ast m^{2\eps}$.

	Note that by event \ref{event-welfare-comes-from-special}, the welfare of $(X_1,\ldots,X_n)$   consists of bidders who receive their special sets and are interested in them given the valuation profile $(v_1',\ldots,v_n')$, plus at most one bidder. By event \ref{event-all-special-with-alpha}, all those bidders (except one) satisfy that the weight of $v_i'$ is $\alpha^\ast$.   
	By construction, if  a bidder who is not interested in the special set contributes anything to the welfare, it is at most $2\alpha^\ast(1+\frac{3^n}{4^n})$, whereas the every  player that gets a special set and is interested in it contributes at most $\alpha^\ast(1+\frac{3^n}{4^n})$.
	 Now,	we define the set $\mathcal H$ as the set of players who are interested in their special set given the valuation $v_i'$ (i.e. $v_i'(A_j)>0$ if $i\in G_j$) and also satisfy that $v_i'(X_i)>0$. By the above:
%
\begin{equation} \label{lower-bound-h}
	|\mathcal H| \ge \frac{\alpha^\ast m^{2\eps}-2\alpha^\ast(1+\frac{3^n}{4^n})}{\alpha^\ast(1+\frac{3^n}{4^n})}
	\ge 
		\frac{\alpha^\ast m^{2\eps}-3\alpha^\ast}{2\alpha^\ast} \ge \frac{m^{2\eps}}{4} 
\end{equation}
	We will now show that for every happy bidder $i\in \mathcal H$ such that $i\in G_j$, we have that if the events \ref{first-event2}-\ref{last-event2} hold, then with probability $1$, $S_i=A_j$. 
	It immediately implies that:
	\begin{equation}\label{eq-special-have-value-1}
		i\in \mathcal H \land S_i=A_j \implies v_i'(S_i)>0 \implies u_i(S_i)=1
	\end{equation}
	
	 It immediately implies that:
	\begin{align*}
		\E_{r\sim R,(u_1,\ldots,u_n)}\Big[\sum_{i=1}^{n} u_i(S_i)\Big] &\ge \Pr[\text{events \ref{first-event2}-\ref{last-event2}}]\cdot \E_{r\sim R,(u_1,\ldots,u_n)} \Big[\sum_{i=1}^{n}u_i(S_i)\big|\text{events \ref{first-event2}-\ref{last-event2}}\Big] \\
		&\ge  \frac{1}{4|A|}\cdot \E_{r\sim R,(u_1,\ldots,u_n)} \Big[\sum_{i\in \mathcal H}u_i(S_i)  \big|\text{events \ref{first-event2}-\ref{last-event2}}\Big] &&\text{(by Claim \ref{claim-simultaneously})}  \\
		&\ge  \frac{|\mathcal H|}{4|A|)} &&\text{(by (\ref{eq-special-have-value-1}))} \\
		&\ge  \frac{1}{4|A|)}\cdot \frac{m^{2\eps}}{4} &&\text{(by (\ref{lower-bound-h}))} \\
		&\ge m^{\epsilon} \quad \quad \quad &&\text{(for large enough $m$)}
	\end{align*}
	where the equality holds because every $i\in \mathcal H$ is special and wins the special set with probability $1$ conditioned on the events \ref{first-event2}-\ref{last-event2}.  	It completes the proof.
	It remains to show that for every $i\in \mathcal H$ we have that $S_i=A_j$.

	Fix a player  $i\in \mathcal H$. 
	We remind that the messages that the players send given the valuations $(v_1,\ldots,v_n)$ and $(v_1',\ldots,v_n')$  in vertex $x$ of the mechanism are $(z_1,\ldots,z_n)$. By the description of the allocation that we output, $S_i=A_j$ only if the message $z_i$ has a valuation $v_i''$ that sends $z_i$ and satisfies that its dominant strategy guarantees for it a valuable set given the messages $z_{-i}$ at vertex $x$. By definition, $v_i$ sends $z_i$ so showing that $v_i$ has a dominant strategy that guarantees a valuable set will do.

	 We denote the leaf that the mechanism $\mathcal M$ reaches given the action profile $(\mathcal S_1^{dom}(v_1'),\ldots,\mathcal S_n^{dom}(v_n'))$ with $l$. Let $t$ be the subtree that $z_i$ leads to. 
	We remind that by definition, the leaf $l$ is labeled with the allocation $X_i$ and that since player $i$ is happy, we have that $v_i(X_i)=v_i'(X_i)>0$. Denote the payment of player $i$ at leaf $l$ with $p$, and note that by condition \ref{event-positive-profit} we have that:
	\begin{equation}\label{eq-strictly-pos}
	 v_i'(X_i)-p>0 \implies v_i(X_i)-p>0 
	\end{equation}
Also, the fact that  $v_i\equiv v_i'$ implies that: 
$$
(\mathcal S_i^{dom}(v_i'),\mathcal S_{-i}^{dom}(v_{-i}'))\to l \implies
(\mathcal S_i^{dom}(v_i),\mathcal S_{-i}^{dom}(v_{-i}'))\to l
$$
Recall that $i\in \mathcal H$,  so by definition player $i$ is interested in his special set given $v_i'$, and therefore by event \ref{event-all-special-with-alpha}, the weight of $v_i'$ is $\alpha^\ast$, so the weight of $v_i$ is   $\alpha^\ast$ as well. Observe that by construction,
if $v_i$ has weight $\alpha^\ast$, then $v_i=\alpha^\ast(1+\eta_i)\cdot u_i$ where $\alpha^\ast$ is critical for the base valuation $u_i$ and for the noise $\eta_i$. 
We therefore have that player $i$ sends a message $z_i^2$ for the valuation $2\alpha^\ast(1+\eta_i)\cdot u_i$ that differs from $z_i$. 
Denote the subtree that the message leads to with $t_2$, and the leaf that $v_i^2$ reaches given $v_{-i}'$ with $l_2$, i.e.:
\begin{equation*}
(\mathcal S_i^{dom}(v_i^2),\mathcal S_{-i}^{dom}(v_{-i}')) \to l_2
\end{equation*}
where $l_2$ is labeled with the payment and allocation $(S_2,p_2)$ for player $i$. 
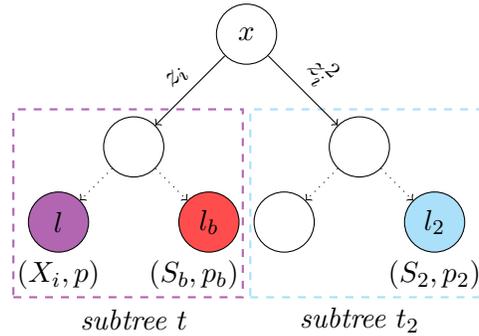
\begin{figure} [h] 
	\centering
	
	\begin{tikzpicture}
		\node[shape=circle,draw=black,minimum size=0.8cm] (r) at (1.5,1.5) {$x$};
		\node[shape=circle,draw=black,minimum size=0.8cm] (v) at (0,0) {};
		\node[shape=circle,draw=black,minimum size=0.8cm,fill=violet!60] () at (-1,-1) {$l$};
		\node[] (bundle) at (-1,-1.7) {$(X_i,p)$};
		\node[shape=circle,draw=black,minimum size=0.8cm] (v') at (3,0) {}; 
		\node[shape=circle,draw=black,minimum size=0.8cm,fill=cyan!30] (l') at (4,-1) {$l_2$};
		\node[] (bundle') at (4,-1.7) {$(S_2,p_{2})$}; 
		\node[shape=circle,draw=black,minimum size=0.8cm] (n') at (2,-1) {}; 
		\node[shape=circle,draw=black,minimum size=0.8cm,fill=red!70] (n) at (1,-1) {$l_b$}; 
		\node[] (badbundle') at (0.8,-1.7) {$(S_b,p_b)$}; 
		
		\draw[thick,violet!60,dashed] (-1.6,0.5) rectangle (1.45,-2) {};
		\node[black,font=\itshape] at (0,-2.3) {subtree $t$};
		
		\draw[thick,cyan!30,anchor=mid west,dashed] (1.55,0.5) rectangle (4.7,-2) {};
		\node[black,font=\itshape] at (3,-2.3) {subtree $t_2$};

		\draw [->] (r) edge  node[sloped, above] {$\footnotesize z_i$} (v);
		\draw [->] (r) edge  node[sloped, above] {$\footnotesize z_i^2$} (v');
		\draw [->] (v) edge[dotted]  node[sloped, above] {}  (l);
		\draw [->] (v') edge[dotted]  node[sloped, above] {} (l');
		\draw [->] (v) edge[dotted]  node[sloped, above] {} (n);
		\draw [->] (v') edge[dotted]  node[sloped, above] {} (n');
		
	\end{tikzpicture}
	\caption{An illustration for the proof of Proposition \ref{prop-alg-is-good}. It describes the tree that the messages of the other players $z_{-i}^u$ induce for player $i$ at vertex $x$. We remind that by definition, the players in $N\setminus \{i\}$ send the same messages $z_{-i}$ for both $v_{-i}$ and for $v_{-i}$. As the figure demonstrates, both subtrees $t,t_2$ have a leaf that is labeled with a bundle and payment such that the profit of player $i$ given the valuation $v_i$ is strictly positive. We will use this structure to show that for every strategy of the other players, the dominant strategy of player $i$ guarantees that he wins a set that is valuable for him.}
	\label{section-g-figure}
\end{figure}
Observe that since $S_i^{dom}(v_i^2)$ is dominant for player $i$, in particular it dominates $\mathcal S_i^{dom}(v_i)$, so we have that $v_i^2(S_2)-p_2 \ge v_i^2(X_i)-p>0$. Therefore, $S_2$ is valuable for player $i$ since the mechanism does not allow negative transfers. Since all valuable sets have the same value for player $i$, we have that $v_i^2(S_2)=v_i^2(X_i)$. Since $v_i^2$ and $v_i$ have the same valuable sets and all valuable sets of each valuation have the same value, we get  that $v_i(S_2)=v_i(X_i)$. 
 Therefore, we have that $p \ge p_2$. Similarly, the fact that $S_i^{dom}(v_i)$ is dominant for player $i$ and not $S_i^{dom}(v_i^2)$ implies that $p_2\ge p$, so $p=p_2$. 

Now, we finally show that $\mathcal S_i^{dom}(v_i)$ guarantees a valuable set at vertex $x$ for player $i$ given the messages $z_{-i}$ of the other players. To this end, assume towards a contradiction that there exists a strategy $\mathcal S_{-i}^b$ and valuations $v_{-i}^b$ of the players in $N\setminus \{i\}$ such that $(\mathcal S_i^{dom}(v_i),\mathcal S_{-i}(v_{-i}^b)))$ 
reaches a leaf $l_{b}$ that is labeled with the bundle and payment $(S_{b},p_b)$ such that $v_i(S_{b})=0$. We do not allow negative transfers, so the profit of player $i$ with valuation $v_i$ given this leaf is necessarily zero. 
Observe that $l_b$ necessarily belongs to subtree $t$. 
See Figure \ref{section-g-figure} for an illustration.

Now, observe the following strategy profile $\mathcal S_{-i}''$: for every valuation $v_{-i}''$, choose the actions specified by both $\mathcal S_{-i}(v_{-i}^b)$ and $\mathcal S_{-i}^{dom}(v_{-i}')$ until vertex $x$, including vertex $x$. 
Note that both strategy profiles are identical up to that point, and that they both dictate sending $z_{-i}$ at vertex $x$. Now, at subtree $t$, choose the actions specified by $\mathcal S_{-i}(v_{-i}^b)$ and in subtree $t_2$ choose the actions specified by $\mathcal S_{-i}^{dom}(v_{-i}')$.  We therefore have that for every $v_{-i}''\in V_{-i}$:
$$
(\mathcal S_i^{dom}(v_i),\mathcal S_{-i}''(v_{-i}''))\to l_b, \quad
(\mathcal S_i^{dom}(v_i^2),\mathcal S_{-i}''(v_{-i}''))\to l_2 
$$
where the $(\mathcal S_i^{dom}(v_i^2),\mathcal S_{-i}''(v_{-i}''))$ reaches $l_2$ because $(\mathcal S_i^{dom}(v_i^2),\mathcal S_{-i}^{dom}(v_{-i}'))$ reaches $l_2$.

We remind that $l_2$ is labeled with $(S_2,p_2)$. We remind that by equation (\ref{eq-strictly-pos}), $v_i(X_i)-p>0$ and also that  $v_i(S_2)=v_i(X_i)$ and  $p=p_2$. Therefore:
$$
v_i(S_2)-p_2 >0 = v_i(S_b)-p_b
$$
Thus, $\mathcal S_i^{dom}$ is not dominant, so we have a contradiction, which completes the proof. 
\end{proof}

\subsection{Properties of Instances}\label{subsec-properties}

\subsubsection{A Diluted Version Exists}
We denote with $L$ the maximum message length in the mechanism $\mathcal M$.

\begin{claim}\label{claim-diluted2}
The probability that a diluted version $I'_{\alpha^\ast}=(v_1',\ldots,v_n')$ of $I_{\alpha^\ast}=(v_1,\ldots,v_n)$ does not exist is at most $\frac 3 n$. 
\end{claim}
\begin{proof}
		For $(v_1',\ldots,v_n')$, denote with $u=(u_1,\ldots,u_n)$ its random base valuations and with $\eta=(\eta_1,\ldots,\eta_n)$ its uniform noises.  
	We want to prove an  upper bound on the probability that there exists a player such that $v_i$ has weight $2\alpha^\ast$ and cannot be replaced. 
	If $L\ge \frac {2^{\Theta(\eps^2 m^\epsilon)}} {n^8}$, then the communication complexity of $\mathcal M$ is larger than $2^{m^{\frac{\eps^2}{2}}}$, and we are done.
	
	 Thus, we can assume that $L<\frac {2^{\Theta(\eps^2 m^\epsilon)}} {n^8}$, so by Claim \ref{claim-change-valuation2} below, the probability that a specific valuation of a player $i$ who is interested in his special set cannot be replaced is at most $\frac 2 {n^4}$. However, we replace only valuations of players
	such that $\alpha^\ast$ is critical for the base valuation $u_i$ and for the noise $\eta_i$,
 so it could be that such valuations are harder to replace.

	We say that $\alpha^\ast$ is   \emph{significant} for player $i$ if the probability that 
	$\alpha^\ast$ is critical for a random base valuation $u_i$ with a uniform noise $\eta_i$
is at least $\frac 1 {n^2}$.
If $\alpha^\ast$ is significant for player $i$, then the probability that he has a valuation $v_i$ that cannot be replaced given that $\alpha^\ast$ is critical for $u_i$ and for $\eta_i$ is (by Bayes theorem) at most   $n^2\cdot \frac 2 {n^4}=\frac 2 {n^2}$.
If $\alpha^\ast$ is not significant for player $i$,
 then the probability that $\alpha^\ast$ is critical for $u_i$ and $\eta_i$
is below $\frac 1 {n^2}$, so clearly the probability that $\alpha^\ast$ is critical for $u_i$ and $\eta_i$  and in addition $v_i$ cannot be replaced is smaller or equal to $\frac 1 {n^2}$.

Thus, the probability that a single player needs to be replaced but cannot be is at most $\frac 3 {n^2}$. By the union bound over all the players, we get that the probability that a diluted version of $I_{\alpha^\ast}$ does not exist is at most $\frac{3}{n}$.
\end{proof}

\begin{claim}\label{claim-change-valuation2}
	Fix a player $i$.
Let $T$ be some base set and let $\mathcal A_j$ be a family of $t = 2^{\Theta(\epsilon^2 m^\epsilon)}$ subsets of size $m^{\eps}$ of $T$ such that each subset is sampled uniformly at random. Let $u$ be a random base valuation with a base set $T$  in which player $i$ is interested in a subset from $\mathcal A_j$ with probability exactly $\frac 1 m$. Choose one of the valuable sets of  $u$ independently at 
random to be the special set.
Let $\eta$ be a noise and denote with $v$ the valuation $2\alpha^\ast(1+\eta)\cdot u$.
 If 
 $L<\frac {2^{\Theta(\eps^2 m^\epsilon)}} {n^8}$, then with probability at least $1-\frac {2} {n^4}$, there is another valuation $v'$, also with  weight $2\alpha^\ast$, in which the player is not interested in this special set and the messages sent by $v$ and $v’$ in vertex $x$ of the mechanism $\mathcal M$ are the same.
\end{claim}
\begin{proof}
Fixing the base set $T$, and the weight $2\alpha^\ast$, we represent each valuation of player $i$ as in the statement by a  vector in $\{0,1\}^t$ that has $k$ coordinates that are equal to $1$, where by Chernoff bounds $k\in \{\frac{t}{2m},\ldots,\frac{3t}{2m}\}$ with probability of at least $1-2\exp(\frac{-t}{12m})=1-2\exp(\frac{-2^{\Theta(\epsilon^2 m^\epsilon)}}{12m})$. If two valuations differ only in their noise, then they have the same representation. 



Assume that the valuation $v$ is interested in some set $S$. Suppose that for this specific $v$ there is no $v'$ as in the statement of the claim. In other words, all vectors that send the same message as $v$ and have weight $2\alpha^\ast$ satisfy that they have a value of $1$ in the coordinate that corresponds to the set $S$. In this case, we say that this specific set $S$ is \emph{static} for the message.

A message of player $i$ is called \emph{dangerous} if there are at least $\frac k {n^4}$ static sets for it given the weight $\alpha$ and the noise $\eta$. Note that if a message is dangerous this implies that the location of $\frac k {n^4}$ coordinates with value $1$ in the vector is already determined. Therefore, the probability that a random base valuation with weight $2\alpha^\ast$  makes the player send a specific dangerous message is at most $\frac 1 {m^{k/n^4}}$, since each of the static sets is  valuable with probability $\frac{1}{m}$. Since the number of possible messages is at most $2^L$, the probability that we sample a random base valuation that sends some dangerous message is at most $2^L \cdot \frac 1 {m^\frac{k}{n^4}}$. 
Note that $L\le \frac{t}{n^8}\le \frac{k\cdot 2m}{n^8}\le \frac{k}{n^4}$, 
so the probability that player $i$ sends 
a dangerous message is exponentially small.

Conditioned on this event not happening (i.e., that the player does not send a dangerous message), the probability that the special set of a player is not static -- assuming that the player is interested in his special set -- is at least $1-\frac 1 {n^4}$, since each $1$ coordinate has the same probability of being the special set. The claim follows.
\end{proof}
\subsubsection{The Diluted Version Has High Welfare}
\begin{claim}\label{claim-diluted-has-high-welfare}
	Sample an instance $I_{\alpha^\ast}=(v_1,\ldots,v_n)$ from the distribution $\mathcal I_{\alpha^\ast}$ and let $I_{\alpha^\ast}'=(v_1',\ldots,v_n')$ be its diluted version.  Then,
	with probability of at least $\frac{1}{3|A|}$, 
	 the optimal welfare of $(v_1',\ldots,v_n')$ is at least $\alpha^\ast\cdot m^{1-2\epsilon}$. 
\end{claim}
\begin{proof}
	For $(v_1',\ldots,v_n')$, denote with $u=(u_1,\ldots,u_n)$ its random base valuations and with $\eta=(\eta_1,\ldots,\eta_n)$ its uniform noises.  
	We remind that $P_{\alpha^\ast}(u,\eta)$ is the subset of special bidders that send different messages in vertex $x$ of the mechanism for $\alpha^\ast(1+\eta_i)\cdot u_i$ and for $2\alpha^\ast(1+\eta_i)\cdot u_i$. Now, we define another subset of bidders:
	$$
	\mathcal B=\{ i\in N\hspace{0.25em}|\hspace{0.25em} v_i'=\alpha^\ast(1+\eta_i)\cdot u_i,\hspace{0.25em} i\in P_{\alpha^\ast}(u,\eta)\}
	$$
	
	For  a lower bound on the optimal welfare of $(v_1',\ldots,v_n')$, consider the allocation $(X_1,\ldots,X_n)$ where every player in $\mathcal B$ gets his special set. It is feasible because the special sets are disjoint. Observe that:
\begin{equation} \label{eq-opt-welfare}
OPT(v_1',\ldots,v_n')\ge \sum_i v_i'(X_i) \ge |\mathcal{B}|\cdot \alpha^\ast
\end{equation}
	Now, observe that every player in $P_{\alpha^\ast}(u,\eta)$ belongs in $\mathcal B$ with probability $\frac{1}{2}$. 
		Therefore the expected welfare of $(X_1,\ldots,X_n)$  is at least  $\frac{|\mathcal{B}|\cdot \alpha^\ast}{2}$. 
		Now, if we assume that  $P_{\alpha^\ast}(u,\eta)\ge \frac{m^{1-\epsilon}-1}{|A|}$, then   
by Chernoff bounds with probability at least $1-\exp(-m)$,  we have that
$
		  |\mathcal B| \ge \frac{m^{1-\epsilon}-1}{4|A|}
	$. 
		 By Lemma \ref{lemma-good-alpha}, $P_{\alpha^\ast}(u,\eta)\ge \frac{m^{1-\epsilon}-1}{|A|}$ with probability of at least $\frac{1}{2|A|}$. 
		 By the law of total probability,  $|\mathcal B| \ge \frac{m^{1-\epsilon}-1}{4|A|}$ with probability of at least $\frac{1}{3|A|}$.
		 By inequality (\ref{eq-opt-welfare}), it means that with probability at least $\frac{1}{3|A|}$,
		  the optimal welfare for $(v_1',\ldots,v_n')$ is at least $\alpha^\ast \cdot \frac{m^{1-\epsilon}-1}{4|A|}$, which is larger than $\alpha^\ast\cdot m^{1-2\epsilon}$.   		 
\end{proof}
\subsubsection{Special Bidders in the Diluted Version Have Weight $\alpha^\ast$}
\begin{claim} \label{claim-all-special-with-alpha}
		Sample an instance $I_{\alpha^\ast}=(v_1,\ldots,v_n)$ from the distribution $\mathcal I_{\alpha^\ast}$ and let $I_{\alpha^\ast}'=(v_1',\ldots,v_n')$ be its diluted version. Then, 
	with probability of at least $1-\frac{m^2}{\exp(m)}$,  every player with weight $2\alpha^\ast$ satisfies that he is not interested in the special set of his group.
%
%
\end{claim}
\begin{proof}
	First,
	denote with $u=(u_1,\ldots,u_n)$ the random base valuations and with $\eta=(\eta_1,\ldots,\eta_n)$ the uniform noises  that $(v_1,\ldots,v_n)$ and $(v_1',\ldots,v_n')$ are based on. 
	  	Fix a player $i$ that belongs to group $G_j$ and satisfies that the weight of $v_i'$ is 
$2\alpha^\ast$. We want to show that with high probability $v_i'(A_j)=0$, where $A_j$ is the special set of group $j$. First of all, observe that since $(v_1',\ldots,v_n')$ is diluted,  if $\alpha^\ast$ is critical for $u_i$ and for $\eta_i$, then $v_i'(A_j)=0$, and we are done. Similarly, if $u_i(A_j)=0$, then we are also done. 

Thus, we only need to handle the case where  $\alpha^\ast$ is not critical for $u_i$ and for $\eta_i$.
 Recall that by construction, every player $i$ satisfies that there are $t=2^{\Theta(\epsilon^2\cdot m^\epsilon)}$  subsets  of a base set $T=A_j\cup B$ of size $m^\epsilon$ such that each subset is valuable for player $i$ with probability of exactly $\frac{1}{m}$. Therefore, by Chernoff bounds, the number of subsets of size $m^\epsilon$ that are  valuable for $u_i$ is at least $\frac{t}{2m}$ with probability that of at least $1-\exp(-\frac{t}{8m})$. 
	
	We remind that since $\alpha^\ast$ is not critical for $u_i$ and $\eta_i$, by construction $v_i$ has only one valuable set that is sampled uniformly at random from the valuable sets of the base valuation $u_i$. Therefore, the probability that the chosen set is $A_j$ is at most $\frac{2m}{t}$. By applying the union bound over all players, we get the claim.
\end{proof}
\subsubsection{Positive Profit}
\begin{claim}\label{claim-positive-profit2}
		Sample an instance $I_{\alpha^\ast}=(v_1,\ldots,v_n)$ from the distribution $\mathcal I_{\alpha^\ast}$ and let $I_{\alpha^\ast}'=(v_1',\ldots,v_n')$ be its diluted version. Then,
		the probability that there is a player that is allocated a set with non-zero value by the mechanism and has  profit of  $0$ is at most $\frac {n} {3^n}$. 
\end{claim} 
\begin{proof}
	Consider some player $i$ with valuation $v_i'$ in the instance $I_{\alpha^\ast}=(v_1',\ldots,v_n')$. By the taxation principle, for every bundle $S$ there exists a price $p_S$ (that does not depend on $v_i'$) such that if all players play their dominant strategies, player $i$ is assigned a bundle $T$ in $\arg\max_S v_i'(S)-p_S$. Recall that all bundles with non-zero value of player $i$ have the same value (denoted $y_i$). Let $p_i$ denote the minimal price of a bundle that player $i$ has a non-zero value for (so  by the taxation principle, $p_i$ is the price that player $i$ pays if he wins a bundle with non-zero value). Since $p_i$ does not depend on $v_i'$ and since, fixing the value of $\alpha$, $y$ has $3^n$ possible values, the probability that $y_i=p_i$ is at most $\frac 1 {3^n}$. It implies that the probability that $\mathcal M$ allocates a non-zero bundle to player $i$ and $y_i=p_i$ is at most $\frac 1 {3^n}$. 
	By the union bound, the probability that there exists  some player $i'$ that wins a valuable bundle and has $0$ profit is at most $\frac n {3^n}$, so the claim follows.
\end{proof}

\subsubsection{Feasible Allocation}
\begin{claim}\label{claim-feasible-allocation}
	With probability $1-\frac{m^3}{\exp(m^\epsilon)}$, the allocation  $(S_1,\ldots,S_n)$ is feasible. 
\end{claim}
\begin{proof}
	Note that we only allocate the special set and  we only allocate every special set $A_j$ to bidders from the group $G_j$. By the construction, the special sets are disjoint from each other, so it suffices to show that no special set is allocated to more than one player in each group. 
	
	Let $i_1,i_2$ be two players that belong in the group $G_j$. We will show prove an upper bound on the probability of the event that $S_{i_1}=S_{i_2}=A_j$, i.e., the event where both players win the special set.
	Recall that the messages of the players given $(v_1,\ldots,v_n)$  are $(z_1,\ldots,z_n)$. 
	We will analyze the tree that the messages $z_{-i_1}$ induce for player $i_1$ at vertex $x$ and the tree that the messages $z_{-i_2}$ induce for player $i_2$ at vertex $x$.  
	
	Since $S_1=A_j$, by definition the message $z_1$ sent by player $i_1$ has valuation $v_1$ whose dominant strategy in the mechanism $\mathcal M$ guarantees for it a valuable set at the initial vertex given that the other players send the messages $z_{-i_1}$. Similarly, the message $z_2$ sent by player $i_2$ has a valuation $v_2$ whose dominant strategy guarantees a valuable set at the initial vertex if  the other players send the messages $z_{-i_1}$. 
	Observe that $i_1,i_2$  belong in the same group $G_j$, so by the construction all their valuable sets intersect with probability at least $1-\frac{1}{\exp(m^\epsilon)}$. Thus, with probability at least $1-\frac{1}{\exp(m^\epsilon)}$,  it is not possible that the dominant strategies of both of them  given the message profile $(z_1,\ldots,z_n)$ at vertex $x$ of the mechanism $\mathcal M$ guarantee a valuable set, so by the construction of the allocation $(S_1,\ldots,S_n)$, $S_{i_1}=S_{i_2}=A_j$ with probability of at most $\frac{1}{\exp(m^\epsilon)}$.   
	
	 By taking the union bound on all pairs of players of each group $j$, and on all groups, we get that the probability at least $1-\frac{m^3}{\exp(m^\epsilon)}$, the special set of each group is allocated at most once, which completes the proof.
%
%
\end{proof}

\subsubsection{Proof of Lemma \ref{lemma-first-round2}}\label{subsec-first-round-proof2}

\begin{definition}\label{def-initial-vertex}
	We call a vertex in the protocol tree of $\mathcal M$ an \emph{initial vertex} if it satisfies that there exists a consecutive range of powers of $2$, $C\subseteq \{1,2,\ldots,2^{2^m}\}$, such that:
	\begin{enumerate}
		\item For every profile of valuations $(v_1,\ldots,v_n)$ with weights in $C$, the mechanism reaches the initial vertex given $(\mathcal S_1^{dom}(v_1),\ldots,\mathcal S_n^{dom}(v_n))$.
		\item There exists a player $i$ with valuations $v,v'$ with weights $\alpha,\alpha'$ that sends different messages in the initial vertex for $v,v'$, where $m^4\cdot \min C \le \alpha,\alpha'\le \frac{\max C}{m^4}$.  
	\end{enumerate}	 
	We call the set $C$ the \emph{initial set of weights}.
\end{definition}

\begin{lemma}\label{lemma-initial-exists}
	The mechanism $\mathcal{M}$ has an initial vertex.
\end{lemma}
\begin{proof}
	Observe the first vertex of the mechanism $\mathcal M$ where there exists a player $i$ with valuations $v_i,v_i'$ such that his dominant strategy dictates different messages for them. Such vertex necessarily exists because otherwise all players send the same message throughout all of the protocol, so it has no approximation guarantee. Now, if the weights of $v_i,v_i'$ belong in 
	$\{2^{4\log m}=m^{4},\ldots,2^{2^{m}-4\log m}=\dfrac{2^{2^m}}{m^4}\}$, we are done. Otherwise,  observe that the approximation ratio of $\mathcal M$ guarantees that there exists a vertex with some player $i$ that sends different messages for two valuations $v_i,v_i'$ with weights in	$\{2^{4\log m}=m^{4},\ldots,2^{2^{m}-4\log m}=\dfrac{2^{2^m}}{m^4}\}$. If the weight of $v_i,v_i'$ are in $\{m^8,\ldots, \dfrac{2^{2^{m}}}{m^8}\}$, then we are done. We continue this iterative process until we reach an initial vertex.

	We now explain why we can safely assume that an initial vertex exists.
	If $\mathcal M$ has communication complexity larger than $poly(\frac{2^{m^{\frac{\epsilon^2}{2}}}}{n})$, then Theorem \ref{theorem-general-impossibility} holds and  we are done. Thus, we assume that the communication complexity is at most $poly(\frac{2^{m^{\frac{\epsilon^2}{2}}}}{n})$, so the number of meaningful rounds (i.e. rounds where there exist players that send different messages for different valuations) is at most $poly(\frac{2^{m^{\frac{\epsilon^2}{2}}}}{n})$, since in each meaningful round at least one bit is sent.  
	If an initial vertex does not exist, after $poly(\frac {2^{m^{\frac {\eps^2} 2}}} n)< \frac{2^{m}}{8\log m}$ iterations, we reach a leaf that contains all base valuations with weight $2^{2^{m-1}}$. By Proposition \ref{prop-simul-hard-dist}, outputting the same allocation for all those valuations does not provide a $m^{1-4\eps}$ approximation to the optimal welfare, so we reach a contradiction.   
\end{proof}
We denote the initial vertex that is guranteed by Lemma \ref{lemma-initial-exists} with $x$. 
\begin{definition}
	\label{def-neighboring2}
	We say that two different valuations $w,w'$ with weights $\alpha, \alpha'$ and noises $\eta,\eta'$  
	respectively, are \emph{neighboring} if $\frac \alpha {\alpha'}\in \{\frac 1 2,1, 2\}$ and \emph{at most one} of the following conditions holds:
	\begin{enumerate}
		\item $|\eta-\eta'|= \pm \frac{1}{3^n}$    
		\item There exist items $j,j'$ such that $w,w'$ permute them, in the sense such that for every subset $S$, either $w(S)=w'(S)$ or $w(S)=w'(S')$ where $S'=S-\{j\}+\{j'\}$.
		\item $w,w'$ differ only on one coordinate, i.e. 
		there exist a subset of items $T$ such that for every subset $S$ other than $T$, $w(S)=w'(S)$. 
	\end{enumerate} 
\end{definition}
\begin{claim}\label{claim-neighboring2}
	Let $v_i,v_i'$ be two valuations of player $i$ with weights $\alpha,\alpha'\in \{m^4\cdot \min C,\ldots, \frac{\max C}{m^4}\}$ and noises $\eta,\eta'$ that send different messages in vertex $x$, the initial vertex. Then, there are two \emph{neighboring} valuations $w,w'$ of player $i$ with weights in $\{m^4\cdot \min C,\ldots,\frac{\max C}{m^4}\}$ that send different messages in vertex $x$. 
\end{claim}
\begin{proof}
	Assume without loss of generality that $\alpha'\geq \alpha$. Define a sequence of valuations $v=w_0, w_1, \ldots,\allowbreak w_k=v'$ as follows.   
	Each valuation $w_r$ is obtained from the valuation $w_{r-1}$ by one of the following operations: first, $w_r=2w_{r-1}$ until the weight of $w_{r}$ is $\alpha'$.
	Afterwards, if the noise $\eta$ of $w_r$ differs from the noise of $v'$, obtain $w_r$ from $w_{r-1}$ by increasing or decreasing the noise of $w_{r}$ by $\frac{1}{3^n}$ at a time, until the noise of $w_r$ is equal to $\eta'$.
	
	Then, suppose that the base sets of $w_{r-1}$ and $v'$ are different ($T$ is the base set of $w_{r-1}$, $T'$ is the base set of $v'$). Obtain $w_{r}$ from $w_{r-1}$ by renaming some item $j\in T \setminus T'$ to some item $j'\in T' \setminus T$. When the bases of the valuations are the same but the valuations are not identical, there is at least one set $S$ such that
	$w_{r-1}(S)\neq v'(S)$. Take such a set $S$ and define $w_r$ to be identical to $w_{r-1}$, except that $w_r(S)=v'(S)$. 
	
	By construction, each pair of valuations $w_r,w_{r+1}$ consists of neighboring valuations. Since the messages that $w_0$ and $w_k$ send are not the same, there must be one pair $w_r,w_{r+1}$ in which each valuation sends a different message. The claim follows since $w_r,w_{r+1}$ are neighboring valuations.
\end{proof}

We say that a player is \emph{insensitive} 
with respect to some base  valuation $w$, noise $\eta$  and a set of weights $A$ if he 
sends the same message for all the valuations in $\{\alpha(1+\eta)\cdot w\}_{\alpha\in A}$ in the initial vertex $x$ of the mechanism $\mathcal M$. 

\begin{lemma}\label{lemma-random2}
	Let $v_i^1,v_i^2$ be two neighboring valuations  of some player $i$ that send different messages in the initial vertex $x$. 
	Denote their base sets  with $T_1,T_2$ and their weights with $\alpha \le \alpha'$ 
	(respectively). 
	Let $W$ be some base set that is obtained by choosing  $m^{\epsilon}$  items from $T_1\cap T_2$ and $m^{\epsilon}$ items arbitrarily.\footnote{The intersection of $T_1,T_2$ is larger than $m^{\epsilon}$ because they are neighboring.} Fix some player $i'\neq i$ and sample a  random base valuation $w$ with a base set $W$. For every noise $\eta$, with probability at least $1-\frac 1 {exp(m^{\eps})}$ (over the construction of the random base valuation), $(1+\eta)\cdot w$ is sensitive in the set $A'=\{\frac{\alpha}{m^4},\frac{2\alpha}{m^4},\ldots,\alpha,\ldots,\frac{m^4\cdot \alpha}{2},m^4\cdot \alpha\}$.
\end{lemma}
\begin{proof}
	We will show that if $(1+\eta)\cdot w$ is not sensitive in $A$, then there exists a valuation of player $i$ that has no dominant strategy. 
	
	Assume towards a contradiction that player $i'$ sends the same message for all the valuations $\{\alpha(w+\eta)\}_{\alpha\in A'}$. Denote this message with $z_{i'}$. 
	For every player $j\neq i,i'$, let $v_j$ be an arbitrary valuation with weight $\frac{\alpha}{m^4}$ and noise $0$. Denote with $z_j$ the message that the dominant strategy $\mathcal{S}_j(v_j)$ dictates at the initial vertex.   
	
	We will analyze the tree that the message profile $z_x^{-i}$ induces for player $i$ at vertex $x$. Denote with $t_1$ and $t_2$ the subtrees that $v_i^1$ and $v_i^2$ lead to. Denote the noises of $v_i^1,v_i^2$ with $\eta_1,\eta_2$. 
	Observe that the subtree $t_1$ has a leaf $l_1$ that is labeled with a bundle $S_1$ such that $v_i(S_1)=\alpha$. 
	The reason for it is that if the valuation of player $i'$ is $\frac{\alpha}{m^4}\cdot (1+\eta)\cdot w$ and the valuation of player $i$ is $v_i^1$, then 
	the welfare is at least $\alpha$ when player $i$ wins a valuable set for $v_i^1$, whereas every allocation where player $i$ does not win a valuable set  has welfare at most
	$m\cdot \frac{\alpha}{m^4}+\eta\cdot \frac{\alpha}{m^4}\le (1+\eta)\cdot \frac{\alpha}{m^3}\le \frac{2\alpha}{m^3}$ 
	(there are at most $m$ items so no more than $m$ players can be satisfied, and only $i'$ has non-zero noise). The price $p_1$ of $S_1$ is at most $\alpha\cdot (1+\eta_1)< 2\alpha$.
	The reason for it is
	that the taxation principle guarantees that each player wins his most profitable bundle, and the price of the empty bundle is $0$ since $\mathcal M$ is normalized, so it necessarily holds that $v_1(S_1)\ge p_1$. 
	Due to the same reasons, the subtree $t_2$ has a leaf $l_2$ that is labeled with $(S_2,p_2)$ such that $v_i^2(S_2)=\alpha'(1+\eta_2)$ and $p_2\le \alpha'(1+\eta_2) \le 4\alpha$. 
	We therefore have that:
	\begin{equation}
		\begin{gathered}\label{eq-strategy-desc}
			(\mathcal S_i^{dom}(v_i^1),\mathcal S_{i'}^{dom}(\frac{\alpha}{m^4}\cdot (1+\eta)\cdot w),\{\mathcal S_j^{dom}(v_j)\}_{j\neq i,i'})\to l_1,  \\ 
			(\mathcal S_i^{dom}(v_i^2),\mathcal S_{i'}^{dom}(\frac{\alpha}{m^4}\cdot (1+\eta)\cdot w),\{\mathcal S_j^{dom}(v_j)\}_{j\neq i,i'})\to l_2
		\end{gathered}
	\end{equation}
	Note that $S_1,S_2$ might be larger than $m^\epsilon$, so we denote with $X_1\subseteq S_1,X_2 \subseteq S_2$ two  valuable sets of $v_i^1,v_i^2$   of size $m^\epsilon$. I.e.,  $v_i^1(X_1)=v_i^1(S_1)$ and  $v_i^2(X_2)=v_i^2(S_2)$.   
	
	We now define a valuation $v$ as follows. $v$ has weight $8\alpha$, noise of zero 
	and its base valuation has  a base set that contains $X_1\cup X_2$ (if $X_1,X_2$ intersect, add more items arbitrarily). We define the valuable sets of $v$ to be $X_1,X_2$. We will now show that $v$ has no dominant strategy.
	Assume towards a contradiction that $v$ does have a dominant strategy, and 
	that it dictates at vertex $x$ the message $z$ that leads to subtree $t$. 
	Observe the following valuation profile: the valuation of player $i$ is $v$, the valuation of player $i'$ is $m^4 \alpha\cdot (1+\eta)\cdot w$ and the valuation of every player $j\neq i,i'$ is the low-weight valuation $v_j$. For this valuation profile, mechanism $\mathcal M$ has to satisfy player $i'$ for approximation of $m^{1-4\eps}$: the optimal welfare is at least $m^4\cdot \alpha$, whereas the welfare of any allocation that does not satisfy player $i'$ is at most $(n-1)\cdot\frac{\alpha}{m^4}\le \frac{\alpha}{m^2}$. 
	\begin{claim}\label{claim-intersect}	
		With probability at least $1-\frac{1}{\exp(m^\eps)}$, all the valuable sets of the valuations $v$ and $m^4 \alpha\cdot(1+\eta)\cdot  w$ intersect. 
	\end{claim}
	By Claim \ref{claim-intersect} (which we prove later on), the fact that 
	the mechanism $\mathcal M$ allocates to player $i'$ a valuable set given 
	this valuation profile, implies that player $i$ wins the empty bundle or a set that he is not interested in. In other words:
	\begin{gather}\label{eq-strategy-desc2}
		(\mathcal S_i^{dom}(v),\mathcal S_{i'}^{dom}(m^4 \alpha \cdot (1+\eta)\cdot w),\{\mathcal S_j^{dom}(v_j)\}_{j\neq i,i'})\to l
	\end{gather}
	where $l$ is labeled with a bundle $(S,p)$ such that $v(S)=0$. 
	
	Note that the subtree $t$
	is a different subtree than either $t_1$ or $t_2$.
	Assume without loss of generality that $t\neq t_2$. 
	We now construct the following strategy $\mathcal{S}_{-i'}$ for every valuation of player $i'$:
	Until vertex $x$ (including vertex $x$), follow the strategy dictated by $\mathcal{S}_{i'}^{dom}(\frac{\alpha}{m^4}\cdot(1+\eta)\cdot w)$. At subtree $t_2$, follow this strategy as well. In contrast, in subtree $t$, choose the strategy specified by $\mathcal{S}_{i'}^{dom}(m^4 \alpha\cdot(1+\eta)\cdot  w)$.

	Observe that if player $i'$ follows the strategy $S_{i'}(v_{i'})$ (where $v_{i'}$ is arbitrary),  and every player $j\in N\setminus \{i,i'\}$ follows $\mathcal S_j^{dom}(v_j)$, then:
	\begin{gather*}
		(\mathcal S_i^{dom}(v),\mathcal S_{i'}(v_{i'}),\{\mathcal S_j^{dom}(v_j)\}_{j\neq i,i'})\to l, \\
		S_i^{dom}(v_i^2),\mathcal S_{i'}(v_{i'}),\{\mathcal S_j^{dom}(v_j)\}_{j\neq i,i'})\to l_2
	\end{gather*}
	This is immediate from the description of the strategy $\mathcal{S}_{-i'}$ and from (\ref{eq-strategy-desc})
	and (\ref{eq-strategy-desc2}).
	Note that $v(S_2)-p_2 \ge 8\alpha -4\alpha = 4\alpha > 0 = v(S)-p$ (here we use the no-negative transfers assumption). Therefore, $\mathcal S_i$ is not dominant for player $i$ given the valuation $v$, which is a contradiction. 
\end{proof}
\begin{proof}[Proof of Claim \ref{claim-intersect}]
	Note that by the construction, both $v_1$ and $v_2$ satisfy that half of their base sets intersect with the base set of $w$.  Thus, according to the explanation provided in Subsection \ref{impossible-sim-general-subsec}, we have that
	with probability $1-\frac{1}{exp(m^\epsilon)}$, all the valuable sets of $v_1$ intersect with all the valuable sets of $w$. The same holds for $v_2$. In particular, it implies that the valuable set $X_1$ intersects with all the valuable sets of $w$ with high probability.  Similarly, $X_2$ intersects with all the valuable sets of $w$ with probability $1-\frac{1}{exp(m^\epsilon)}$. By the union bound, we have that both  $X_1$ and $X_2$ intersect with the valuable sets of $w$ with probability $1-\frac{1}{exp(m^\epsilon)}$. Those are the valuable sets of the valuation $v$, so it completes the proof.  
\end{proof}
Lemma \ref{lemma-first-round2} can now be deduced. 
Let vertex $x$ be the initial vertex of the mechanism $\mathcal M$ (whose existence is guaranteed by Lemma \ref{lemma-initial-exists}) and let $C$ be its initial set of weights. 
By Definition \ref{def-initial-vertex}, there  exists a player $i$ that has two valuations $v_i^1,v_i^2$ 
with weights $\alpha_1 \le \alpha_2$ in  $\{m^4\cdot \min C,\ldots,\frac{\max C }{m^4}\}$ that send different messages in vertex $x$. By Claim \ref{claim-neighboring2}, we can assume that $v_i^1,v_i^2$ are neighboring. Denote their base sets with $T_1,T_2$, respectively. 

We want to show that if we sample a  random base valuation $v'$ of player $i'$ with  base set $T'$ and noise $\eta'$, then with high probability it is sensitive in the set $\{\frac{\alpha}{m^{8}},\frac{2\alpha}{m^{8}},\ldots,\alpha,\ldots,m^{8}\cdot \alpha\}$, where $\alpha=\min\{\alpha_1,\alpha_2\}$.  

Obtain a base set $T''$ by choosing at random half of the items of $T_1\cap T_2$ and half of the items of $T'$ (if $T_1\cap T_2$ and $T'$ intersect, we might need  to arbitrarily add more items to make sure that $|T''|=2m^\eps$). Sample a random base valuation $u''$ for some player $i''\neq i,i'$ with base set  $T''$ and an arbitrary noise $\eta''$. By Lemma \ref{lemma-random2}, with probability at least $1-\frac 1 { exp(m^{\eps})}$, $(1+\eta'')\cdot u''$ is sensitive for player $i''$ on the set  $\{\frac{\alpha}{m^4},\frac{2\alpha}{m^4},\ldots,\alpha,\alpha',\ldots,m^4\cdot \alpha\}$. 

Thus, there exist two consecutive weights $\beta\le \beta'$ such that player  $i''$ sends a different messages when his valuation is
$\beta(1+\eta'')\cdot v''$ and $\beta'(1+\eta'')\cdot v''$.
Note that $\beta(1+\eta'')\cdot v''$ and $\beta'(1+\eta'')\cdot v''$ are neighboring, so by applying Lemma \ref{lemma-random2} again we get that 
$(1+\eta')\cdot w$
is sensitive on the set 
$\{\frac{\beta}{m^4},\frac{2\beta}{m^4},\ldots,\beta,\ldots,m^4\cdot \beta\}$
with  probability at least $1-\frac 1 {\exp(m^{\eps})}$. Recall that  $\frac \alpha {m^4}\leq \beta  \leq m^4\cdot \alpha $, so  we get Lemma \ref{lemma-first-round2} for every player other than player $i$ with respect to the set 
$\{\frac{\alpha}{m^8},\frac{2\alpha}{m^8},\ldots,\alpha,\ldots,m^8\cdot \alpha\}$. 

The proof for player $i$ follows similar lines. 
Fix a base set $T_i$ and noise $\eta_i$ and sample a random base valuation $u_i$ with base set $T_i$. 
Obtain a base set $T''$ by taking half of the items of $T_i$ and $m^\eps$ items arbitrarily. 
Let $v_{i''}$ be 
a random base valuation  for player $i''\neq i$ with base $T''$ and an arbitrary noise $\eta''$. Since we have proven Lemma \ref{lemma-first-round2}
for all players except player $i$, we have that with probability $1-\frac 1 {\exp(m^{\eps})}$, there exist two consecutive $\gamma\le \gamma'\in \{\frac{\alpha}{m^8},\frac{2\alpha}{m^8},\ldots,m^8\cdot \alpha\}$, such that player $i''$ sends different messages at vertex $x$ for $\gamma(1+\eta'')\cdot v_{i''}$ and for $\gamma'(1+\eta'')\cdot   v_{i''}$ in vertex $x$ of the mechanism $\mathcal{M}$.  Observe that the valuations $\gamma(1+\eta'')\cdot v_{i''}$ and $\gamma'(1+\eta'')\cdot v_{i''}$  are neighboring and that $|T'' \cap T_i|\ge m^\eps$, so by  Lemma \ref{lemma-random2} the  valuation $(1+\eta_i)\cdot u_i$ is sensitive with respect to 
$\{\frac{\gamma}{m^8},\frac{2\gamma}{m^8},\ldots,\gamma,\ldots,m^8\cdot \gamma\}$
with probability $1-\frac 1 {\exp(m^{\eps})}$. Since $m^{8}\cdot \alpha \ge \gamma \ge \frac{\alpha}{m^8}$, we have that  Lemma \ref{lemma-first-round2} holds for 
$\{\frac{\alpha}{m^{16}},\frac{2\alpha}{m^{6}},\ldots,m^{16}\cdot \alpha\}$.


\bibliographystyle{alpha}
\bibliography{dominant}

\appendix

\section{Dominant Strategy Implementations for Three Players}\label{separation-sec}
Consider two player combinatorial auctions with general valuations. 
 In \cite{D16b} it is shown that every social choice function that is implementable in ex-post equilibrium can also be implemented in dominant strategy equilibrium with only a polynomial blow up in the communication. Whether a similar result holds also for three players was left as an open question. We answer it as follows:

\begin{theorem}\label{separationthm}
	There exists a three player social choice function $f$ for combinatorial auctions with general valuations such that: 
	\begin{enumerate}
		\item There exists a mechanism that implements $f$ in an ex-post Nash equilibrium with $\mathcal{O}(m)$ bits.
		\item The communication complexity of  every dominant strategy implementation of $f$ is $\exp(m)$.    
	\end{enumerate}
\end{theorem}
\begin{proof}
	Consider the following social choice function. There are three players, Alice, Bob and Charlie. The set of items is $M$, where $|M|=m$. For each player, the values of all bundles are integers in  $\{1,\ldots,2^{m}\}$.
	Fix three items $a,b,c \in M$. 
	Let $\mathcal X = \{X_1,\ldots,X_{\binom{m}{\frac{m}{2}}}\}$ be the collection of all subsets of $M$ of size $\frac{m}{2}$. 
	We define $f$ as follows:
	\begin{itemize}
		\item 
		If $v_A(\{c\})> \binom{m}{\frac{m}{2}}$, then item $c$ is not allocated. Else, let $X_{C}=X_{v_A(\{c\})}$, i.e. $X_{C}$ is the bundle in $\mathcal{X}$ whose index is equal to $v_A(\{c\})$. If $v_B(X_{C})< 1$, then Charlie wins $c$. Otherwise, $c$ is not allocated. 
		
				\item 
		If $v_B(\{a\})> \binom{m}{\frac{m}{2}}$, then item $a$ is not allocated. Else, let $X_{A}=X_{v_B(\{a\}}$. If $v_C(X_{A})< 1$, then Alice wins $a$. Otherwise, $a$ is not allocated. 
		
			\item	If $v_C(\{b\})> \binom{m}{\frac{m}{2}}$, then item $b$ is not allocated. Else, let $X_{B}=X_{v_C(\{b\})}$. If $v_A(X_{B})< 1$, then Bob wins $b$. Otherwise, $b$ is not allocated. 
		\end{itemize}
	Note that there are only three items $a,b,c$ that $f$ ever allocates. 
 The following protocol realizes $f$ in ex-post Nash equilibrium. All three players send $v_A(a),v_B(b)$ and $v_C(c)$ respectively, so the players can compute $X_{A}$, $X_{B}$ and $X_{C}$. Afterwards, Alice, Bob and Charlie each send a bit that indicates whether $v_A(X_{B})$, $v_B(X_{C})$ and $v_C(X_{A})$ are strictly larger than $1$ or not. Regardless of the outcome, each player always pays $0$.   
 The communication cost of the protocol is $\mathcal{O}(m)$ since each player sends his value for one bundle plus one additional bit. 

The protocol achieves an ex-post Nash equilibrium since the valuation of each player has no impact at all on the items that she wins. For example, every $v_B,v_C\in V_B\times V_C$ satisfy that for all $v_A\in V_A$, either $f(\cdot,v_B,v_C)$ allocates $a$ to Alice for every valuation $v_A\in V_A$, or $f(\cdot,v_B,v_C)$ constantly outputs the empty allocation for Alice. Thus,  a payment of zero indeed guarantees an ex-post equilibrium. 

	We now show hardness of dominant strategy implementations. The intuition for that is is that determining whether Alice wins item $a$ is the INDEX function in disguise, where Bob holds the index and Charlie holds the array. Thus, by the reduction, we have that after the first round of communication, it is not known whether Alice wins item $a$ or not, as long as the number of bits sent is $poly(m)$. The same applies to Bob and Charlie as well.   We are going to take advantage of this fact in to show that the first player who sends a \textquote{meaningful} message necessarily does not have a dominant strategy. We formalize it as follows.

The first step is to reduce $f$ to the well-known function $INDEX_k$. We use the reduction to prove that the simultaneous communication complexity of the social function $f$ is at least $\exp(m)$. Afterwards, we show that simultaneous hardness implies hardness of  dominant strategy implementation of $f$.

$INDEX_k:\{0,1\}^k\times \{1,\ldots,k\}\to \{0,1\}$ is a function where one player holds an array $arr\in \{0,1\}^k$ and the other player holds an index $j\in [k]$. $INDEX(arr,j)$ is equal to $arr[j]$, i.e. the value of the $j'$th bit in $arr$.  
\begin{theorem}(\cite{kremer1999randomized}) \label{index-hard-theorem}
	The simultaneous communication complexity of $INDEX_k$ is $\Omega(k)$.
\end{theorem}
	\begin{proposition}\label{fahardprop}
		The simultaneous communication complexity of $f$ is at least $\binom{m}{\frac{m}{2}}$.
	\end{proposition}
	\begin{proof}
		Let $INDEX_k$ be the index problem with $k=\binom{m}{\frac{m}{2}}$ bits. We show that simultaneously computing the allocation of, say, Alice is at least as hard as simultaneously computing $INDEX_k$.
		
		For the reduction, we define $int:\mathcal X \to [k]$ as a function that maps each subset in the collection $\mathcal X$ to its index. For example, $int(X_1)=1$.
	
		Given a array in $arr\in\{0,1\}^{k}$,
		we construct the 
		following valuation for Charlie:
		$$
		\quad v_C(T)=\begin{cases}
			0 \quad & |T|<\frac{m}{2}, \\
			arr[\ell] \quad & |T|=\frac{m}{2} \text{ and } int(T)=\ell, \\
			1  \quad & |T|>\frac{m}{2}.
		\end{cases}
		$$  
		In words, each coordinate $l$ in the array $arr$ is linked to a subset $T$ of $M$ of size $\frac{m}{2}$, and the value of $arr[l]$ determines the value of the subset $T$ for Charlie.  

		Given an index in $j\in \{0,\ldots,k-1\}$,  let $v_B$ be an arbitrary valuation such that  $v_B(\{b\})=j$. 
		By construction, $INDEX(arr,j)=1$ if and only if $f$ allocates item $a$ to Alice. By Theorem \ref{index-hard-theorem}, the simultaneous communication complexity of $INDEX(k)$ is $\Omega(k)$, which completes the proof. See Figure \ref{three-player-fig} for an illustration.    
	\end{proof}
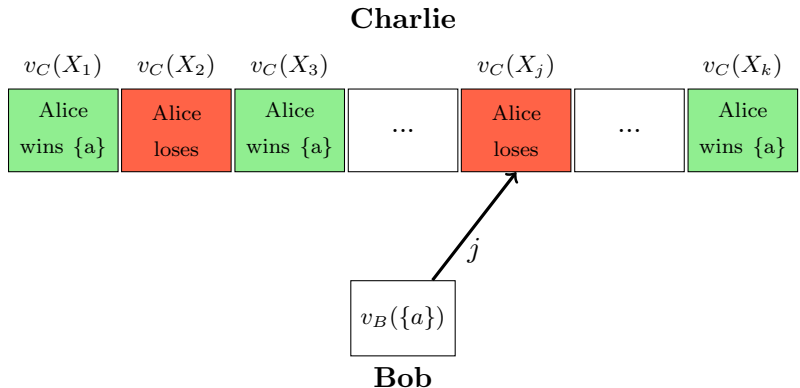
\begin{figure}[H]
	\centering 
	\setlength{\belowcaptionskip}{-1pt}
	\centering
	\caption{An illustration of the reduction of $INDEX$ to the allocation of Alice. The figure depicts a situation where Bob's value for item $a$, $j$, points at a bundle $X$ such that $v_C(X)<1$, so Alice gets the empty bundle.}
	\label{three-player-fig}
	\begin{tikzpicture}
		
		[every node/.style={anchor=base,
			text width=1em,
			font=\footnotesize,
			align=center
		}]
		
		\node[shape=rectangle,minimum size=0.5cm] (Charlie) at  (0,0) {\normalsize
			\textbf{Charlie}};

		\matrix (mat) [matrix of nodes,
		row sep=0pt, 
		column sep=1pt,
		nodes={draw,text width=1.2cm, minimum 
			height=1.1cm, 
			anchor=center,
			align=center,
		}]
		at (0,-1.5)
		{
			|[fill=LightGreen]| {\scriptsize Alice wins \{a\}}  &
			|[fill=Tomato]|	\scriptsize Alice loses
			& |[fill=LightGreen]| \scriptsize Alice wins \{a\} & ... & 
			|[fill=Tomato]| \scriptsize Alice loses &
			... &
			|[fill=LightGreen]| \scriptsize Alice wins \{a\}
			\\
		};

		
		\foreach \i [count=\xi from 1]
		in  {1,2,3}
		{
			\node also [label=above: \footnotesize $v_C(X_{\xi})$] (mat-1-\i) {}; 
		}
		\node also [label=above:\footnotesize $v_C(X_{j})$] (mat-1-5) {}; 	
		\node also [label=above: \footnotesize $v_C(X_{k})$] (mat-1-7) {};

		\node[shape=rectangle,draw,minimum size=1cm,label=below:{\normalsize \textbf{Bob}}] (Bob) at  (0,-4) {\footnotesize
			$v_B(\{a\})$};
		
		%
		
		\draw [->,very thick] (Bob) edge  node[below] {$j$} (mat-1-5.south);
		%
	\end{tikzpicture}
\end{figure}
We now leverage the simultaneous hardness of $f$ to show hardness of implementation in dominant strategies. 
Consider a dominant strategy implementation of $f$ with communication complexity $c <  {\binom{m}{\frac{m}{2}}}/3$ with the dominant strategies $(\mathcal{S}_A,\mathcal{S}_B,\mathcal{S}_C)$.
By Lemma \ref{minimal-lemma}, we can assume that the implementation is minimal without loss of generality.
Thus, there exists a player, without loss of generality Alice, that has at least two valuations with different messages at the root vertex of the protocol, which we denote with $r$.  

By Proposition \ref{fahardprop}, it is impossible to determine whether Alice wins $a$ or not after the first round of communication (otherwise, repeat the protocol three times, appropriately switching the roles of the players, and compute $f$ simultaneously with $3c$ bits). 
Thus, there are two pairs of valuations $(v_B^0,v_C^0)$ and $(v_B^1,v_C^1)$ such that the following assertions hold:
\begin{enumerate}
	\item Bob's dominant strategy dictates  the same message $z_B$ in the first round for both $v_B^0$ and $v_B^1$. 
	\item Bob's dominant strategy dictates  the same message $z_C$ in the first round for both $v_C^0$ and $v_C^1$.
	\item For every $v_A\in V_A$, $f(v_A,v_B^0,v_C^0)$ assigns the empty bundle to Alice and $f(v_A,v_B^1,v_C^1)$ assigns the item $\{a\}$ to Alice. Thus, every message of Alice at the root vertex leads to a subtree that has a leaf labeled with the allocation of $\{a\}$ for Alice and another leaf labeled with the allocation of the empty bundle for her. \label{prop-3}
\end{enumerate}
We now explain properties of the induced tree of Alice given the messages $z_B,z_C$, which we will use to show that Alice's supposedly dominant strategy fails her. 

Let $z_A,z_A'$ be two possible messages of Alice in the first round, that she sends for the valuations $v_A$ and $v_A'$.  
By the above, there exists a leaf at the subtree of the messages $(z_A',z_B,z_C)$ labeled with the bundle $\{a\}$ for Alice (which is every leaf that the strategies $\mathcal{S}_B(v_B)$ and $\mathcal{S}_C(v_C)$ lead to), and the same holds for the subtree that $(z_A'',z_B,z_C)$ leads to. Due to the same reasons, both subtrees also have a leaf labeled with the empty bundle for Alice.  Figure \ref{fig-combined-protocol-tree} provides an explanation for that.

\begin{figure}[H]
	\centering
	\setlength{\belowcaptionskip}{3pt}
	\setlength{\abovecaptionskip}{3pt}
	\caption{Illustration of part of the induced tree of Alice at the root vertex, given the messages $z_B$ and $z_C$ of Bob and Charlie respectively. Given the dominant strategies $\mathcal{S}_A,\mathcal{S}_B,\mathcal{S}_C$, where the valuations of Alice, Bob and Charlie are  $(v_A,v_B^0,v_C^0)$, we have that the execution of the protocol ends at the leftmost leaf. The same holds for the remaining leaves in the figure. 
	}	
	\label{fig-combined-protocol-tree}
	\begin{tikzpicture}
		\node[shape=circle,draw=black,minimum size=0.8cm] (u) at (1.5,1.5) {$r$};
		\node[shape=circle,draw=black,minimum size=0.8cm] (v) at (-0.5,0) {};
		\node[shape=circle,fill=Tomato,draw=black,minimum size=1.4cm,align=center,label=below:{\tiny$(v_A,v_B^0,v_C^0)$}] (l) at (-1.5,-1.5) {\tiny A loses \\ \tiny $\{a\}$};
		\node[shape=circle,draw=black,minimum size=0.8cm] (v') at (3.5,0) {}; 
		\node[shape=circle,draw=black,fill=LightGreen,minimum size=1.4cm,align=center,label=below:{\tiny$(v_A',v_B^1,v_C^1)$}] (l') at (4.5,-1.5) {\tiny A wins \\ \tiny $\{a\}$};
			\node[shape=circle,fill=Tomato,draw=black,minimum size=1.4cm,align=center,label=below:{\tiny$(v_A',v_B^0,v_C^0)$}] (n') at (2.5,-1.5) {\tiny A loses \\ \tiny $\{a\}$};
		\node[shape=circle,draw=black,,fill=LightGreen,minimum size=1.4cm,align=center,label=below:{\tiny$(v_A,v_B^1,v_C^1)$}] (n) at (0.5,-1.5) {\tiny A wins \\ \tiny $\{a\}$};

		\draw [->] (u) edge  node[sloped, above] {\small$z_A$} (v);
		\draw [->] (u) edge  node[sloped, above] {\small $z_A'$} (v');
		\draw [->] (v) edge[dotted]  node[sloped, above] {}  (l);
		\draw [->] (v') edge[dotted]  node[sloped, above] {} (l');
		\draw [->] (v) edge[dotted]  node[sloped, above] {} (n);
		\draw [->] (v') edge[dotted]  node[sloped, above] {} (n');
	\end{tikzpicture}
\end{figure}
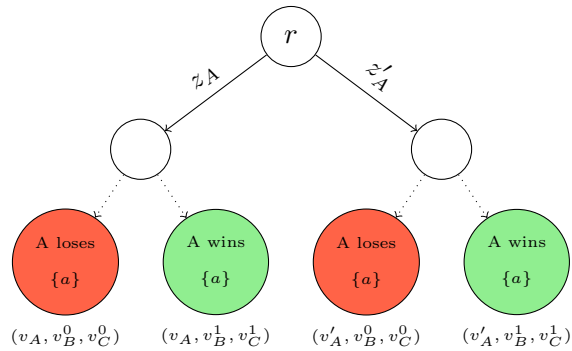

By the following proposition, we get that all the leaves labeled with the allocation $\{a\}$ for Alice have the same payment for her, and that the same holds for every leaf labeled with $\varnothing$ for Alice. 

\begin{proposition}\label{prop-known-prices}
	Consider the induced tree of Alice at the root vertex $r$ given the messages $z_B,z_C$ of Bob and Charlie. 	
	Let $S\subseteq M$ be a an allocation of Alice that appears in two different subtrees.
	\emph{Then}, all the leaves in this induced tree that are labeled with $S$ have the same payment for Alice.  	 
\end{proposition}
The proof of Proposition \ref{prop-known-prices} is identical to the proof of Lemma \ref{lemma-known-prices}. The only difference between the two is that Proposition \ref{prop-known-prices} refers to the more restricted setting of combinatorial auctions, so two allocations where player $i$ wins the same bundle are identical for him.  

By the above, the subtrees of both $z_A$ and $z_A'$ in the induced tree given the messages $z_B,z_C$ have a leaf labeled with $\{a\}$ and a leaf labeled with $\varnothing$. By Proposition \ref{prop-known-prices}, we have that all the leaves in the induced tree of $z_B,z_C$ that are labeled with $\{a\}$ have the same price $P_a$. Similarly, all leaves labeled with the empty bundle have the same price, denoted by $P_\varnothing$.

We can now use the those properties of the induced tree to show that there exists a valuation of Alice with no dominant strategy. 
Let $v_A^\ast\in V_A$ be a valuation such that $v_A^\ast(\{a\})\neq P_{a}-P_{0}$. We assume that $v_A^\ast(a)-P_{a}>-P_{0}$. The proof for the complementary case $v_A^\ast(a)<P_{a}-P_{0}$ is analogous. Assume that the dominant strategy $S_A$ dictates the message $z_A^\ast$ at vertex $r$. Let $z_A^{\ast\ast}\neq z_A^{\ast}$ be a different message that Alice can send at the root. 

Observe the following strategy profiles $\mathcal{S}_{B}'$ and $\mathcal{S}_C$ of Bob and Charlie: For every valuation profile $(v_B',v_C')$, Bob and Charlie send the messages $z_B,z_C$ (respectively) at the root vertex $r$. At the subtree that the messages $(z_A^\ast,z_B,z_C)$ leads to, Bob sends the messages that the dominant strategy $\mathcal{S}_B$ dictates for $v_B^0$ and Charlie sends the messages that $\mathcal{S}_C$ dictates for $v_C^0$.  At every other subtree, Bob and Charlie choose the  actions specified by the dominant strategies $S_B$ and $S_C$ for the valuations $v_B^1$ and $v_C^1$. By property \ref{prop-3}, for every $v_B',v_C'\in V_B\times V_C$,  $(\mathcal{S}_A(v_A),\mathcal{S}_{B}'(v_B'),\mathcal{S}_{C}'(v_C'))$ leads to the outcome $(\varnothing,P_0)$ and any strategy that dictates a message other than $z_A$ leads to $(a,P_a)$ given the actions dictated by the strategies $\mathcal{S}_{B}'(v_B'),\mathcal{S}_{C}'(v_C')$. Recall that by assumption $v_A^\ast(a)-P_{a}>-P_{0}$, so $\mathcal{S}_A(v_A)$ is not dominant for Alice, and we reach a contradiction. 

 


\end{proof}

\section{The Structure of Mechanisms for General Valuations}\label{sec-characterization} 
In this section we study the structure of dominant strategy mechanisms for general (monotone) valuations.  We will show that each such mechanism is semi-simultaneous, as we define later. 
Recall that every communication protocol can be described as a tree: In each vertex, each player has several messages that he can send. The leaves are labeled with the outcome (which are allocation and payment, in the case of combinatorial auctions). 


In what follows, we fix a dominant strategy mechanism $\mathcal M$, the dominant strategies $\mathcal{S}_1,\ldots,\mathcal{S}_n$ of the players and their valuation sets $V_1\times\cdots\times V_n$.
By Lemma \ref{minimal-lemma}, we can assume that the mechanism $\mathcal{M}$ is minimal with respect to the dominant strategies and the valuation sets without loss of generality.
For simplicity, we assume that all leaves in the induced tree are only labeled with the allocation and payment of player $i$.
\begin{definition}
	Fix some player $i$, vertex $u$, and messages of the other players $z^u_{-i}$. The \emph{minimal price} of bundle $S$ at vertex $u$ is $p_S$ if $p_S$ is the minimal price of all leaves in the induced tree of player $i$ given $z^u_{-i}$ that are labeled with a bundle that contains $S$.
\end{definition}
\begin{definition}
	Fix some player $i$, vertex $u$, and messages of the other players $z^u_{-i}$. Suppose that all players except player $i$ send the messages $z^u_{-i}=(z^{u}_1,\ldots,z^{u}_{i-1},z^{u}_{i+1},\ldots,z^{u}_{n})$. A bundle $S$ is \emph{decisive at price $p$}  at the induced tree of player $i$ at vertex $u$ given $z^{u}_{-i}$ if there is a strategy of player $i$ at vertex $u$ that guarantees that for every strategy profile of the other players that is consistent with $z^{u}_{-i}$, the protocol reaches a leaf labeled with a bundle that contains $S$ and price at most $p$.     
\end{definition}
A \emph{trivial} message of player $i$ is a message such that the dominant strategy of player $i$ is always to send this message, no matter what his valuation is.
\begin{definition}
	A  dominant strategy mechanism is called \emph{semi-simultaneous} if the following holds for each player $i$ and vertex $u$ in which player $i$ sends his first non-trivial message, and for every possible set of messages of the other players $z^u_{-i}$:
	\begin{itemize}
		\item[]Denote by $T$ the set of subtrees in the induced tree of player $i$ at vertex $u$ given $z^u_{-i}$. Then, there is at most one subtree $t^\ast\in T$ (the \emph{special subtree}) such that every leaf $l$ which does not belong to $t^\ast$ is labeled with a bundle that is decisive at its minimal price.
	\end{itemize}
\end{definition}

Recall that an ascending auction on the bundle of all items is one example for a semi-simultaneous mechanism. In an ascending auction, the special subtree for each player at each level of the protocol tree is the subtree where he agrees to pay the given price for the bundle of all items. 
Another example is a mechanism in which player $1$ is offered to buy in the first round item $a$ at price $1$. If he declines, he is offered some arbitrary prices for bundles that do not contain $a$ (the prices are a function of the valuations of the other players and they are sent  only after the first round), as well as the option to buy item $a$ at price $1$.

\begin{theorem}\label{characterizationtheorem}
	Every minimal dominant strategy mechanism $\mathcal M$ for general valuations  is semi-simultaneous.
\end{theorem}
We call the message $z^*$ of player $i$ that leads to the subtree $t^*$ the \emph{special message} of player $i$ at vertex $u$. Note that this special message $z^*$ might depend on the messages $z_{-i}^u$ that the other players send simultaneously. In other words, at the same vertex $u$, it might be the case that sometimes a message $z'$ is a special message and sometimes a different message $z''$ is a special message. The following is an important corollary of the theorem:

\begin{corollary}\label{cor-commit-profit}
	Fix some vertex $u$, player $i$ with valuation $v$, and messages of the other players $z^u_{-i}$. Consider the induced tree of player $i$ at vertex $u$ given the messages $z^u_{-i}$. Suppose that according to the dominant strategy of player $i$, $\mathcal{S}_i(v_i)$, he sends his first non-trivial message at vertex $u$ and that this message is not special given $z^{u}_{-i}$. Then, he is guaranteed to be allocated a bundle that will give him profit of at least $\max_{S\in \mathbb{S}}v(S)-p_S$, where $\mathbb S$ is the set of all bundles $S$ that are decisive at their minimal price $p_S$ in the induced tree of vertex $u$ given the messages $z_{-i}^u$.
\end{corollary}
\begin{proof}
	Let $T$ be some bundle in $\arg\max_{S\in \mathbb S}v(S)-p_S$.
	Note that $T$ is decisive at price $p_T$, thus there is some strategy $\mathcal{S}'_i$ that guarantees it or a bundle that contains it at its minimal price $p_T$. For every strategy profile $\mathcal{S}_{-i}$, the dominant strategy $\mathcal{S}_i$ cannot lead to an outcome with smaller profit, because otherwise  player $i$ should play according to $\mathcal{S}'_i$ instead of $\mathcal{S}_i$. 
\end{proof}



\subsection{Proof of Theorem \ref{characterizationtheorem}}

We say that a bundle $S$ \emph{appears} in some subtree $t$ with price $p_S$ if there is a leaf in the subtree $t$ that is labeled with $(S,p_S)$.
The following claims will be useful:
\begin{claim}\label{containment-claim}
	Fix some player $i$, vertex $u$, and messages of the other players $z^u_{-i}$. Consider the induced subtree of player $i$ at vertex $u$ given $z^u_{-i}$. If the bundle $S$ appears in some subtree $t$ with price $p_S$ and a bundle $S'$ such that $S\subseteq S'$ appears in a different subtree $t'$ with price $p_{S'}$, then $p_{S'}\ge p_S$.
\end{claim}
The proof is very similar to that of Claim \ref{lemma-known-prices}. We write it for completeness.
\begin{proof}
	Denote with $\ell$ the leaf labeled with $(S,p_S)$ at the subtree $t$ and with $\ell'$ the leaf labeled with $(S',p_{S'})$ at the subtree $t'$. By the minimality of the mechanism, every leaf in the protocol has 
	valuations such that the dominant strategies  $(\mathcal{S}_1(v_1),\ldots,\mathcal{S}_n(v_n))$ reach this leaf. Thus,	 
	there exist valuations $v,v'\in V_i,v_{-i},v_{-i}'\in V_{-i}$ such that:
	$$
	(\mathcal{S}_i(v),\mathcal{S}_{-i}(v_{-i}))\to \ell,\quad
	(\mathcal{S}_i(v'),\mathcal{S}_{-i}(v_{-i}'))\to \ell'
	$$
	Observe the following strategy profile $\mathcal{S}_{-i}''$: For every valuation $v_{-i}''$ of the players in $N\setminus\{i\}$, choose the actions specified by both $\mathcal{S}_{-i}(v_{-i})$ and $\mathcal{S}_{-i}(v_{-i}')$ until vertex $u$ (including vertex $u$). Afterwards, at the subtree $t$, choose the actions specified by $\mathcal{S}_{-i}(v_{-i})$, and at the subtree $t'$, choose the actions specified by $\mathcal{S}_{-i}(v_{-i}')$. Note that:
	$$
	(\mathcal{S}_i(v),\mathcal{S}_{-i}''(v_{-i}''))\to \ell,\quad
	(\mathcal{S}_i(v'),\mathcal{S}_{-i}''(v_{-i}''))\to \ell'
	$$
	where the profit of player $i$ given the valuation $v$ at the leaf $\ell$ has to be larger than her profit at leaf $\ell'$, since $\mathcal{S}_i(v)$ is a dominant strategy for her. Thus:  
	$$
	v(S)-p_S \ge v(S') - p_{S'} 
	$$
	Recall that $S \subseteq S'$, so $v(S')\ge v(S)$, which implies that $p_{S'}\ge p_S$, as needed. 
\end{proof}

\begin{claim}\label{claim-subtree-maximizes-bundle}
	Fix some player $i$ with valuation $v$, vertex $u$, and messages of the other players $z^u_{-i}$. Consider two subtrees $t,t'$ in the induced tree of player $i$ at vertex $u$ given $z^u_{-i}$. Suppose that the subtree $t$ contains a leaf $l$ that is labeled $(S,p_S)$ and that for each leaf $l'$ of $t'$ with label $(S',p_{S'})$ we have that $v(S)-p_S> v(S')-p_{S'}$. Then, no dominant strategy of $v$ dictates a message that leads to a subtree $t'$.
\end{claim}
\begin{proof}
	A strategy that leads to a subtree $t'$ is not dominant for $v$: if the player sends a message that leads to the subtree $t$ then there are strategies of the players will lead to the leaf $l$ and will imply strictly higher profit than that of any leaf in the subtree $t'$.
\end{proof}


\begin{claim}\label{claim-undecisive-unique-subtree}
	Fix player $i$  that sends his first non-trivial message at vertex $u$. 
	Let $z^{u}_{-i}$ be the messages of the other players.
	Consider the induced tree of player $i$ at vertex $u$ given $z^u_{-i}$. Let $S$ be a bundle that is not decisive at its minimal price $p_S$ in the induced tree. Let $l_1$ and $l_2$ be two leaves labeled $(S_1,p_S), (S_2,p_S)$, respectively, where $S\subseteq S_1,S_2$. Then, $l_1$ and $l_2$ are in the same subtree of $u$. 
	
\end{claim}

\begin{figure} [H] 
	\centering
	\caption{This figure is an illustration of the proof of Claim \ref{claim-undecisive-unique-subtree}. Below is the tree that $z_{-i}^u$ induces for player $i$ at vertex $u$ if $l_1$ and $l_2$  belong to different subtrees, which we denote with $t_1$ and $t_2$. Note that there are possibly more subtrees in the tree besides $t_1,t_2$ that are depicted.
	} 
	
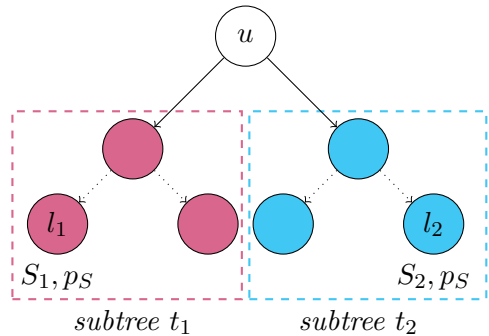
\label{charac-tikz}
	\begin{tikzpicture} 
		\node[shape=circle,draw=black,minimum size=0.8cm] (r) at (1.5,1.5) {$u$};
		\node[shape=circle,draw=black,minimum size=0.8cm,fill=purple!60] (v) at (0,0) {};
		\node[shape=circle,draw=black,minimum size=0.8cm,fill=purple!60] (l) at (-1,-1) {$l_1$};
		\node[] (bundle) at (-1,-1.7) {$S_1,p_S$};
		\node[shape=circle,draw=black,minimum size=0.8cm,fill=cyan!60] (v') at (3,0) {}; 
		\node[shape=circle,draw=black,minimum size=0.8cm,,fill=cyan!60] (l') at (4,-1) {$l_2$};
		\node[] (bundle') at (4,-1.7) {$S_2,p_S$}; 
		\node[shape=circle,draw=black,minimum size=0.8cm,fill=cyan!60] (n') at (2,-1) {}; 
		\node[shape=circle,draw=black,minimum size=0.8cm,fill=purple!60] (n) at (1,-1) {}; 
		
		\draw[thick,purple!60,dashed] (-1.6,0.5) rectangle (1.45,-2) {};
		\node[font=\itshape,black] at (0,-2.3) {subtree $t_1$};
		
		\draw[thick,cyan!60,anchor=mid west,dashed] (1.55,0.5) rectangle (4.7,-2) {};
		\node[font=\itshape,black] at (3,-2.3) {subtree $t_2$};

		\draw [->] (r) edge  node[sloped, above] {} (v);
		\draw [->] (r) edge  node[sloped, above] {} (v');
		\draw [->] (v) edge[dotted]  node[sloped, above] {}  (l);
		\draw [->] (v') edge[dotted]  node[sloped, above] {} (l');
		\draw [->] (v) edge[dotted]  node[sloped, above] {} (n);
		\draw [->] (v') edge[dotted]  node[sloped, above] {} (n');
		
	\end{tikzpicture}
	
\end{figure}

\begin{proof}
	For an illustration of the induced tree if $l_1$ and $l_2$ are \emph{not} in the same subtree, see Figure \ref{charac-tikz}. 
	
	Let $p_{max}$ and $p_{min}$ denote the maximum and the minimum prices in the label of every leaf  in the induced tree (note that a maximum and a minimum exist since the number of leaves is finite because the communication protocol is finite). Let $v$ be the additive valuation such that for every bundle $X$, $v(X)=2\cdot |X\cap S|\cdot p$,
	where $p=\max\{1,p_{max}-p_{min}\}$, i.e., each item has value of $2p$ according to $v$.  
	
	At vertex $u$, the dominant strategy of the valuation $v$ must be to send a message that leads to some subtree $t$ that contains some leaf labeled with $S$ or a superset of it at price $p_S$ (a \emph{good} leaf). To see this, note that any bundle $S'$ that is a superset of $S$ but has higher price is less profitable, since the marginal value of an item outside $S$ is $0$. Also, any bundle $S''$ that contains only a a strict subset of the items is always less profitable than $S$ at price $p_S$, even if the price of $S''$ is $p_{min}$: the marginal value of adding the items $S\setminus S''$ to $S''$ is at least $2p$ whereas the price difference is at most $p_{max}-p_{min}$ (if $p_S=p_{max}$ and $p_{S''}=p_{min}$). 
	Thus, Claim \ref{claim-subtree-maximizes-bundle} guarantees that the player sends a message that leads to some subtree $t$ that contains a good leaf.
	%
	%
	%
	%
	%
	
	We now show that there exist strategies of the other players such that sending the message that leads to subtree $t$ is not dominant. Since $S$ is not decisive at price $p_S$, there are strategies of the other players such that the mechanism reaches a leaf in which player $i$'s profit is strictly less than $v(S)-p_S$. This strategy is not dominant because player $i$ is better off sending a message that leads to the subtree $t'\neq t$ that contains either leaf $l_1$ or $l_2$, if the players play in a way that leads to it. Therefore, we reach a contradiction, so we have that $l_1,l_2$ belong in the same subtree. 
\end{proof}

\begin{claim} \label{claim-undecisive-bundle-unique-subtree}
	Fix some player $i$ that sends his first non-trivial message at  vertex $u$. 
	Let the messages of the other players be $z^u_{-i}$. Consider the induced tree of player $i$ at vertex $u$ given $z^u_{-i}$. Let $S_1$, $S_2$ be two bundles that are not decisive at their minimal prices (denoted $p_1,p_2$, respectively). Let $l_1$ be a leaf labeled $(S_1,p_1)$ and $l_2$ be a leaf labeled $(S_2,p_{2})$. Then, the leaves $l_1,l_2$ are in the same subtree.
\end{claim}

\begin{proof}
	Let $p_{min}$ be the minimal price of of any bundle in the induced tree of vertex $u$ given $z_{-i}^u$. Let $q=\min\{p_{min},0\}$.  
	
	Let $p_1'>p_1$ be the minimal price (that is not $p_1$) that appears in a label $(S,p)$ of a leaf in the induced tree where $S_1\subseteq S$ (if $p_1'$ is undefined, set it to some value strictly bigger than $p_1$). Define $p_2'$ similarly with respect to $S_2$ and $p_2$. 
	We prove the claim by showing that the following (monotone) valuation has no dominant strategy: 
	$$
	v(X)=\begin{cases}
	\max\{\frac {p_1'+p_1} 2,\frac {p_{2}'+p_2} 2\}-q \quad &S_1\cup S_2 \subseteq X, \\
	\frac {p_{1}'+p_1} 2-q \quad &S_1 \subseteq X  \text{ and } S_2 \not\subseteq X, \\
	\frac {p_2'+p_2} 2 -q \quad &S_1 \not \subseteq X \text{ and } S_2  \subseteq X, \\
	0 \quad &\text{otherwise.}
	\end{cases}
	$$

	Assume towards a contradiction that $l_1$ and $l_2$ are not in the same subtree, and let
	$t_1$ be the subtree of $l_1$ and $t_2$ be the subtree of $l_2$ (See Figure \ref{my-first-tikz-xos}). Observe that the profit of $v$ from the leaf $l_1$ is strictly larger than $-q$:
	$v(S_1)-p_1 \ge \frac{p_1+p_1'}{2}-p_1-q > -q$.
	Similarly, the profit from the leaf $l_2$ is also strictly larger than $-q$.
	
	Observe that the only leaves with label $(T,p_T)$ such that $v(T)-p_T> -q$ belong in the subtrees $t_1$ and $t_2$: by Claim \ref{containment-claim}  and Claim \ref{claim-undecisive-unique-subtree}, all leaves that are labeled with bundles that contain $S_1$ and do not appear in $t_1$ have price strictly larger than $p_{1}$ (so it is at least $p_1'$), and thus have a profit of at most $-q$ for $v$. Similarly, all leaves that are labeled with bundles that contain $S_2$ and do not appear in $t_2$ have price at least $p_2'$, and thus have a profit of at most $-q$ for $v$. $v$ gives value $0$ for all other bundles, so for every bundle $S$ that does not contain $S_1,S_2$ we have that the profit is at most $0-p_{min}\le -q$ . Thus, by Claim \ref{claim-subtree-maximizes-bundle}, the dominant strategy of $v$ dictates a message that leads either to $t_1$ or to $t_2$.
	
	
	Assume without loss of generality that it dictates a message that leads to $t_1$.
	Since $S_1$ is not decisive at price $p_1$, if player $i$ sends a message that leads to the subtree $t_1$, there are strategies of the other players in which he will not win $S_1$ or a bundle that contains it at price $p_1$. Thus, by definition if he wins $S_1$ the price is at least $p_1'$, so his profit is at most $-q$.  
	Also, by Claim \ref{claim-undecisive-unique-subtree}, every leaf in $t_1$ that is labeled with a bundle that contains $S_2$ has price at least $p'_2$. Therefore, if player $i$ wins some bundle that contains $S_2$ he pays at least $p'_2$ and has at most $-q$ profit as well. All other bundles have $0$ value for player $i$, so his profit is at most $-q$. 
	However, if the player sent a message that leads to the subtree $t_2$, there are strategies of the players that would make him win the bundle $S_2$ at price $p_2$, so he would have a profit that is strictly higher than $-q$. Thus, given this specific strategy profile of the players in $N\setminus\{i\}$, sending a message that leads to the subtree $t$ is not a dominant strategy. 
	
	Similarly, sending a message that leads to the subtree $t_2$ is not a dominant strategy. Overall, we conclude that player $i$ with valuation $v$ has no dominant strategy, a contradiction.
\end{proof}
	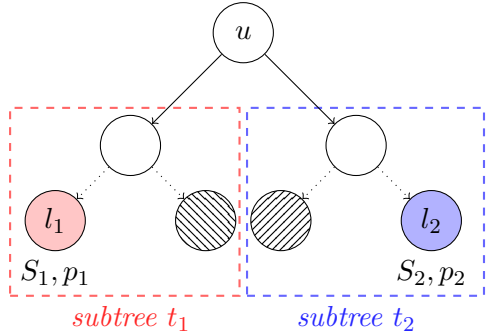
\begin{figure} [h] 
	\centering
	
	\begin{tikzpicture}
	\node[shape=circle,draw=black,minimum size=0.8cm] (r) at (1.5,1.5) {$u$};
	\node[shape=circle,draw=black,minimum size=0.8cm] (v) at (0,0) {};
	\node[shape=circle,draw=black,minimum size=0.8cm,fill=pink!90] () at (-1,-1) {$l_1$};
	\node[] (bundle) at (-1,-1.7) {$S_1,p_1$};
	\node[shape=circle,draw=black,minimum size=0.8cm] (v') at (3,0) {}; 
	\node[shape=circle,draw=black,minimum size=0.8cm,fill=blue!30] (l') at (4,-1) {$l_2$};
	\node[] (bundle') at (4,-1.7) {$S_2,p_{2}$}; 
	\node[shape=circle,draw=black,minimum size=0.8cm,pattern=north east lines] (n') at (2,-1) {}; 
	\node[shape=circle,draw=black,minimum size=0.8cm,pattern=north west lines] (n) at (1,-1) {}; 
	
	\draw[thick,red!60,dashed] (-1.6,0.5) rectangle (1.45,-2) {};
	\node[black,font=\itshape,text=red!85] at (0,-2.3) {subtree $t_1$};
	
	\draw[thick,blue!60,anchor=mid west,dashed] (1.55,0.5) rectangle (4.7,-2) {};
	\node[black,font=\itshape,text=blue!85] at (3,-2.3) {subtree $t_2$};

	\draw [->] (r) edge  node[sloped, above] {} (v);
	\draw [->] (r) edge  node[sloped, above] {} (v');
	\draw [->] (v) edge[dotted]  node[sloped, above] {}  (l);
	\draw [->] (v') edge[dotted]  node[sloped, above] {} (l');
	\draw [->] (v) edge[dotted]  node[sloped, above] {} (n);
	\draw [->] (v') edge[dotted]  node[sloped, above] {} (n');
	
	\end{tikzpicture}
	\caption{An illustration for the proof of Claim \ref{claim-undecisive-bundle-unique-subtree}. It describes the tree that the message of the other players $z_{-i}^u$ induces for player $i$ at vertex $u$ if $l_1$ and $l_2$ do not belong to the same subtree. Note that there are possibly other subtrees in the tree besides $t_1,t_2$.
		The colored leaves $l_1,l_2$ are the most profitable leaves for the valuation $v$. The leaves with the lines pattern signify less profitable leaves, which necessarily exist in the subtrees $t_1,t_2$ due to the fact that neither $S_1,S_2$ are decisive at their minimal prices.}
	\label{my-first-tikz-xos}
	
\end{figure}

By Claim \ref{claim-undecisive-bundle-unique-subtree}, we have that in every possible induced tree of player $i$ given $z_{-i}^u$, all leaves labeled with some bundle $S$ at its minimal price $p_S$, where $S$ is not decisive at $p_S$, satisfy that they all belong to the same subtree.
We call this subtree the special subtree. If all bundles are decisive at their minimal prices, Theorem \ref{characterizationtheorem} holds trivially.

To finish the proof of Theorem \ref{characterizationtheorem}, consider a leaf $l$ that does not belong to the special subtree.
Let the label of $l$ be $(T,p_T)$. Our goal is to show that $T$ is decisive at its minimal price. We prove it by showing that    
$p_T$ is the minimal price of $T$ and hence, by Claim \ref{claim-undecisive-bundle-unique-subtree}, $T$ is decisive at price $p_T$ (because $l$ is outside the special subtree).

First, note that the special subtree does not contain a leaf with label $(T',p'_T)$ where $T\subseteq T'$ and $p'_T < p_T$, by Claim \ref{containment-claim}. Thus, if $p^{min}_T<p_T$ is the minimal price of $T$, it appears on some leaf not in the special subtree. Thus, by Claim \ref{claim-undecisive-unique-subtree}, $T$ is decisive at price $p^{min}_T$. We have that no dominant strategy of player $i$ can lead to $l$, since it always has strictly lower profit. $l$ is thus useless and should not appear in the protocol of a minimal mechanism. This concludes the proof of Theorem \ref{characterizationtheorem}.

\end{document}